\documentclass[openright,twoside]{iitbthesis}

\usepackage{graphicx}
\usepackage{amssymb}
\usepackage{amsmath}
\usepackage{amsfonts}
\usepackage{amsthm}
\usepackage{verbatim}
\usepackage{url}
\usepackage{hyperref}
\usepackage{breakurl}
\usepackage[numbers]{natbib}
\usepackage{wrapfig}
\usepackage{algorithmic}
\usepackage{hyphenat}
\usepackage{appendix}
\usepackage{color}
\usepackage{pdfpages}

\newtheorem{theorem}{Theorem}[chapter]
\newtheorem{lemma}[theorem]{Lemma}
\newtheorem{corollary}[theorem]{Corollary}
\newtheorem{observation}{Observation}[chapter]

\newtheorem{algorithm}{Algorithm}[chapter]
\newtheorem{definition}{Definition}[chapter]

\begin{document}
%

\clearpage\pagenumbering{roman}

\title{Wireless Coverage Area Computation and Optimization}
\author{Prateek R.~Kapadia}
\date{2015}

\rollnum{05429702}

\iitbdegree{Doctor of Philosophy}

\thesis

\department{DEPARTMENT OF COMPUTER SCIENCE \& ENGINEERING}

\setguide{Prof.~Om Damani}
\setcoguide{Prof.~Animesh Kumar}


\maketitle

\begin{dedication}
To my family, for giving me selfless support for my selfish endeavor.
\end{dedication}

  \begin{abstract}
A wireless network's design must include the optimization of the area of
coverage of its wireless transmitters - mobile and base stations in cellular 
networks, wireless access points in WLANs, or nodes on a transmit schedule in a 
wireless ad-hoc network. Typically, the coverage
optimization for the common channels is ``solved'' by spatial multiplexing,
i.e. keeping the access networks far apart. However, with increasing densities 
of wireless network deployments (including the {\em Internet-of-Things}) and 
paucity of spectrum, and new developments like whitespace devices and 
self-organizing, cognitive networks, there is a need to manage interference and 
optimize coverage by efficient algorithms that correctly set the transmit 
powers to ensure that transmissions only use the power necessary.

In this work we study methods for computing and optimizing interference-limited
coverage maps of a set of transmitters. We progress successively through 
increasingly realistic network scenarios. We begin with a disk model with a 
fixed set of transmitters and present an optimal algorithm for computing the 
coverage map. We then enhance the model to include updates to the network, in 
the form of addition or deletion of one transmitter. In this dynamic setting, 
we present an optimal algorithm to maintain updates to the coverage map. We 
then move to a more realistic interference model - the SINR model. For the SINR 
model we first show geometric bases for coverage maps. We then present a method 
to approximate the measure of the coverage area. Finally, we present an 
algorithm that uses this measure to optimize the coverage area with minimum 
total transmit power.
\end{abstract}


\tableofcontents
\listoftables
\listoffigures

\cleardoublepage\pagenumbering{arabic} 

\chapter{Introduction}
A wireless network is a communication network in which network nodes 
communicate with each other over wireless media. Wireless communication between 
two nodes involves the transmission of radio signals from one node (the 
transmitter) which are then decoded by the intended node (the receiver). 
Successful decoding of a radio signal requires sufficient receive signal 
energy at the receiver. Radio signals from simultaneous transmissions also 
typically interfere with each other, and receivers are unable to decode 
signals that are disrupted by interfering transmissions.

The set of locations at which potential receivers in the network are able to 
decode transmissions intended for them is the {\em coverage map} of the 
wireless network. For example, Figure \ref{figure:CoverageIntro} shows a 
typical wireless network with a fixed set of transmitters. The shaded area in 
this figure shows the network's coverage map.

\begin{figure}
\centering
\includegraphics[width=3in]{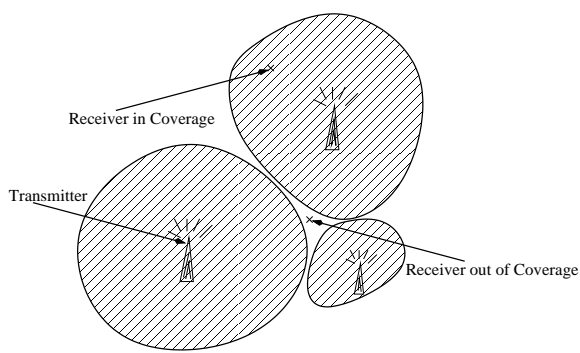}
\caption{Example of Coverage Map with 3 Transmitters}
\label{figure:CoverageIntro}
\end{figure}
\section{Motivation and goals}
A key service quality indicator of a wireless network is the size or measure of 
its coverage map, i.e.~the {\em coverage area}. The maintenance of an 
adequately large coverage area is achieved by coverage area maximization. The 
objective function for this optimization is the coverage area, and the 
optimization variables and constraints are drawn from the network topology and 
energy profile, such as the transmitter and receiver location constraints, 
channels available for transmission, receive sensitivity and environmental 
noise, transmission schedules, MAC parameters (for example, the backoff 
algorithm and retransmit policy), and available and allowed transmit power for 
individual transmitters and the network.

Both speed and accuracy of coverage area measurement and maximization are 
necessary to manage the phenomenal expanse of applications on the wireless 
medium - cell sites, mobile subscribers, access points, WiFi-connected portable 
devices, and on-demand mobile software.

Infrastructure wireless networks, like WLANs and cellular networks, have their
wireless hops anchored to access points and base stations. In these networks,
much of the network coverage optimization is done before the infrastructure is
deployed. In ad-hoc networks, in contrast, coverage optimization happens while 
the network is in operation. In both cases coverage is limited by the RF 
environment - due to path loss, fading and shadowing - and by interference from 
simultaneous transmissions on conflicting bands. In cellular networks, the 
optimization begins as early as designing the regulatory framework - 
apportioning RF bands by country, region and network operator - to limit 
interference and thereby, improve coverage. WLANs use non-regulated bands, and 
hence the interference management for coverage optimization happens at network 
deployment. Some of the coverage optimization happens during the signalling 
between device and access point (or base station): the device is informed of 
specific channels in the band to use for each transmission in order to avoid 
interference.

Transmitter locations may be restricted due to availability of shelter, power
source, cooling, or other infrastructure restrictions. The schedule of
transmissions too may not be within the control of the designer, given the large
variety of applications. Once a network deployment location is chosen the only 
variables available for coverage optimization possibly are the channels, number 
of transmitters and their transmit powers.

Furthermore, optimization methods that allow the designer to experiment
with configurations quickly are desirable. Such optimization methods are
also required for self-management of wireless networks, where the transmitters
dynamically vary their parameters to increase coverage.


\section{Coverage Maps}
An intermediate step in designing to optimize coverage is the computation of
the `coverage map' for a given set of transmitters. By coverage map we mean the
set of feasible locations for placement of receivers such that each receiver
can decode correctly a transmission intended for it. Our rationale for 
suggesting this step is that computing the coverage map may validate an optimum 
choice of the network's parameters - for example, the transmitter locations.

Computing and then optimizing the coverage map requires consideration of some 
basic constraints of wireless transmissions. The constraints we focus on are 
the following: wireless transmission energy reduces with distance, and 
simultaneous transmissions in close proximity of the receiver cause 
interference. Interference at a receiver may be avoided by ensuring that
transmissions are only on a single `channel' - fundamentally: frequency,
time-slot, or orthogonal code. This may be feasible in infrastructure networks,
since link control protocols are available to transmit out-of-band channel
information to the receiver. However, in less regulated networks, like wireless
ad-hoc networks, simultaneous transmission and reception on the same channel may
be unavoidable, and automated self-management is required.

The first step in the computation of the coverage map is the computation of the
coverage map for a single `channel'. Suppose the coverage map of a given set of
transmitters for one channel is known. Then the designer can attempt to optimize
coverage by varying the location of the transmitters or their transmission
power, or by employing more channels.

While an exact coverage map can be computed for disk models, we shall see that
for the SINR model, we can only estimate the area in coverage - since the
coverage boundaries are not defined by closed-form equations. Extending 
coverage estimation beyond the geometric SINR is further complicated by 
irregular coverage areas, coverage areas that change shape due to fading and 
shadowing, interference caused by transmitters outside the administrative 
control of the network, and asymmetric channel conditions between the 
communicating pair.
\section{Statement of Work}
This thesis contributes new algorithms that: 1) compute the coverage map of a
wireless network for different wireless topology models, and 2) find an optimal 
assignment of transmission parameters that maximizes the coverage area. The 
algorithms for computing the coverage map report the boundary of the set of 
points in coverage. For wireless models in which geometric computation of the 
boundary is inefficient, we show an algorithm that reports a bounded-accuracy 
estimate of the coverage area. Further, we show an algorithm for optimizing 
(maximizing) the coverage area of a given network under topology and energy 
constraints.

We first consider the problem of computing the coverage map for a single
channel, modeling the network with the `protocol model' of Gupta et
al.~\cite{GK00}.
\begin{enumerate}
\item Transmissions occur only in the 2-dimensional plane.
\item Each transmitter's transmission range is a circular disk centered around 
it, and its interference range is a larger concentric disk.
\item The coverage region of a transmitter is the set of points lying on its 
transmission disk and outside every other transmitter's interference disk.
\end{enumerate}
The coverage map is thus the union of the coverage regions of all transmitters.

Figure \ref{figure:CoverageMap} shows an example of a coverage map in the 
protocol model.

\begin{figure}[htbp]
\begin{center}
\includegraphics{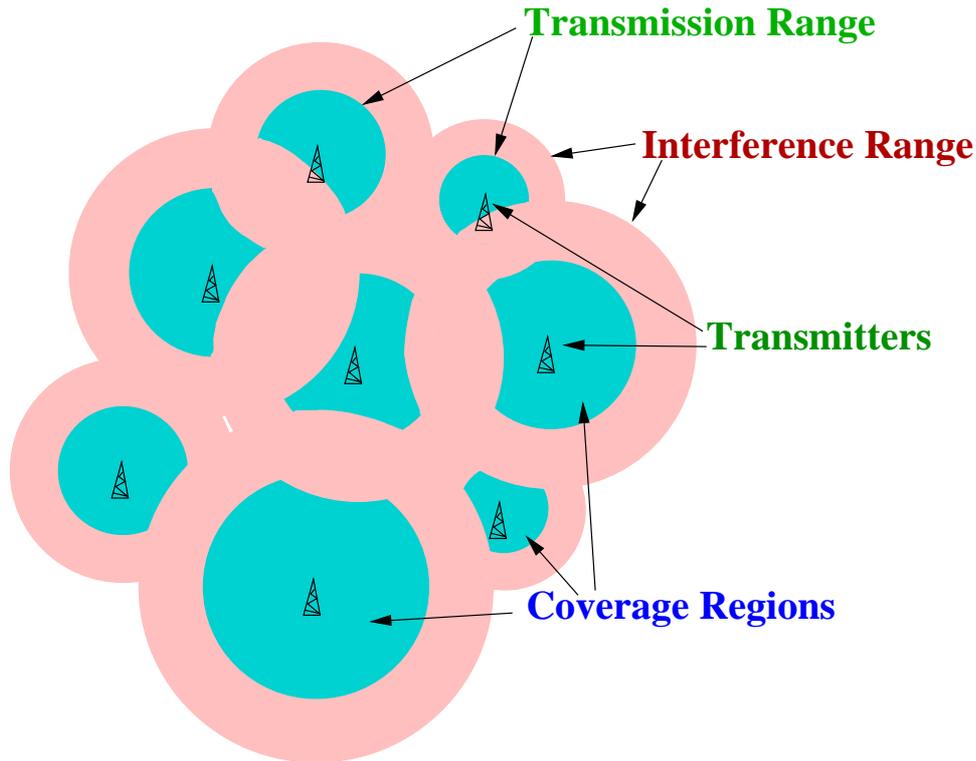}
\end{center}
\caption{A Coverage Map in the Protocol Model}
\label{figure:CoverageMap}
\end{figure}

In the protocol model, coverage is decided by set-membership alone - coverage
regions being points inside appropriate transmission and interference
ranges. This allows us to compute the coverage map efficiently using simple
computational geometric primitives. However, this model precludes {\em path
loss} - the loss of transmission energy with distance from the transmitter.

In Chapter \ref{chapter:Fixed} we demonstrate a method for computing the
coverage map (for a single channel) under our model. This method requires the
entire set of transmitters to be known {\it a priori}. However, the designer
would have to re-calculate the coverage map for an incremental change in the
network, like for example, addition of a new transmitter. In Chapter
\ref{chapter:dynamicCoverage} we extend these results to allow maintenance of
the coverage map when one transmitter is added or removed at a time.

In Chapter \ref{chapter:SINR}, we explore the coverage problem for a more
`realistic' model - the SINR (Signal-to-Interference-plus-Noise-Ratio) model,
also called {\it Physical Model} by Gupta et al.~\cite{GK00}. We observe that
partitioning methods similar to the protocol model studied earlier can be
employed - each coverage region corresponds to a partition. However, there are
differences in the shapes of the curves enclosing the coverage regions, and
hence their representation in coverage computation is different.

In Chapter \ref{chapter:Optimization} we propose a simple algorithm that 
finds an optimal transmit power assignment that maximizes the coverage area of 
a set of transmitters operating on the same channel. This algorithm uses a 
coverage estimation function for the coverage area that may be adapted to any 
deterministic coverage model. We demonstrate the efficiency of this algorithm 
experimentally using the SINR model to compare it with other algorithms. We 
also note how the algorithm can be extended to a generic coverage model.

\chapter{Literature Review}
A recent survey by Phillips et al.~\cite{phillips} covers various coverage 
mapping methods and transmission models reported over the past few decades.
This work considers the problems of computing and optimizing the coverage map 
for a single channel. It is inspired by the asymptotic capacity limits for 
wireless networks espoused in the {\it protocol} and {\it physical} models by 
Gupta et al.~\cite{GK00}.

Our work in coverage map computation generalizes the coverage map computation 
reported by So et al.~\cite{Ye}. They compute the coverage map for a static 
wireless sensor network without considering interference. Our ideas for 
coverage map computation are inspired by generalizations of Voronoi diagrams 
reported by Aurenhammer et al.~in \cite{DisksBalls}, \cite{AurenhammerPowerDia},
 and \cite{AurenhammerNotes}. The extensions of these ideas to dynamic coverage 
maps has been influenced by the works on randomized geometric algorithms from 
Mulmuley (\cite{Ketan}), Aragon et al.~\cite{Seidel}, and Clarkson et 
al.~\cite{RS2}. Our work on the coverage map in the SINR model is related to 
research by Avin et al.~\cite{Chen09}, who report an approximation algorithm 
that decides membership of a point in an SINR coverage map.

Our approach to coverage optimization extends that of `successive refinement' 
by Ahmed et al.~\cite{Keshav} that reports optimum transmit power assignments 
to access points assuming a protocol model.  A recent work by Plets et 
al.~\cite{plets2014calculation} describes a tool for optimal design for indoor 
wireless LANs. Optimal placement of transmitters in bounded geometric areas 
(modeling typical room shapes) has been recently reported by Yu et 
al.~\cite{yu2013wireless}. Our optimization procedure is inspired by Mudumbai 
et al.~\cite{madhow}. They report a randomized procedure to synchronize 
multiple transmissions to send a common message coherently in a distributed 
beamforming system.

Some optimization problems in link scheduling and power control in the SINR 
model are closely related to our problem. Goussevskaia et al.~\cite{gous} show 
that a discrete problem of `single-shot scheduling' with weighted links is 
NP-hard. Lotker et al.~\cite{lotker} and Zander et al.~\cite{zander} give 
efficient algorithms for optimizing the maximum achievable SINR in a set of 
links. Yates et al.~\cite{yates} report an algorithm to optimize the total 
uplink transmit power for users served by a base station, assuming that all 
users meet the minimum SINR constraint. A recent study by Altman et 
al.~\cite{Altman} considers {\em SINR games} played co-operatively and 
co-optively between base stations to maximize coverage area for mobile 
receivers or determine optimal placement for base stations themselves.

A recent study of coverage optimization algorithms for indoor coverage 
appears in Reza et al.~\cite{reza2014comprehensive}. A new optimization model 
based on extrapolation of data collected from measurement tools is given by 
Kazakovtsev in \cite{kazakovtsev2013wireless}.

In a survey of practical tools for coverage optimization, we found recent tools 
appearing in the research literature, like those by Kim et 
al.~\cite{kim2014cell}, Chen et al.~\cite{chen2013placement}, and Zhang et 
al.~\cite{zhang2014m2m}, discuss coverage management using measurements from 
wireless devices in the network. These, along with the use of modern data 
analytics tools - for example, Kim et al.~\cite{kim2014cell} and Kazakovtsev 
(\cite{kazakovtsev2013wireless}) who demonstrate methods for estimating and 
optimizing wireless coverage by analyzing radio ``fingerprints'' - appear to be 
candidate tools of the future.

Extended summaries of our work appear in two publications: Kapadia and Damani 
\cite{kap} and Kapadia, Damani and Kumar \cite{kap2}.

\chapter{Coverage in the Protocol Model: Fixed Set of Transmitters}
\label{chapter:Fixed}
Our work generalizes the coverage map computed by So et al.~\cite{Ye}. They
compute the coverage map for a wireless sensor network without considering
interference.
\section{Problem Statement}
We are given $n$ transmitter locations (points) in the plane. We are also given
the transmission and interference ranges of each transmitter. All transmitters
share the same wireless channel. We need to compute the set of points that lie
within the transmission range of one transmitter, and outside the interference
range of every other. We call this set the `coverage region' of a
transmitter. The union of all $n$ coverage regions is the `coverage map' of the
network.

The shaded region in figure \ref{figure:Coverage} shows the coverage region of 
one transmitter surrounded by 7 other transmitters. The interference ranges 
(disks) are shown in solid perimeter and the transmission disk of transmitter 
$p$ is shown with a dotted perimeter. Note that the transmission disk of other 
transmitters is a subset of the interference disk, and hence is not shown in the 
figure.
\begin{figure}[htbp]
\begin{center}
\includegraphics{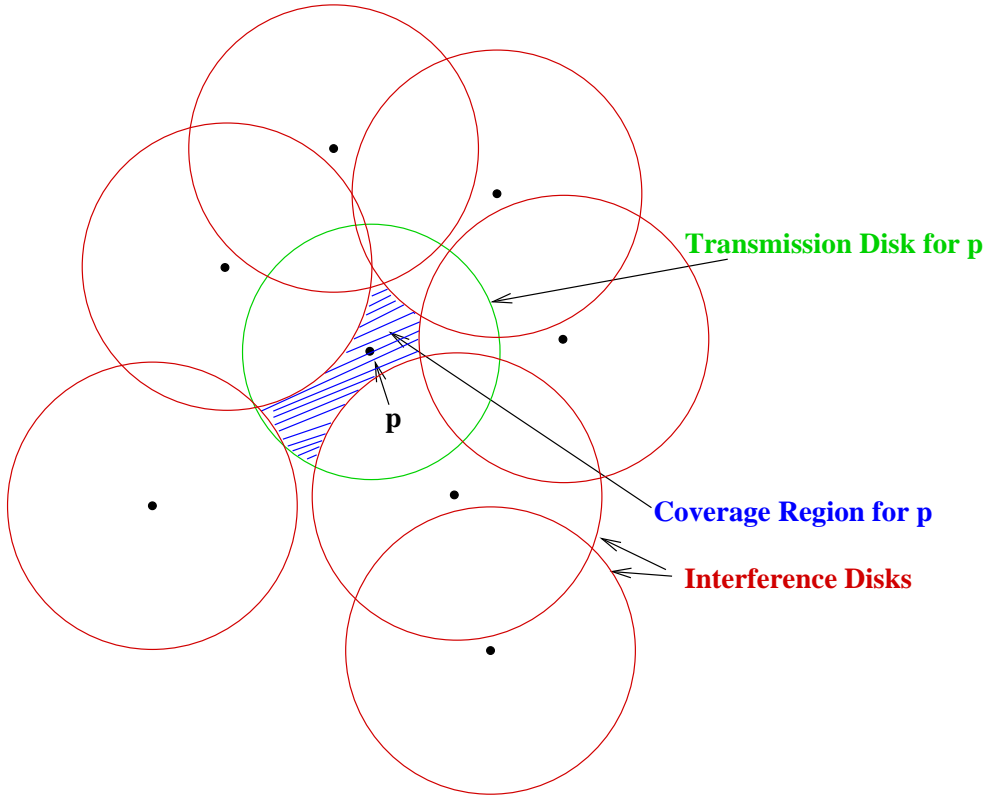}
\end{center}
\caption{Coverage Region of a Transmitter}
\label{figure:Coverage}
\end{figure}
\section{Our Approach}
\label{section:approach}
We show that for an appropriate choice of {\it distance measure}, coverage at 
each point can be computed by considering only certain {\it nearby} 
transmitters. The distance measure imposes a partition of the plane which our 
algorithm uses to compute the coverage map efficiently.

For equal ranges, the partition corresponds to the division of the plane into 
{\em cells} or {\em Voronoi regions} in the well-known {\em closest point 
Voronoi Diagram}\cite{AurenhammerNotes}. For unequal ranges, we use a closely 
related structure called the {\em Power Diagram}. The partition by a power 
diagram is very similar in shape, structure, representation, and properties to 
the Voronoi Diagram\cite{AurenhammerPowerDia}.

The Voronoi Diagram (see, for example Figure \ref{figure:Voronoi}) of a given 
set of points $T$ in 2-D partitions the plane into Voronoi regions - each 
point $p \in T$ corresponds to one region $\triangle(p, T)$. Every point in 
$\triangle(p, T)$ is closer, in Euclidean distance, to $p$ than to any other 
point in $T$.

Aurenhammer et al.~\cite{AurenhammerNotes} give a detailed treatment of the
Voronoi Diagram, its properties and algorithms for its construction.

\begin{figure}[htbp]
\begin{center}
\includegraphics[height=2in]{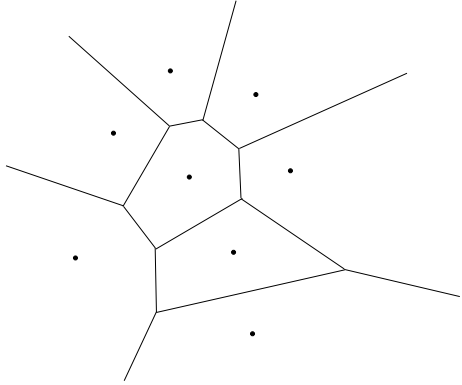}
\end{center}
\caption{Voronoi Diagram of 8 Points}
\label{figure:Voronoi}
\end{figure}
\newcommand{\itemref}[1]{\ref{chapter:Fixed}.{#1}}

Two key claims to justify our approach -
\begin{description}
\item{\bf Claim \itemref{1}:} The coverage region for a transmitter lies
{\it entirely} in its Voronoi (or power) cell.
\item{\bf Claim \itemref{2}:} The Voronoi (or power) cell can be further 
partitioned such that coverage in each (sub-)partition can be decided by 
considering interference from only the {\it nearest} transmitter.
\end{description}
We develop our arguments starting with equal ranges and Voronoi diagrams, and
later extend these to unequal ranges and power diagrams.
\section{Equal Ranges and the Voronoi Diagram}
\label{section:Equal}
We propose an augmentation to the Voronoi diagram of the point set
corresponding to transmitter locations. This augmentation yields an efficient 
algorithm for computing the coverage map.

Let $T$ be the set of transmitter locations, and let $\mathbb{V}(T)$ denote
their Voronoi diagram. Voronoi diagrams are classically used for reasoning
about point sets. We show that we can apply Voronoi diagrams to our problem
concerning interference disks with equal radii.

Let $\bigcirc(c, r)$ be a disk with center at $c$ and radius $r$. We write
$\bigcirc$ when the parameters $c$ and $r$ are clear from the context. We 
define $\delta$, a {\em signed} distance measure between a point and a disk, as 
follows: $\delta(x, \bigcirc(c, r)) = d(x, c) - r$, where $d(x, c)$ denotes 
the Euclidean distance between points $x$ and $c$. $\delta$ is a signed 
distance, since it is negative if $x \in \bigcirc$ and non-negative otherwise.

\begin{observation}
Let $\bigcirc_1(c_1, r)$ and $\bigcirc_2(c_2, r)$ be two disks of equal radii,
$r$, with centers $c_1$ and $c_2$ respectively. For a given point $x$, let
$\delta_1$ and $\delta_2$ be the {\it signed} distances of $x$ from 
$\bigcirc_1$ and $\bigcirc_2$, respectively. $x$ is equidistant from the 
centers $c_1$ and $c_2$ if, and only if, $\delta_1 = \delta_2$.
\label{obs:equal}
\end{observation}
\begin{proof}
$\delta_1 = \delta_2 \Leftrightarrow d(x, c_1) - r = d(x, c_2) - r$\\
$\Leftrightarrow d(x, c_1) = d(x, c_2)$
\end{proof}
Similarly, \[ \delta(x, \bigcirc(c, r)) > \delta(y, \bigcirc(c, r)) 
\Leftrightarrow d(x, c) > d(y, c) \]

We make an observation that relates the distance measure to interference disks.
\begin{observation}
\label{obs:closestInterferes}
Consider a point $x$ on the interference disk for $p$. $x$ is also on the
interference disk of every transmitter closer than $p$.
\end{observation}
\begin{proof}
Let $q \in T \setminus \{p\}$. Let $x$ be closer to $q$ than $p$. Let
$\bigcirc_p$ and $\bigcirc_q$ be the interference disks of $p$ and $q$,
respectively.

Since $d(x, q) < d(x, p) \Rightarrow \delta(x, \bigcirc_q) < \delta(x,
\bigcirc_p)$, if $x$ is in the interference range of $p$, $\delta(x, \bigcirc_p)
< 0$. Which means that $\delta(x, \bigcirc_q) < 0$, i.e. $x$ must also be in
the interference range of $q$.
\end{proof}
The Voronoi diagram of the centers of disks tells us which center is
closest to a given point, in Euclidean distance. Given disks with equal radii,
let us have a diagram which partitions points in the plane according to
the disk nearest to them by the distance measure $\delta$; i.e. a `Voronoi'
diagram of disks. The above Observation \ref{obs:equal} shows that the
`Voronoi' diagram of disks of equal radii using the signed distance measure
$\delta$ is the same as the Voronoi diagram of the centers of the respective
disks. Thus we can compute the classical Voronoi diagram, using the Euclidean
distance measure, instead of the signed distances $\delta$. The ensuing
discussion in this section thus refers directly to the Voronoi diagram of the
transmitter locations and Euclidean point distances from them.

We now state notation and expressions for the points corresponding to a Voronoi
region, its extreme points, and the Voronoi diagram.

The Voronoi region corresponding to a point $p \in T$ is given by:
\[\triangle(p) = \{x\ |\ d(x, p) < d(x, q), \forall q \in T \setminus \{p\}\}\]
The extreme points of the Voronoi region for $p$ are given by:
\[\partial(p) = \{x\ |\ d(x, p) \leq d(x, q), \forall q \in T \setminus \{p\}\}
\setminus \triangle(p)\]
The Voronoi diagram is given by:
\[\mathbb{V}(T) = \displaystyle\bigcup_{p \in T} \partial(p)\]
We now prove an important property of the Voronoi region - it shows that the 
coverage region for a transmitter is confined to that transmitter's Voronoi 
region.

\begin{lemma}[Coverage in Voronoi Region]

\mbox{}
Consider a point $x$ outside $\triangle(p)$ that is on the transmission disk
for $p$. $x$ is also on some interference disk other than that of $p$.
\label{lemma:onlyNeighbors}
\end{lemma}
\begin{proof}
Let $q \in T \setminus \{p\}$, and $x \in \triangle(q)$. The Voronoi property
implies that $p$ is farther away from $x$ than $q$. The transmission disk of
$p$ is a subset of its interference disk; thus $x$ lies on the interference
disk of $p$. Hence, by Observation \ref{obs:closestInterferes}, since $x$ is in
the interference range of $p$, $x$ must also be in the interference range of
$q$.
\end{proof}
Extreme points of $\triangle(p)$ appear as either Voronoi edges or vertices.
These in turn correspond to some other transmitters in $T$ that we call the set
of Voronoi neighbors of $p$, denoted by $\Gamma(p)$. Lets assume that the
Voronoi diagram for $T$ has been built by some classical algorithm 
(see \cite{AurenhammerNotes}). Suppose we delete the point $p$ from $T$, and build
$\mathbb{V}(T \setminus \{p\})$. Figure \ref{figure:DeleteP} shows the effect
of these changes. :
\begin{enumerate}
\item Some existing edges are extended - see dotted portions in Figure
\ref{figure:DeleteP}
\item Some existing edges are deleted - see dashed edges in Figure
\ref{figure:DeleteP}
\item Some new vertices are added - see intersections of dotted
extensions in Figure \ref{figure:DeleteP}
\item Some new edges are added - see edges between new vertices in
Figure \ref{figure:DeleteP}
\end{enumerate}
Note that in figure \ref{figure:DeleteP}, only the neighborhood of $p$ changes
to form $\mathbb{V}(T \setminus \{p\})$. We will formally show later, in Lemma
\ref{lemma:closestNeighbor}, that this is true for every Voronoi diagram.

\begin{figure}[htbp]
\begin{center}
\includegraphics[height=2.5in]{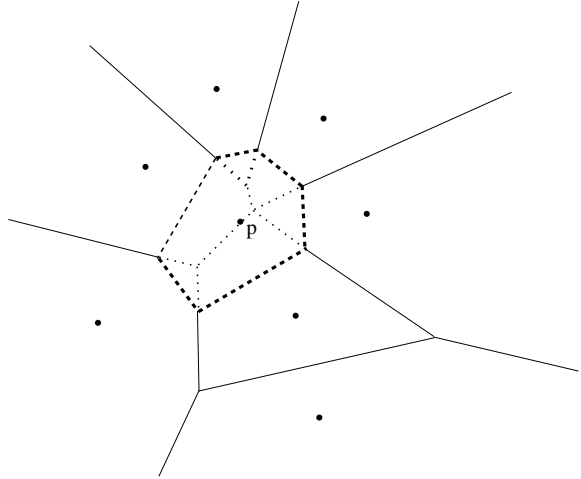}
\end{center}
\caption[Deleting a Transmitter]{Effect of deleting \em{p}}
\label{figure:DeleteP}
\end{figure}

We call the new set of vertices, edges, and extension edges the {\em Voronoi
Frame} corresponding to $p$. The dotted portions in Figure \ref{figure:DeleteP}
are the Voronoi Frame. Edges in the frame correspond to the set of points
equidistant from two transmitters adjacent to $p$ in $\mathbb{V}(T)$. Vertices
in the frame correspond to points equidistant from three (or more) transmitters
adjacent to $p$.

Before we delve into our discussion, we define certain terms in the context of
what has been said so far.
\begin{description}
\item{\textbf{Feasible Coverage Map}:} For a given transmitter $p$, this is
the set of points lying in the Voronoi region of $p$, but outside every
interference disk other than $p$. We reiterate that due to Lemma
\ref{lemma:onlyNeighbors}, to compute coverage, we can restrict our attention
only to the Voronoi region of $p$.
\item{\textbf{Actual Coverage Map}:} For a given transmitter $p$, this is the 
set of points that lie in its Feasible Coverage Map and also on its 
transmission disk. The Actual Coverage Map of $p$ is the intersection of the 
transmission disk of $p$ with the Feasible Coverage Map. An Actual Coverage 
Area is demarcated by arcs, each of which correspond to the portion of the 
periphery of either the transmission disk of $p$, or an interference disk of a 
Voronoi neighbor of $p$ that bounds the Feasible Coverage Map.
\item{\textbf{Contiguous Feasible Region}:} The Feasible Coverage Map may be
composed of several disjoint maximal simply-connected subsets (see, for
example, the two shaded sets in figure \ref{figure:LargeCircles}). These
subsets are called Contiguous Feasible Regions. In the subsequent text we use
the term `Feasible Region' when the word `Contiguous' is obvious from the
context.
\item{\textbf{Voronoi Frame}:} For a given transmitter $p$, this is the set of
points on extension edges, new edges and vertices in $\triangle(p)$ obtained
from deleting the point $p$ and adding new extreme points from the Voronoi
diagram of $T \setminus \{p\}$. The only extreme points in $\mathbb{V}(T)$ that
do not belong in $\mathbb{V}(T \setminus \{p\})$ are the extreme points on the
edges in $\triangle(p)$. The Voronoi Frame is thus the set of points in
$\mathbb{V}(T \setminus \{p\}) \setminus \mathbb{V}(T)$.
\end{description}

\begin{figure}
\centering
\includegraphics[height=3in]{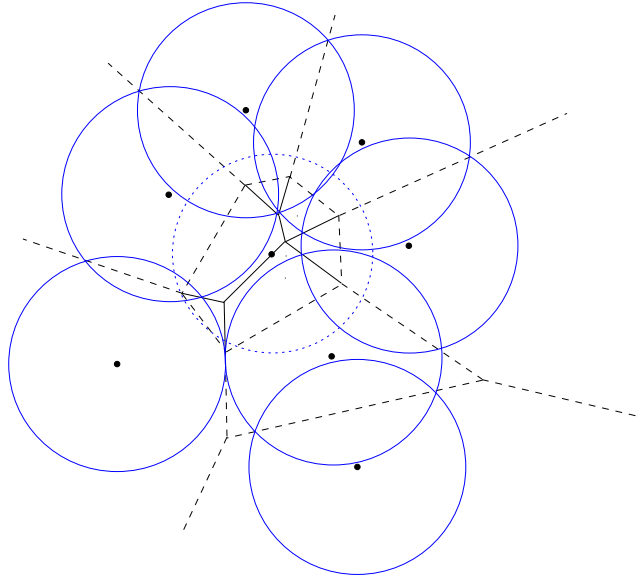}
\caption[Closed Feasible Region]{Feasible Coverage Map: Closed Feasible Region}
\label{figure:MediumCircles}
\end{figure}

\begin{figure}
\centering
\includegraphics[height=3in]{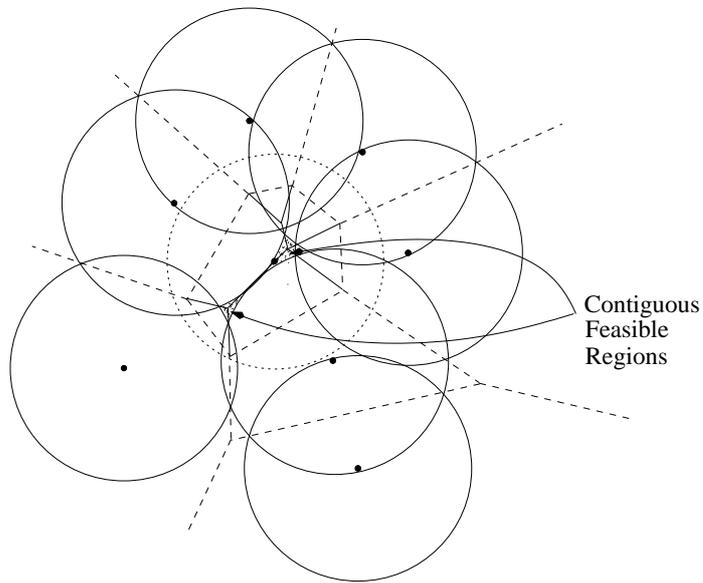}
\caption[Multiple Contiguous Feasible Regions]{Feasible Coverage Map: Multiple Contiguous Feasible Regions}
\label{figure:LargeCircles}
\end{figure}

\begin{figure}
\centering
\includegraphics[height=3in]{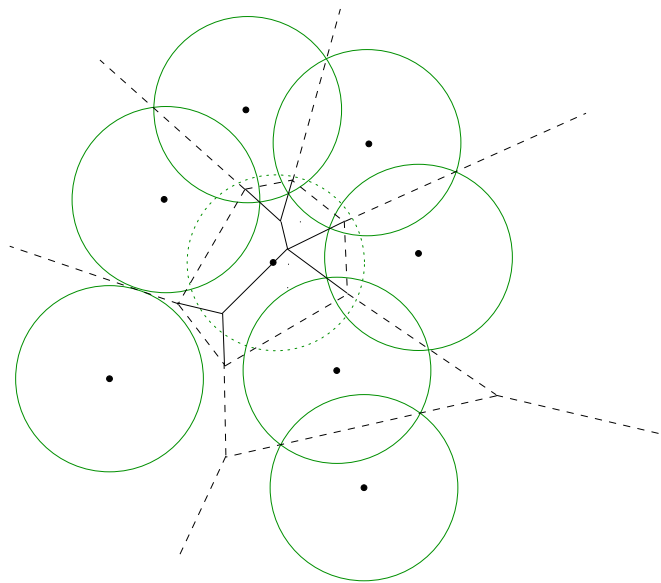}
\caption[Open Feasible Region]{Feasible Coverage Map: Open Feasible Region}
\label{figure:SmallCircles}
\end{figure}

Our goal in this remainder of this section is to demonstrate how the Voronoi 
Frame aids in reasoning about the Feasible and Actual Coverage Maps.

Figures \ref{figure:MediumCircles}, \ref{figure:LargeCircles}, and
\ref{figure:SmallCircles} will help the reader visualize the concepts being
discussed. Each figure corresponds to the same set of transmitter
locations. The interference radii, however, are different.
\begin{enumerate}
\item Interference disks are shown with solid perimeters.
\item Only one transmission disk is shown. It is shown by a dotted perimeter.
\item The original Voronoi diagram is shown by dashed lines.
\item The Voronoi Frame is shown by solid lines.
\item Figure \ref{figure:MediumCircles} illustrates a closed Feasible Region; 
i.e. one that is bounded on all sides by interference disks.
\item Figure \ref{figure:LargeCircles} illustrates two Contiguous Feasible 
Regions in one Feasible Coverage Map.
\item Figure \ref{figure:SmallCircles} illustrates an open Feasible Region.
\end{enumerate}

We now show that only neighbors of $p$ in $\mathbb{V}(T)$ contribute edges
in the Voronoi Frame for $p$. We actually observe a more general result - that
all points in a Voronoi region are closer to Voronoi neighbors than to any
other point -
\begin{lemma}
\label{lemma:closestNeighbor}
Let $x \in \triangle(p)$. The closest point to $x$ in $T \setminus \{p\}$ is a
Voronoi neighbor of $p$ in $\mathbb{V}(T)$.
\end{lemma}
\begin{proof}
Let $q$ be the closest point to $x$ in $T \setminus \{p\}$. Assume that $q$ is
not a neighbor of $p$. We show that this leads to a contradiction. Note that
$p$ and $q$ have distinct Voronoi regions, and each of them is a partition of
the plane. Thus the line segment joining $x$ and $q$ must intersect an
edge of $\partial(p)$. Let this intersection point be $x_0$, and the Voronoi
neighbor of $p$ on this edge be $q_0$. The position of the points is shown in
figure \ref{figure:Closest}.

Since $p$ and $q_0$ are neighbors, $d(x_0, p) = d(x_0, q_0)$ \\
$\Rightarrow$ \{Since $p$ and $q$ are not neighbors, \} $d(x_0, q) > d(x_0,
p)$ \\
$\Rightarrow$ \{Adding $d(x, x_0)$ to both sides, \} $d(x, x_0) + d(x_0, q) =
d(x, q) > d(x, x_0) + d(x_0, q_0)$ \\
$\Rightarrow$ \{Since $x$, $x_0$, and $q$ are collinear, \} $d(x, q) > d(x, x_0)
+ d(x_0, q_0)$ \\
$\Rightarrow$ \{By the triangle inequality, \} $d(x, x_0) + d(x_0, q_0) \geq
d(x, q_0)$ \\
$\Rightarrow d(x, q) > d(x, q_0)$\\
This contradicts the assumption that $q$ is the closest point to $x$ in $T
\setminus \{p\}$.
\end{proof}
\begin{figure}
\centering
\includegraphics{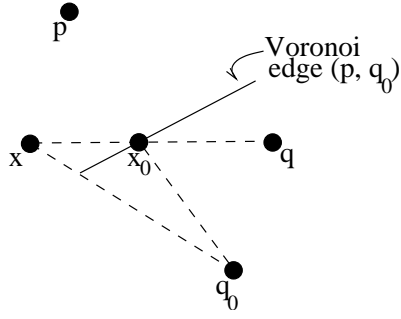}
\caption{Illustration of Lemma \ref{lemma:closestNeighbor}}
\label{figure:Closest}
\end{figure}

\begin{corollary}
\label{cor:neighborFrame}
If a point $x \in \triangle(p)$ is in the interference range of some
transmitter $t \neq p$, it is in the interference range of a neighbor of $p$.
\end{corollary}
\begin{proof}
Lemma \ref{lemma:closestNeighbor} states that the closest point to $x$ in $T
\setminus \{p\}$ is a neighbor of $p$. Thus, using Observation
\ref{obs:closestInterferes} we infer that $x$ is in the interference range of a
neighbor of $p$.
\end{proof}
We denote by $\nu(x, S)$ the set of points in $S$ closest to $x$. We can now 
give an expression for the Voronoi Frame corresponding to $p$. Due to 
Corollary \ref{cor:neighborFrame}, the Voronoi Frame has edges only from
Voronoi neighbors :
\[
\bot(p) = \{x \in \triangle(p)\ |\ \exists \{q_1, q_2\} \subseteq \nu(x, 
\Gamma(p)) \}
\]
Our aim is to compute the Actual Coverage Map for $p$. We will use the Voronoi
Frame of $p$ to do so. Note that not all points on the Voronoi Frame are in the
Feasible Coverage Map. This is because some interference disk corresponding to
a neighbor may include part of an edge on the Voronoi Frame. Our first task is
to exclude points on the Voronoi Frame that are on some interference disk in $T
\setminus \{p\}$. We call the resultant subset of the Voronoi Frame the
\textit{Feasible Coverage Frame}. Corollary \ref{cor:neighborFrame} shows that
to obtain the Feasible Coverage Frame it is sufficient to exclude points on
interference disks adjacent to edges in the Voronoi Frame. For each edge in the
frame, we remove the portion of the edge on one of its adjacent interference
disks. This results in zero, one, or two corresponding edges - depending on
whether the edge has {\em no points} in any Feasible Coverage Region, is {\em
partially} in a Feasible Coverage Region, or is {\em entirely} in a Feasible
Coverage Region. We illustrate this operation in figure
\ref{figure:clipVoronoiFrame}.

\begin{figure}
\centering
\includegraphics[height=3in]{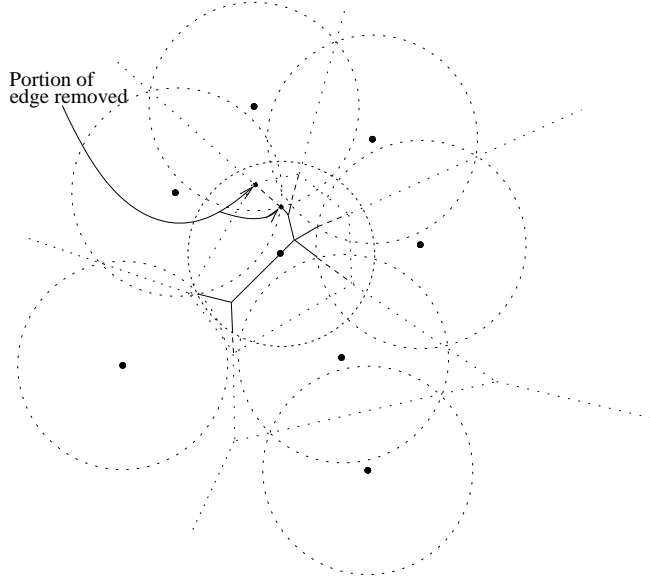}
\caption{Voronoi Frame Yields Feasible Coverage Frame}
\label{figure:clipVoronoiFrame}
\end{figure}

\begin{definition}[Feasible Coverage Frame]
For a given transmitter $p$, this is the set of points that lie on its Voronoi 
Frame and outside the union of interference disks of its neighbors.
\end{definition}
Formally, the Feasible Coverage Frame for transmitter $p$ is given as:
\[
\bot_{g}(p) = \bot(p) \setminus \{x\ :\ x \in \bigcirc^{i}(q), \forall
q \in \Gamma(p)\}
\]
where $\bigcirc^{i}(q)$ denotes the interference range of the transmitter at
$q$.

\setlength{\parskip}{1ex plus 0.5ex minus 0.2ex}

We now have a procedure for building a Voronoi Frame for a transmitter, and for
using this Voronoi Frame to find the transmitter's Feasible Coverage Frame.
We note a property of the Voronoi Frame that we will use to show its
correlation with the Feasible Coverage Region.
\begin{observation}
$\bot(p)$ partitions $\triangle(p)$, and each partition corresponds to exactly 
one neighbor in $\Gamma(p)$.
\label{obs:framePartition}
\end{observation}
\begin{proof}
$\mathbb{V}(T \setminus \{p\})$, due to the Voronoi property, partitions the
plane. By definition, the Voronoi Frame is the subset of this Voronoi diagram
lying inside $\triangle(p)$. Hence $\triangle(p)$ is partitioned by the Voronoi
Frame. Each point in a partition belongs to some Voronoi region in
$\mathbb{V}(T \setminus \{p\})$. Due to Lemma \ref{lemma:closestNeighbor}, a
point in such a partition can be closest only to a neighbor of $p$ in
$\mathbb{V}(T)$.

Each point in a partition can, by definition, be closest only to one point in
$T \setminus \{p\}$. Hence, the partition corresponds to exactly one neighbor.
\end{proof}
Note that the edges in the Voronoi Frame bounding this partition correspond to
the edges contributed by a neighbor of $p$. Further, since $p$ is closer than
$q$ to each point in this partition, the edge between $p$ and $q$ also bounds
the partition.

This observation is illustrated by figure \ref{figure:cellPartition}. We denote
by $\blacktriangle(p,q)$ the partition of $\triangle(p)$ by edges on the Voronoi
Frame corresponding to $q$. Formally,
\[ \blacktriangle(p, q) = \triangle(p) \cap \{ x\ |\ q \in \nu(x, \Gamma(p)) \} 
\]

\begin{figure}
\centering
\includegraphics[height=3in]{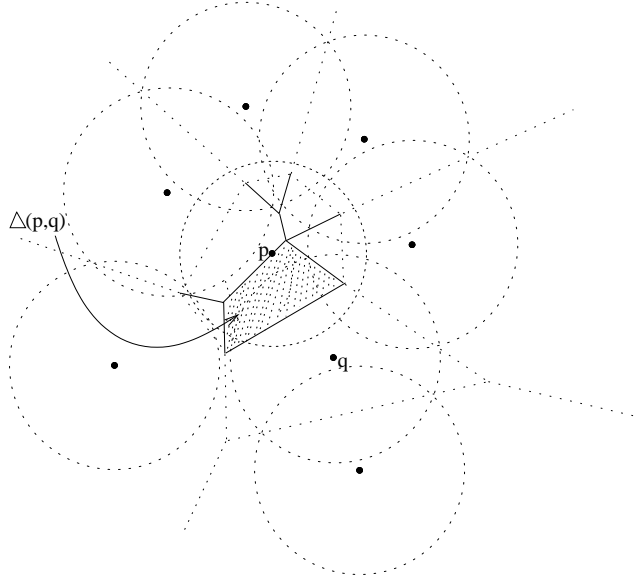}
\caption{Voronoi Frame Partitions Voronoi Region}
\label{figure:cellPartition}
\end{figure}

We note a correlation between $\blacktriangle(p, q)$ and the interference disk 
for $q$.

\begin{corollary}
If $x \in \blacktriangle(p,q)$ and $x$ is in the interference range of some
transmitter $t \neq p$, then $x$ is in the interference range of $q$.
\label{cor:onlyPartition}
\end{corollary}
\begin{proof}
By definition, every point in $\blacktriangle(p, q)$ is closer to $q$ than any 
other
point in $T$. Since $x$ is closer to $q$ than $t$, we can apply Observation
\ref{obs:closestInterferes} to see that $x$ is in the interference range of
$q$.
\end{proof}
We refer again to figures \ref{figure:MediumCircles},
\ref{figure:LargeCircles}, and \ref{figure:SmallCircles}. We observe that each
Contiguous Feasible Region is bounded by the arcs of the rims of interference
disks. These disks correspond to edges in the Feasible Coverage Frame enclosed
within the Feasible Region.

We could compute the Actual Coverage Map of a transmitter directly by
intersections of each Contiguous Feasible Region with the transmission
disk. This would get us the coverage map. However, this would require our
algorithm to represent the Contiguous Feasible Region as a sequence of arcs.
Instead, we obtain the Actual Coverage Map by an alternative approach that
uses the Voronoi properties.

We denote the Actual Coverage Map of a transmitter $p$ by $\chi(p)$. The
following result shows that the partition given by the Voronoi Frame allows us
to compute the coverage map by excluding interference from just neighboring 
transmitters.

\begin{theorem}[Transmitter's coverage map can be computed by excluding 
interference only from Voronoi neighbors]
\label{thm:main}
\[
\chi(p) = \bigcup_{q \in \Gamma(p)} (\bigcirc^{t}(p) \cap \blacktriangle(p,q))
\setminus \bigcirc^{i}(q)
\]
\end{theorem}
\begin{proof}
Let $\bigcirc^{i}(p)$ and $\bigcirc^{t}(p)$ denote the interference and
transmission disks, respectively, of transmitter $p$. Lemma
\ref{lemma:onlyNeighbors} implies that the Actual Coverage Map lies inside
$\triangle(p)$. Also, Observation \ref{obs:framePartition} states that the
Contiguous Feasible Region for $p$ is composed of contributions from each
neighbor. Corollary \ref{cor:onlyPartition} shows that to find points within
$\blacktriangle(p, q)$ that lie in the Feasible Coverage Map, it is sufficient 
only
to exclude points on the interference disk of $q$. Thus the Actual Coverage
Area can be computed from the individual regions contributed by each partition.
\end{proof}
The advantage of using the Voronoi Frame is now clear - we need only the
Feasible Frame to represent the Actual Coverage Map. Also, to compute the
Actual Coverage Map, only two arc intersection computations are required for
each neighbor $q$ - one for $\bigcirc^{t}(p)$, and one for $\bigcirc^{i}(q)$.

In order to compute $\chi(p)$ we need a generalized polygon representation
that allows circular arcs as edges. Berberich et al.~\cite{ConicPolygon}
study intersections of polygons with arcs. We defer discussion on these
generalized polygons until Section \ref{section:Algo}.

In this section we have shown how to compute the coverage map for a set of
transmitters having the same interference (and transmission, respectively)
radius. We generalize the arguments presented here to transmitters with unequal
interference (and transmission, respectively) radii in Section
\ref{section:UnEqual}.
\section{Unequal Ranges and the Power Diagram}
\label{section:UnEqual}
Note that Observation \ref{obs:equal} does not hold when the interference
radii (and transmission radii, respectively) are not the same. This is because
equal minimum Euclidean distance of point $x$ from two disks $\bigcirc(c_1,
r_1)$ and $\bigcirc(c_2, r_2)$, where $r_1 \neq r_2$ does not imply equal
Euclidean distance of the centers $c_1$ and $c_2$ from $x$. Thus, we need find
an alternative distance measure.

The Voronoi Diagram of the centers of disks tells us which center is closest to
a given point, in Euclidean distance. The \textit{Power Diagram} (see
Aurenhammer et al.~\cite{AurenhammerPowerDia}) is a generalization of the
Voronoi Diagram; and is based on a different distance measure (between a point
and a disk), called the \textit{Power Distance}. We denote the power distance
by $\rho$.

The power distance of a point $x$ from a disk $\bigcirc$ of radius $r$ and
center $c$ is defined by $\rho(x, \bigcirc) = d(x, c)^2 - r^2$; where $d(x, c)$
is the Euclidean distance between $p$ and the center $c$. Geometrically, the
power distance of a point outside a disk is the square of the length of the
tangent from that point to the disk rim. Inside the disk perimeter, the power
Distance is negative in sign, and corresponds to the square of half the length
of the chord normal to the line joining the point and the center of the disk.

Figure \ref{figure:UnequalBadCase} shows a power diagram of 7 disks in the
plane. Some fundamental properties of the power diagram are stated in
\cite{AurenhammerPowerDia}. The power diagram for $\Psi$, a set of disks,
partitions the plane into convex polygonal regions, i.e. shapes exactly like
Voronoi regions. A point lies in the \textit{Power Region} corresponding to
disk $\bigcirc \in \Psi$ if its power distance from $\bigcirc$ is less than its
power distance from every other disk in $\Psi$.

\begin{figure}
\centering
\includegraphics[height=3in]{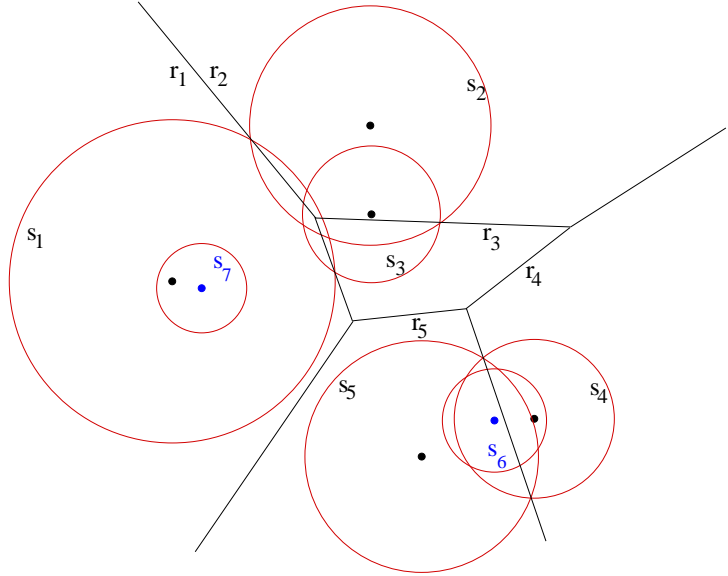}
\caption[Power Diagram]{Power Diagram: 7 disks, 5 regions (adapted from
\cite{AurenhammerPowerDia})}
\label{figure:UnequalBadCase}
\end{figure}
The power diagram for disks with equal radii (under distance measure $\rho$) is
a Voronoi diagram (under Euclidean distance measure $d(p, c)$.) The
generalization, however, leads to instances where a disk does not have a
corresponding power region. One such instance is shown in figure
\ref{figure:UnequalBadCase}. There are no power regions corresponding to disks
$s_6$ and $s_7$; since no point in the plane is closest (by $\rho$) to any of
these disks. In this chapter we assume that every disk has a corresponding
power region. We discuss the impact of relaxing this assumption later in
Chapter \ref{chapter:dynamicCoverage}.

The power diagram generalizes the Voronoi distance measure of $\delta$ to 
$\rho$. This leads us to three key facts, which together form the core of our 
progression from Voronoi diagrams to power diagrams as tools for computing the 
coverage map.
\begin{enumerate}
\item A generalization of Observation \ref{obs:closestInterferes} shows the
same relationship exists between the power region and interference disks. This
generalization is shown in Observation \ref{obs:closestPowerInterferes}.
\item A generalization of Lemma \ref{lemma:onlyNeighbors} (Lemma
  \ref{lemma:onlyPowerNeighbors} also holds when the distance measure is
  replaced by $\rho$. This means that the Actual Coverage Map will lie only in
  the power region.
\item A generalization of Lemma \ref{lemma:closestNeighbor} is also possible,
as we will see shortly in Lemma \ref{lemma:closestPowerNeighbor}. This means
that we need to consider interference only from transmitters that are Power
Neighbors.
\end{enumerate}
In Table \ref{table:VoronoiPower}, we introduce the new notation for unequal ranges and the power
diagram. We also note the corresponding notation with equal ranges and the
Voronoi diagram.

\begin{table}
\centering
\begin{tabular}{|p{2in}|p{2in}|}
\hline
{\bf Equal Ranges} & {\bf Unequal Ranges} \\
\hline
Voronoi Diagram $\mathbb{V}(T)$ & Power Diagram $\mathbb{P}(T)$ \\
\hline
Transmitter Location $p$ & Transmitter Location $\tilde{p}$ \\
\hline
Voronoi Region $\triangle(p)$ & Power Region $\triangle(\tilde{p})$ \\
\hline
Voronoi Neighbors $\Gamma(p)$ & Power Neighbors $\Gamma(\tilde{p})$ \\
\hline
Voronoi Frame $\bot(p)$ & Power Frame $\bot(\tilde{p})$ \\
\hline
Voronoi Region partition by Voronoi Frame $\blacktriangle(p, q)$ & Power Region 
partition by Power Frame $\blacktriangle(\tilde{p}, \tilde{q})$ \\
\hline
\end{tabular}
\label{table:VoronoiPower}
\caption{Notation for Unequal Ranges}
\end{table}
We begin with an observation that relates the distance measure $\rho$ to 
interference disks. Note its correspondence with Observation 
\ref{obs:closestInterferes}

\begin{observation}
\label{obs:closestPowerInterferes}
Consider a point $x$ on the interference disk for $\tilde{p}$. $x$ is also on
the interference disk of every transmitter closer than $\tilde{p}$ in power
distance.
\end{observation}
\begin{proof}
Let $\tilde{q} \in T \setminus \{\tilde{p}\}$. Let $x$ be closer, by power
distance measure $\rho$, to $\tilde{q}$ than $\tilde{p}$. Let
$\bigcirc_{\tilde{p}}$ and $\bigcirc_{\tilde{q}}$ be the interference disks of
$\tilde{p}$ and $\tilde{q}$, respectively.

Since $\rho(x, \bigcirc_{\tilde{q}}) < \rho(x, \bigcirc_{\tilde{p}})$, if $x$
is in the interference range of $\tilde{p}$, $\rho(x, \tilde{p}) < 0$. Which
means that $\rho(x, \tilde{q}) < 0$, i.e. $x$ must also be in the interference
range of $\tilde{q}$.
\end{proof}
We now observe the following generalization of Lemma \ref{lemma:onlyNeighbors}
-
\begin{lemma}[Coverage in Power Region]

\mbox{}
%
Consider a point $x$ outside $\triangle(\tilde{p})$ that is on the transmission
disk for $\tilde{p}$. $x$ is also on some interference disk other than that
of $\tilde{p}$.
\label{lemma:onlyPowerNeighbors}
\end{lemma}
\begin{proof}
Let $\tilde{q} \in T \setminus \{\tilde{p}\}$, and $x \in
\triangle(\tilde{q})$. Since the power diagram is a partition of the plane,
such a $\tilde{q}$ exists. Thus, $x$ is closer to $\tilde{q}$, in power
distance, than it is to $\tilde{p}$. The transmission disk of $\tilde{p}$ is a
subset of its interference disk; thus $x$ lies on the interference disk of
$\tilde{p}$. Hence, by Observation \ref{obs:closestPowerInterferes} $x$ must
also be on the interference disk of $\tilde{q}$.
\end{proof}

\begin{figure}
\centering
\includegraphics[height=2.5in]{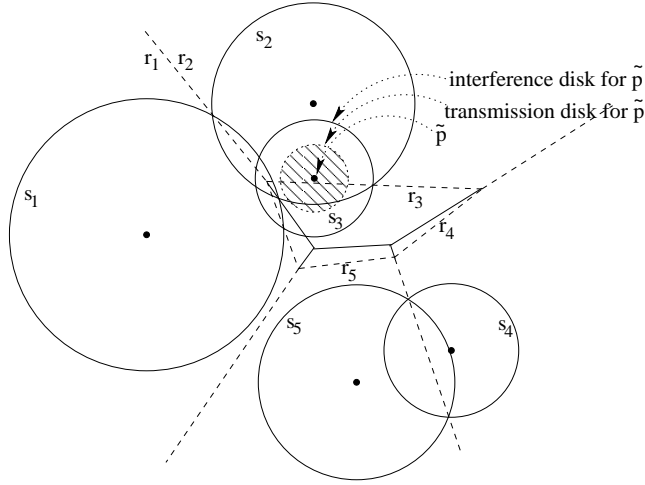}
\caption{Power Diagram with Power Frame}
\label{figure:PowerFrame}
\end{figure}

In the ensuing discussion, we will implicitly assume the use of the distance
measure $\rho$, the power distance; that is, we will say `closest' or `closer'
(respectively, `farthest' or `farther') to mean closest or closer (farther,
farthest, respectively) in power distance.

\begin{lemma}
\label{lemma:closestPowerNeighbor}
Assume that every transmitter has a non-empty power region. Let $x$ be a point
in the power region of transmitter $\tilde{p}$, i.e. $\triangle(\tilde{p})$. If
$\tilde{q}$ is the closest (by power distance) transmitter in $T \setminus
\{\tilde{p}\}$ to $x$, then $\tilde{q}$ is a power neighbor of $\tilde{p}$.
\end{lemma}
\begin{proof}
We denote by $\triangle(\tilde{p}, T)$ the power region of $\tilde{p}$ in the
power diagram of the set of transmitters $T$. Similarly, $\triangle(\tilde{q},
T \setminus \{\tilde{p}\})$ denotes the power region of $\tilde{q}$ in the set
of transmitters $T \setminus \{\tilde{p}\}$. By definition, $x \in
\triangle(\tilde{p}, T) \cap \triangle(\tilde{q}, T \setminus \{\tilde{p}\})$.

Assume that $\tilde{q}$ is not a power neighbor of $\tilde{p}$. We show that 
this leads to a contradiction.
\begin{description}
\item{\em Case 1 :} Assume that an extreme point $x_0$ of $\triangle(\tilde{p},
T)$ lies in $\triangle(\tilde{q}, T \setminus \{\tilde{p}\})$. We show that
this leads to a contradiction.

Let $\tilde{q_0}$ be a power neighbor of $\tilde{p}$ corresponding to the point
$x_0$.
Thus, \[ \rho(x_0, \tilde{p}) = \rho(x_0, \tilde{q_0}) \]
However, since $\tilde{q}$ is closer to $x_0$ than $\tilde{q_0}$,
\[ \rho(x_0, \tilde{p}) = \rho(x_0, q_0) > \rho(x_0, \tilde{q}) \]
This is a contradiction since no transmitter can be closer to $x_0$ than
$\tilde{p}$. Thus, no extreme points of $\triangle(\tilde{p}, T)$ lie in
$\triangle(\tilde{q}, T \setminus \{\tilde{p\}})$.
\item {\em Case 2 :} Assume that no extreme point of $\triangle(\tilde{p}, T)$
  lies in $\triangle(\tilde{q}, T \setminus \{\tilde{p}\})$. We show that this
  leads to a contradiction.

$\Rightarrow$ All edges
in $\triangle(\tilde{p}, T)$ lie outside $\triangle(\tilde{q}, T \setminus
\{\tilde{p}\})$.\\
$\Rightarrow$ Either $\triangle(\tilde{q}, T \setminus
\{\tilde{p}\}) \subset \triangle(\tilde{p}, T)$, or
$\triangle(\tilde{q}, T \setminus \{\tilde{p}\}) \cap \triangle(\tilde{q}, T
\setminus \{\tilde{p}\}) = \phi$.\\
In other words, the power region of $\tilde{p}$ in $\mathbb{P}(T)$ either
encloses that  of $\tilde{q}$ in $\mathbb{P}(T \setminus \tilde{p})$, or the
two power regions are disjoint.\\
Since $x \in \triangle(\tilde{p}, T) \cap \triangle(\tilde{q}, T \setminus
\{\tilde{p}\}) \Rightarrow \triangle(\tilde{p}, T) \cap
\triangle(\tilde{q}, T \setminus \{\tilde{p}\}) \neq \phi$\\
Thus, $\triangle(\tilde{q}, T \setminus \{\tilde{p}\}) \subset
\triangle(\tilde{p}, T)$ \\
$\Rightarrow$ each point in $\triangle(\tilde{q}, T \setminus \{\tilde{p}\})$
  is closer to $\tilde{p}$ than $\tilde{q}$. \\
$\Rightarrow \triangle(\tilde{q}, T) = \phi$\\
This contradicts the assumption that no power region is empty.\qedhere
\end{description}
\end{proof}

Having established the basic correspondence between Voronoi diagrams and power
diagrams, we note that Observation \ref{obs:framePartition}, Corollary 
\ref{cor:onlyPartition}, and Theorem \ref{thm:main} are directly applicable 
to power diagrams. We give corresponding results next.


\begin{observation}
$\bot(\tilde{p})$ partitions $\triangle(\tilde{p})$, and each partition
corresponds to exactly one neighbor in $\Gamma(\tilde{p})$.
\label{obs:powerFramePartition}
\end{observation}
\begin{proof}
$\mathbb{P}(T \setminus \{\tilde{p}\})$, by definition, partitions the
plane. By definition, the power frame is the subset of this power diagram lying
inside $\triangle(\tilde{p})$. Hence $\triangle(\tilde{p})$ is partitioned by
the power frame. Each point in a partition belongs to some power region in
$\mathbb{P}(T \setminus \{\tilde{p}\})$. Due to Lemma
\ref{lemma:closestPowerNeighbor}, a point in such a partition can be closest
only to a neighbor of $\tilde{p}$ in $\mathbb{P}(T)$.

Each point in a partition can, by definition, be closest only to one point in
$T \setminus \{\tilde{p}\}$. Hence, the partition corresponds to exactly one
neighbor.
\end{proof}

We note that the same relationship as Corollary \ref{cor:onlyPartition} holds
between $\blacktriangle(\tilde{p}, \tilde{q})$ and the interference disk for
$\tilde{q}$.

\begin{corollary}
If $x \in \blacktriangle(\tilde{p}, \tilde{q})$ and $x$ is in the interference 
range
of some transmitter $\tilde{t} \neq \tilde{p}$, then $x$ is in the interference
range of $\tilde{q}$.
\label{cor:onlyPowerPartition}
\end{corollary}
\begin{proof}
By definition, every point in $\blacktriangle(\tilde{p}, \tilde{q})$ is closer 
to $\tilde{q}$ than any other point in $T$. Since $x$ is closer to $\tilde{q}$
than $\tilde{t}$, we can apply Observation \ref{obs:closestPowerInterferes} to
see that $x$ is in the interference range of $\tilde{q}$.
\end{proof}
We now state and prove our main theorem. Note that it is a generalization of
Theorem \ref{thm:main}.

We denote the Actual Coverage Map of a transmitter $\tilde{p}$ by
$\chi(\tilde{p})$. The following result shows that, if no power region is
empty, then the partition given by the power frame allows us to compute the
coverage map by excluding interference from just one transmitter.

\begin{theorem}
\mbox{}
Assume no power region is empty. Then, 
\[
\chi(\tilde{p}) = \bigcup_{q \in \Gamma(\tilde{p})} (\bigcirc^{t}(\tilde{p})
\cap \blacktriangle(\tilde{p},\tilde{q})) \setminus \bigcirc^{i}(\tilde{q})
\]
\end{theorem}
\begin{proof}
Let $\bigcirc^{i}(\tilde{p})$ and $\bigcirc^{t}(\tilde{p})$ denote the
interference and transmission disks, respectively, of transmitter
$\tilde{p}$. Since no power region is empty, Lemma
\ref{lemma:onlyPowerNeighbors} implies that the Actual Coverage Map lies
inside $\triangle(\tilde{p})$. Also, Observation \ref{obs:powerFramePartition}
states that the Contiguous Feasible Region for $\tilde{p}$ is composed of
contributions from each neighbor. Corollary \ref{cor:onlyPowerPartition} shows
that to find points within $\blacktriangle(\tilde{p}, \tilde{q})$ that lie in 
the
Feasible Coverage Map, it is sufficient only to exclude points on the
interference disk of $\tilde{q}$. Thus the Actual Coverage Map can be computed
from the individual regions contributed by each partition.
\end{proof}

\subsection{Removing Redundant Transmitters}
The power region corresponding to a circle may be empty - as is the case with
$s_7$ in figure \ref{figure:UnequalBadCase}. No point on the power bisector of
$s_1$ and $s_7$ appears in the power diagram since each point on this bisector
is closer to either $s_2$, $s_3$, or $s_5$.

However, as we will show below, if a disk $\bigcirc$ has an empty power region
then it is included in the union of other disks in $\Psi$. A disk that belongs
in the union of other disks has an empty coverage map, since every point on it
is in the interference range of some other transmitter. Hence, for our
purposes, during preprocessing we can remove disks that do not have a
corresponding power region.

We show a more general result - that if a disk and its corresponding power
region have no points in common, then that disk is included in the union of
other disks.
\begin{lemma}[Empty Power Regions]
\mbox{}
Let $\tilde{q}$ be a transmitter with interference disk $\bigcirc_{\tilde{q}}$,
such that $\bigcirc_{\tilde{q}} \cap \triangle(\tilde{q}) = \phi$. Then,
\[\bigcirc_{\tilde{q}} \subseteq \bigcup_{\tilde{p} \in T \setminus \tilde{q}}
\bigcirc_{\tilde{p}} \]
\end{lemma}
\begin{proof}
We prove this result by contradiction.\\

Let $\bigcirc_{\tilde{q}} \nsubseteq \displaystyle\bigcup_{\tilde{p} \in T
  \setminus \tilde{q}} \bigcirc_{\tilde{p}}$\\
$\Rightarrow \exists x \in \bigcirc_{\tilde{q}}$ such that $\forall \tilde{p}
\neq \tilde{q},\ x \notin \bigcirc_{\tilde{p}}$\\
$\Rightarrow (\rho(x, \bigcirc_{\tilde{q}}) < 0) \land (\forall \tilde{p}
\neq \tilde{q},\ \rho(x, \bigcirc_{\tilde{q}}) > 0)$\\
$\Rightarrow x \in \triangle(\tilde{q})$\\
$\Rightarrow x \in \triangle(\tilde{q}) \cap \bigcirc_{\tilde{q}}$\\
This contradicts the assumption that $\triangle(\tilde{q}) \cap
\bigcirc_{\tilde{q}}$ is empty.
\end{proof}
This fact justifies our pre-processing step for removing disks that have an
empty power region.
\section{Algorithm}
\label{section:Algo}
We collate the observations made in the preceding text into an algorithm. The
inputs to the algorithm are: a set $T$ of transmitters, their locations in the
plane, and their transmission and interference radii. The algorithm outputs a
coverage map for $T$, denoted by $\widehat{\chi}(T)$.

Notation defined in Table \ref{table:VoronoiPower} is used in the algorithm. In 
addition, we denote by $\hbar(\tilde{p}, \tilde{q})$ the half-space of points 
closer to $\tilde{p}$ than $\tilde{q}$.
\begin{algorithm}[Coverage Map]
\label{algo:coverage-map}
\end{algorithm}
\begin{enumerate}
\item Initialize: $\widehat{\chi}(T) \leftarrow \phi$.
\item\label{step:pdcons} Compute the Power Diagram $\mathbb{P}(T)$.
\item For each transmitter $\tilde{p} \in T$, do
   If $\triangle(\tilde{p}) = \phi$, $T \leftarrow T \setminus \{\tilde{p}\}$.
\item\label{step:mainloop} For each transmitter $\tilde{p} \in T$, do
   \begin{enumerate}
   \item $\chi(\tilde{p}) \leftarrow \phi$
   \item \label{step:pdcons2} Find the Power Diagram of $\Gamma(\tilde{p})$,
   i.e.$\mathbb{P}(\Gamma(\tilde{p}))$.
   \item \label{loop:parts} For each region $\triangle(\tilde{q},
   \Gamma(\tilde{p}))$, do
      \begin{enumerate}
      \item\label{step:polyint} $\blacktriangle(\tilde{p}, \tilde{q})
      \leftarrow \triangle(\tilde{q}, \Gamma(\tilde{p})) \cap \hbar(\tilde{p},
      \tilde{q})$
      \item\label{step:circint} $\chi(\tilde{p}) \leftarrow \chi(\tilde{p})
      \cup ((\blacktriangle(\tilde{p}, \tilde{q}) \cap \bigcirc^t_{\tilde{p}})
      \setminus \bigcirc^i_{\tilde{q}})$
      \end{enumerate}
   \end{enumerate}
\item For each transmitter $\tilde{p} \in T$, do
$\widehat{\chi}(T) \leftarrow \widehat{\chi}(T) \cup \chi(\tilde{p})$
\end{enumerate}
\subsection{Running Time Analysis}
We show a proof sketch of a result from Aurenhammer et 
al.~\cite{AurenhammerPowerDia} that bounds the number of power edges in a power 
diagram.
\begin{observation}
\label{obs:numedges}
The number of power edges in a power diagram is less than $3n - 6$.
\end{observation}
\begin{proof}
The power diagram in the plane is a planar graph. Its dual graph
$D(\mathbb{P})$ contains exactly one vertex for each region of
$\mathbb{P}$. Two vertices of $D(\mathbb{P})$ are connected  by an edge if, and
only if, the boundaries of the corresponding regions of $\mathbb{P}$ have an
edge in common. $D(\mathbb{P})$ is a triangulation on $n$ vertices. A
triangulation on $n$ vertices cannot have more than $3n - 6$ edges.
\end{proof}
We make an observation that the sum of the number of neighbors over all
transmitters is linear in $n$. This result will be invoked in our proof.
\begin{observation}[Sum of Neighbors]
\label{obs:linear}
\[\sum_{\tilde{p} \in T} |\Gamma(\tilde{p})| = O(n)\]
\end{observation}
\begin{proof}
The sum $\displaystyle\sum_{\tilde{p}\in T}|\Gamma(\tilde{p})|$ is also the
number of ordered pairs $(\tilde{p}, \tilde{q})$ such that $\tilde{p}$ and
$\tilde{q}$ are neighbors in $\mathbb{P}(T)$. Since each power edge in
$\mathbb{P}(T)$ corresponds to two transmitters, the latter is twice the number
of power edges, which is $O(n)$ by Observation \ref{obs:numedges}
\end{proof}
\begin{theorem}[Runtime]
\mbox{}
The coverage map of `$n$' transmitters with equal or unequal ranges can be
constructed in $O(n\log{n})$ time.
\end{theorem}
\begin{proof}
\begin{description}
\item{Step \ref{step:pdcons}} A power diagram of $n$ disks in the plane
can be constructed in $O(n\log{n})$ time (see \cite{DisksBalls}).
\item{Step \ref{step:pdcons2}:} The power diagram of $\Gamma(\tilde{p})$ can be
constructed in $O(|\Gamma(\tilde{p})|\log{|\Gamma(\tilde{p})|}$ time. Now,
$\log{|\Gamma(\tilde{p})|} \leq \log{n}$\\
$\Rightarrow \displaystyle\sum_{\tilde{p} \in T} |\Gamma(\tilde{p})|
\log{|\Gamma(\tilde{p})|} \leq (\log{n}) \sum_{\tilde{p}\in T}|
\Gamma(\tilde{p})|$\\
$\Rightarrow$ \{By Observation \ref{obs:linear}\} $(\log{n}) \displaystyle\sum_{\tilde{p}\in T} |\Gamma(\tilde{p})| = O(n\log{n})$\\
The total time to compute the power diagrams for all transmitters is thus
$O(n\log{n})$.
\item{Step \ref{step:polyint}:} A well-known algorithm (see \cite{Rourke}) for
convex polygon intersection can be used to compute the partition. This
algorithm is linear in the total number of edges, i.e. in our case
$O(|\Gamma(\tilde{p})|)$. By Observation \ref{obs:linear}, the total time taken executing this step is $O(n)$.
\item{Step \ref{step:circint}:}
The union and set difference operations can be performed by the sweep-line algorithm from \cite{ConicPolygon}. This computation is also linear time in the number of line segments (edges) and arcs; i.e. in our case
$O(|\Gamma(\tilde{p})|)$. By Observation \ref{obs:linear}, the total time taken executing this step is $O(n)$.
\item{Step \ref{loop:parts}:} Each transmitter $\tilde{p}$ can contribute a
partition only to a neighbor (Lemma \ref{lemma:onlyPowerNeighbors}). Thus the
total number of partitions created by the algorithm is the sum of neighbors,
which is $O(n)$ by Observation \ref{obs:linear}. Since each edge appears in at most two partitions, the total number of edges created in this step is also $O(n)$.
\end{description}
Hence, the coverage map can be computed in $O(n\log{n})$ time.
\end{proof}
\section{A Lower Bound on Coverage Map Computation}
We show here that $O(n\log{n})$ time is {\em optimal} in an algebraic decision-
tree model. We use a result on the lower bound of the classical {\em $\epsilon$-
closeness problem} to show that computing the coverage map is 
$\Omega(n\log{n})$.
Our proof is adapted from the proof for the lower bound for constructing 
Voronoi diagrams by reduction from the $\epsilon$-closeness problem 
\cite{AurenhammerNotes}.

We first present a formal representation of the coverage map. We then use this 
representation to locate a transmitter with a certain property we call {\em 
interference-bound}. We show that this operation takes $O(n)$ time. We then 
reduce the $\epsilon$-closeness problem to that of computing the coverage map 
for a suitable set of transmitters and using it to locate an interference-bound 
transmitter. Given a coverage map, an interference-bound transmitter can be 
found in $O(n)$ time. Thus, if the coverage map can be constructed in 
$o(n\log{n})$ time, then $\epsilon$-closeness can be solved in $o(n\log{n})$ 
time as well. This reduction shows that constructing the coverage map is 
$\Omega(n\log{n})$.
\subsection{A Representation of the Coverage Map}
We represent the coverage map $\widehat\chi(T)$ as the union of all coverage 
regions for transmitters in $T$. Each coverage region is represented by a set 
of arc-polygons, each a list $(s_1, s_2 \ldots s_k, s_1)$ of $k$ circular arcs 
forming a connected closed chain.

Though not explicit in the representation, the areas enclosed by these 
chains form the coverage region. In a particular chain $(s_1, s_2 \ldots s_k, 
s_1)$, corresponding to transmitter $\tilde{p}$, there is at most one arc 
corresponding to the rim of the transmission disk for $\tilde{p}$ (curving {\em 
outward}), whereas the remaining arcs correspond to rims of interference disks 
of transmitters interfering with $\tilde{p}$ (curving {\em inward}).
\subsection{Locating an Interference-Bound Transmitter}
A transmitter whose transmission disk intersects with an interference disk of 
another transmitter is called {\em interference-bound}. Assume that such a 
transmitter $\tilde{p}$ exists and has a non-empty coverage region. The coverage 
region for $\tilde{p}$ has at least one inward arc - corresponding to a 
transmitter that intersects with its transmission disk.

Suppose we are given a coverage map $\widehat\chi(T)$, and we want to find 
whether there exists an interference-bound transmitter in $T$. We first test for 
a transmitter with an empty coverage region. We then test each arc in the 
coverage map to check that it is not a complete circle. Thus, given the coverage 
map, one can find an interference-bound transmitter in time linear in the number 
of arcs in the coverage map. We show that the number of arcs is $O(n)$. 
%
\begin{lemma}
If no power region is empty, then the total number of arcs in the coverage map 
is $O(n)$.
\end{lemma}
\begin{proof}
We analyze using the partition of the power region by the power frame. Each 
neighboring pair of transmitters contributes to one partition each in two 
power regions (one for each neighbor in the pair). Thus, the total number of 
partitions is twice the sum of neighboring pairs, i.e $O(n)$.
Each partition corresponds to at most one inward arc and at most one outward 
arc. Thus, the total number of arcs is also $O(n)$.
\end{proof}

Thus, given a coverage map for the above configuration, an interference-bound 
transmitter can be located in $O(n)$ time.
\subsection{A reduction from the $\epsilon$-closeness problem}
The lower bound for the $\epsilon$-closeness problem is a classical problem 
related to many fundamental proximity problems in computational geometry. We 
formally state the problem in Lemma \ref{lemma:eps-closeness} below.
\begin{lemma}[$\epsilon$-closeness]
\label{lemma:eps-closeness}
Consider a real number $\epsilon$ and a sequence $(a_1, a_2, \ldots, a_n)$ of 
$n$ real numbers. Finding whether there exists a pair of real numbers $\{a_i, 
a_j\}$ in this sequence such that $|a_i - a_j| < \epsilon$ is 
$\Omega(n\log{n})$.
\end{lemma}
\begin{proof}
Refer Chapters 5 and 8 in Preparata et al.~\cite{preparata}.
\end{proof}

We now prove our main result by reducing the $\epsilon$-closeness problem to 
computing a coverage map.
\begin{theorem}
The coverage map problem is $\Omega(n\log{n})$.
\end{theorem}
\begin{proof}
The input given to the $\epsilon$-closeness problem is $\epsilon$ and the 
sequence of real numbers $(a_1, a_2, \ldots,$ $\ a_n)$. Given this input we 
construct a set of $n$ transmitters as input for the coverage problem as follows 
- the center of transmitter $\tilde{p}_i$ is $(a_i, 0)$, the transmission radius 
of each transmitter is $\frac{\epsilon}{3}$, and the interference radius of each 
transmitter is $\frac{2\epsilon}{3}$.

Suppose we have the coverage map for this set of transmitters. A transmitter is 
interference-bound only if its transmission disk intersects with some other 
transmitter's interference disk. Given our placement of the disks, this can 
occur only if there exists a pair of transmitters whose centers are located less 
than $\epsilon$ distance apart. The distance between two transmitters, however, 
is the same as the (unsigned) difference between the corresponding real numbers 
in the $\epsilon$-closeness problem. Hence, the existence of an 
interference-bound transmitter implies the existence of a pair $\{a_i, a_j\}$ 
separated by distance less than $\epsilon$. Thus, the coverage map problem must 
be $\Omega(n\log{n})$.
\end{proof}

\chapter{Coverage in the Protocol Model: Dynamic Set of Transmitters}
\label{chapter:dynamicCoverage}
We consider two update operations - addition and deletion of one transmitter,
and propose a method to maintain the coverage map efficiently on each
update. Our purpose is to achieve efficiency comparable to the static algorithm
in Chapter \ref{chapter:Fixed}.

In this chapter we report a randomized algorithm whose efficiency is expressed
as a sum of two components - expected and deterministic. The {\em expected} cost
is an expectation over random choices made in building internal data structures 
- in other words, the expected cost of locating coverage regions affected by 
the update. This cost is $O(\log{n})$; independent of the sequence of updates. 
The other component is a {\em deterministic} cost. Lets say we have found $k$ 
disks whose power regions are affected by a particular update. Our algorithm 
updates the contours of the corresponding coverage regions in $O(k)$.
\section{Our Approach}
\label{section:approachDynamic}
The key ideas we use in our approach are:
\begin{enumerate}
\item {\em Dynamic maintenance of Power Diagrams}
  \begin{enumerate}
    \item {\em Mapping from set of 2-D disks to a 3-D convex polytope} This
    mapping is used for dynamic construction of Voronoi diagrams in Mulmuley 
\cite[Chapters 3 and 4]{Ketan}, following the original idea from
    Brown \cite{BrownMap}. A similar map is used for static construction of
    power diagrams in Aurenhammer \cite {AurenhammerPowerDia}.
    \item {\em Dynamic maintenance of a 3-D convex polytope} We have chosen an
    algorithm from Mulmuley \cite[Chapter 4]{Ketan}. We have adapted
    this algorithm to suit our purpose without compromising on its
    efficiency. This adaptation is discussed later in Section \ref
    {section:dynamicIntersection}.
  \end{enumerate}
\item {\em Dynamic maintenance of disk intersections} We use the new power 
  regions to update disk intersections efficiently.
\end{enumerate}
In Section \ref{section:1-Dvoronoi} we illustrate our ideas using dynamic 1-D
Voronoi diagram maintenance as a conceptual tool. Later, in Section
\ref{section:2-D}, we show how these ideas can be extended to 2-D power
diagrams. Section \ref{section:2-D} shows the use of the new power region to
compute the new coverage regions. In Section \ref{section:2-D} we also present
an analysis of these algorithms. In Section \ref{section:dynamicIntersection}, 
we present the use of the power frame to update disk intersections efficiently. 
Finally, in Section \ref{section:hiddenDisks}, we explore dynamic maintenance 
of the coverage map in presence of transmitters with empty power regions.
\section{Dynamic 1-D Voronoi Diagrams}
\label{section:1-Dvoronoi}
A Voronoi diagram of a set of points on the $x$-axis directly corresponds to
the sorted sequence of these points, ordered in co-ordinate order. In this 
section we discuss a method for dynamic maintenance of a sorted sequence. This 
method lays the groundwork for dynamic maintenance of power diagrams in 2-D.

We will follow a randomized model (as in \cite{Ketan}, \cite{Seidel}), and show 
in Section \ref{section:2-D} that its concepts extend coverage maps as well. 
The performance guarantees in this model are probabilistic - the {\em
expected} cost per addition (deletion) of one point to (from) a set of $n$
points is $O(\log{n})$.
\subsection{The randomization model}
The purpose of the model is to give probabilistic guarantees on the operations
of addition and deletion of one point from a sorted sequence. Each inserted
point is assigned a unique random {\em priority} from a uniform probability
distribution on the interval $(0, 1)$. The data structure maintained is a
binary search tree on the co-ordinate order with the min-heap property on the
priority order. We use the term {\em treap} for this data structure, following
Seidel et al.~\cite{Seidel}. Each point corresponds to a node in the tree. Each 
node is the root of a subtree, with left and right subtrees below it. The 
coordinate order of a node is lower than all nodes in its left subtree, and 
higher than all nodes in its right subtree. The priority of a root node is the 
lowest in its subtree.

Figure \ref{fig:treap} shows a treap of 6 points with co-ordinate values $\{10,
30, 50, 70, 90, 110\}$. The corresponding priority values are given by the
function $\{(10, 0.3),$ $(30, 0.13),$ $(50, 0.4),$ $(70, 0.22),$ $(90, 0.56),$
$(110, 0.43)\}$.
%
\begin{figure}
\centering
\includegraphics[width=5in]{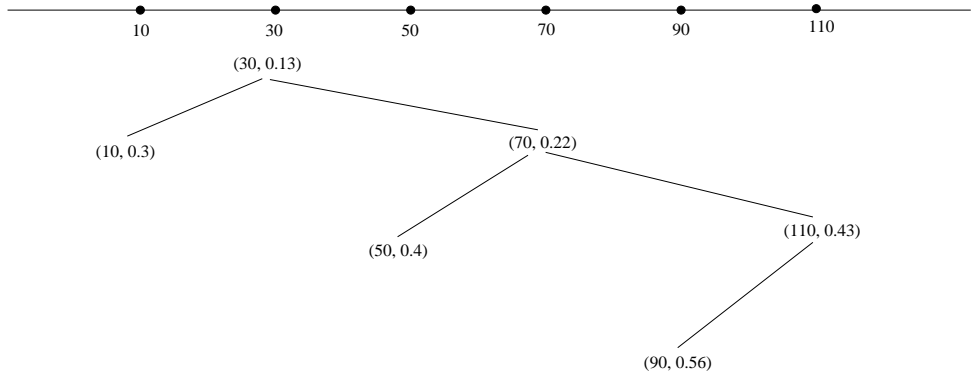}
\caption[Treap]{A treap for points $\{10, 30, 50, 70, 90, 110\}$}
\label{fig:treap}
\end{figure}

Addition of a point requires treap properties to be maintained. This 
maintenance is effected by balancing rotations. A new node is first added at 
the leaf position corresponding to the point's position in the co-ordinate 
order. If the node's priority is higher than its parent's, then the treap is 
correct. Otherwise, the node's position is `rotated' with its parent's to 
preserve the co-ordinate order. This process is repeated until the treap 
property is satisfied. Deletion of a point is analogous - it is performed in 
reverse order. Figure \ref{fig:treapadd} shows the addition of $(130, 0.2)$ by 
rotations to the treap in Figure \ref{fig:treap}.
%
\begin{figure}
\centering
\includegraphics[width=5in]{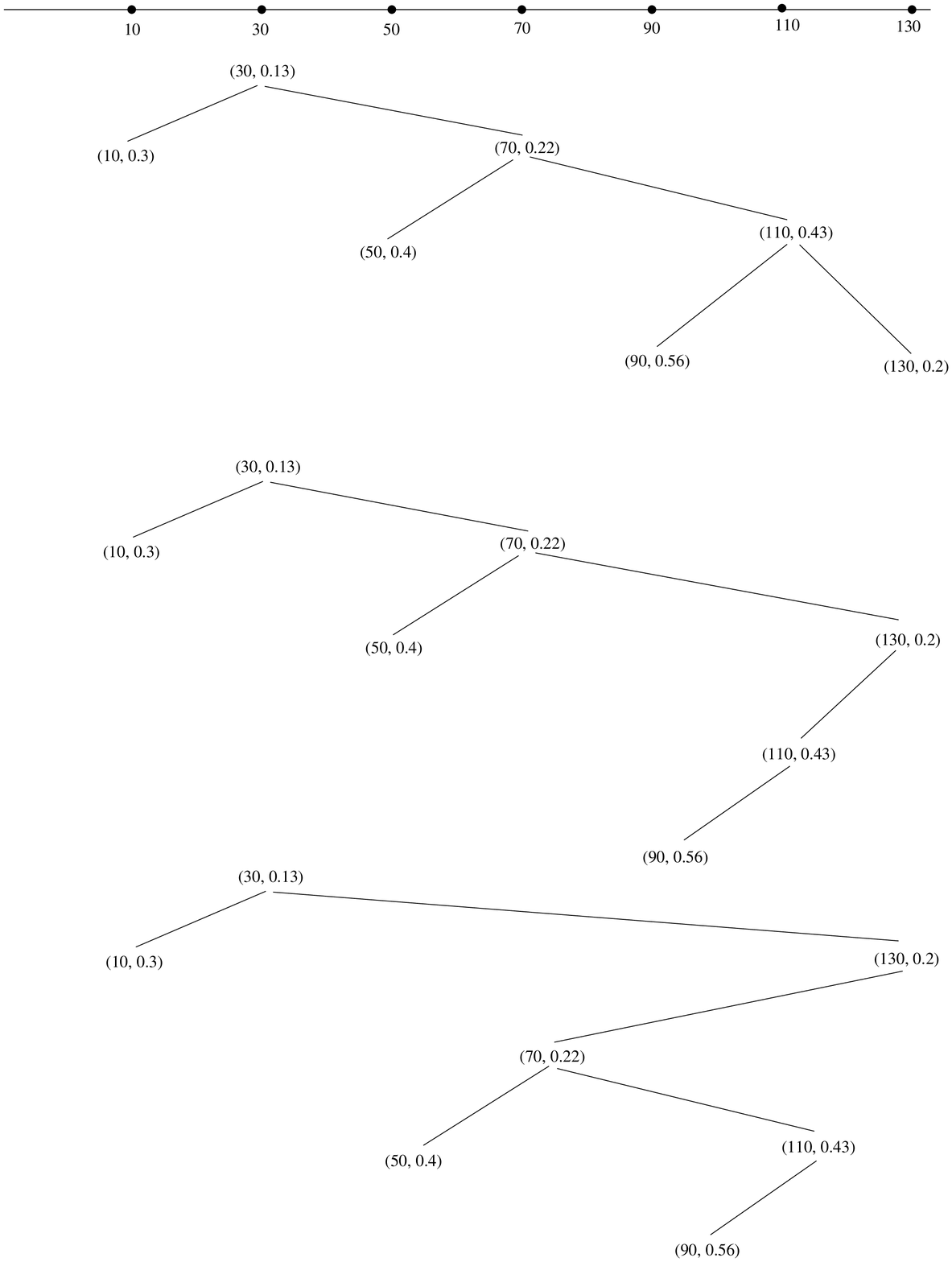}
\caption[Adding a point to a treap]{Adding $130$ to the set $\{10, 30, 50, 70, 90, 110\}$}
\label{fig:treapadd}
\end{figure}
\subsection{Expected cost of updates}
The cost of addition or deletion consists of four components -
\begin{enumerate}
\item The cost of locating the point to delete (or locating the position of the
  point being added),
\item The cost of rotations,
\item The cost of updating the sequence, and
\item The number of random bits used to distinguish the priority of a new node.
Note that a prefix of most significant bits of the priority is sufficient to
distinguish a node's priority.
\end{enumerate}
%
\begin{figure}
\centering
\includegraphics[width=5in]{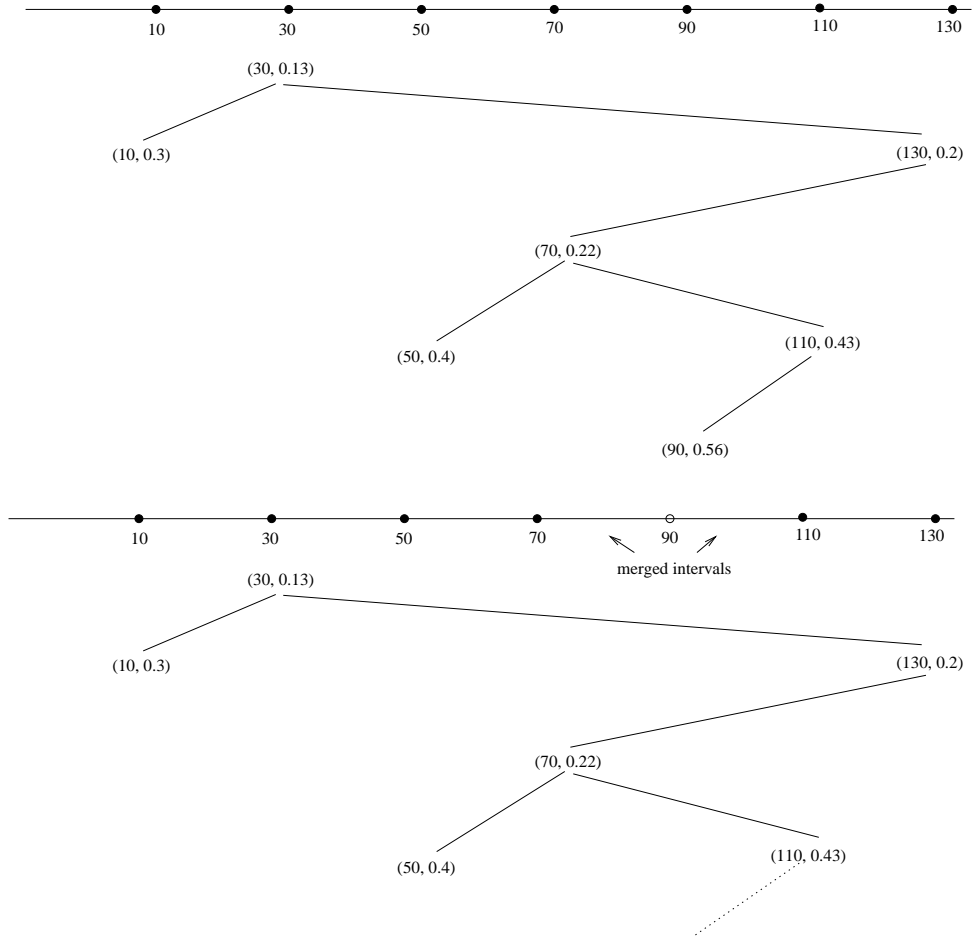}
\caption[Deleting a point from a treap]{Merged intervals after deleting $90$ from the set $\{10, 30, 50, 70, 
90, 110, 130\}$}
\label{fig:treapdelete}
\end{figure}
An analysis of the expected costs of these operations on the treap is
presented in \cite{Seidel}. We summarize the results in Table
\ref{table:expectation}.
\begin{table}
\centering
\begin{tabular}{|p{5in}|p{1in}|}
\hline
Cost of Locating the point to delete or the position of the point being
added & $O(\log{n})$ \\
\hline
Cost of rotations required per update & $O(1)$ \\
\hline
Cost of updating the sequence & $O(1)$ \\
\hline
Number of bits of priority used for a new point & $O(1)$ \\
\hline
\end{tabular}
\caption{Expected Costs for Treap Operations \label{table:expectation}}
\end{table}
\subsection{From 1-D to 2-D}
In subsequent sections we will show extensions of these concepts to higher
dimensions. The partition of the line imposed by the set of points may be
viewed as a set of intervals. Addition of a point splits the interval to which
the point belongs into two adjacent intervals; while a deletion merges the two
intervals adjacent to the deleted point. In the next section, we use a 2-D
extension of the interval - a 2-D convex region. We extend to 2-D the concepts
of random priorities, imaginary sequence of additions in priority order, and
rotation in the data structure to maintain this sequence.
\section{Dynamic 2-D Power Diagrams}
\label{section:2-D}
We now show an extension of the dynamic methods for 1-D Voronoi diagrams to 2-D
power diagrams. Table \ref{table:1D-to-2D} shows the correspondence of some
concepts between dimensions.
\begin{table}
\centering
\begin{tabular}{|p{2.5in}|p{3.5in}|}
\hline
{\bf 1-D} & {\bf 2-D} \\
\hline
Partition by Intervals & Partition by Power Regions \\
\hline
Locate intervals affected by point to delete & Locate partitions affected by 
disk to delete \\
\hline
Merge intervals after deletion & Merge neighboring Power Regions after
deletion \\
\hline
Locate interval for point being added & Determine Power Regions changed by disk
being added \\
\hline
Split interval after addition & Split neighboring Power Regions after addition
\\
\hline
\end{tabular}
\caption{Correspondence between Partitions in 1-D and 2-D
\label{table:1D-to-2D}}
\end{table}
We extend the treap data structure and its analysis. We first show in
Subsection \ref{subsec:struct} that the cost of merging and splitting
partitions in 2-D is $\Omega(n)$. We call this the {\em cost of structural
change}. In Subsection \ref{subsec:dynamicpolytope} we present an extension of the
treap and show that the expected cost of locating partitions is $O(\log{n})$.
\subsection{Cost of Structural Change in 2-D Voronoi Diagrams}
\label{subsec:struct}
%
\begin{figure}
\centering
\includegraphics[width=5in]{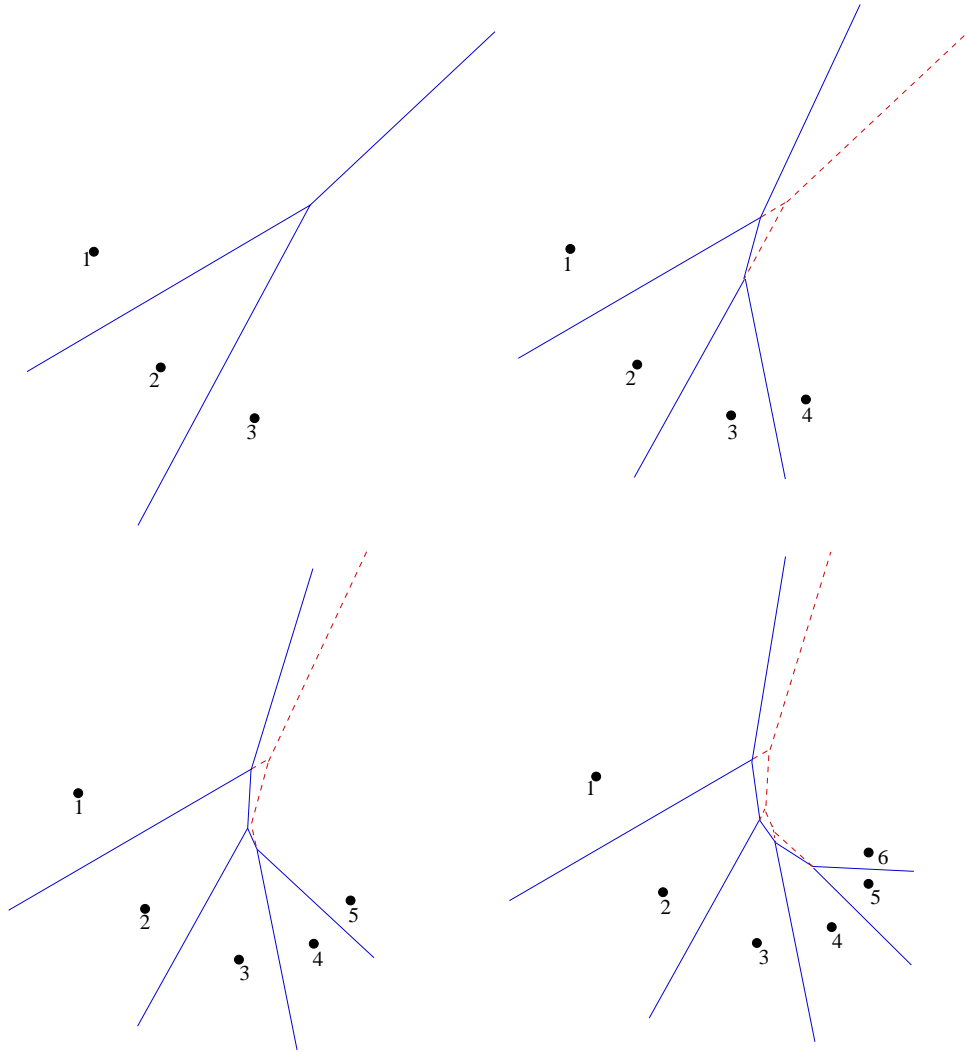}
\caption{Worst-case Sequence of Additions to a Voronoi Diagram}
\label{fig:Voronoilowerbound}
\end{figure}
A lower bound for the performance of any dynamic algorithm is the amount of
structure update in the output per addition or deletion. In 1-D, the structural
change is the merging or splitting of one interval, which is $O(1)$. 1-D
intervals extend to convex polytopic partitions in 2-D. We show (Lemma
\ref{lemma:Voronoilowerbound}) that for 2-D Voronoi (and thereby, power)
diagrams this lower bound is the number of existing partitions -
i.e. $\Omega(n)$.
\begin{lemma}
\label{lemma:Voronoilowerbound}
There exists a sequence of addition of points in an online construction of a
Voronoi Diagram such that the amortized cost per addition is $\Omega(n)$.
\end{lemma}
\begin{proof} See Figure \ref{fig:Voronoilowerbound}. \end{proof}
Figure \ref{fig:Voronoilowerbound} shows points being added from left to right
to a Voronoi diagram; each point is inside the circumcircle of the three points
succeeding it. Each addition affects all existing regions in the
partition. This sequence of points demonstrates that there are sequences of 2-D
inputs for which the structure change for each update is $\Omega(n)$. In
contrast, the cost of structure update in 1-D ($O(1)$) is independent of the
sequence in which points are added.

If the input sequence is {\em random}, however, the expected structure update
cost for every addition is $O(1)$ (see \cite{Seidel}). The algorithm
randomizes choices made in maintaining an internal data structure. We report
the efficiency of our algorithm as a sum of two components - deterministic and
expected. The {\em expected} cost of locating one region affected by the update
is $O(\log{n})$. This cost is independent of the input sequence. It is an
expectation over random choices for the data structure. The other component is
the {\em deterministic} cost of the structure change. If $k$ is this number of
disks whose power regions are affected by a particular update then our
algorithm updates the power regions, and corresponding coverage regions in
$O(k)$.
\subsection{Mapping between 2-D Power Diagrams and 3-D Convex Polytopes}
\label{subsection:mapping}
An upper convex polytope in 3-D is the intersection of a set of 3-D half-spaces
containing the point $(0, 0, \infty)$.  The method we propose for dynamic
updates to power diagrams uses a mapping from 2-D disks to 3-D upper convex
polytopes.

Consider a set of disks $\{C_1, C_2, \ldots, C_n\}$ in the $xy$-plane. Let
$r_i$ be the radius, and $(x_i, y_i)$ be the center of disk $C_i$. Disk $C_i$
is mapped to the half-space $S_i \equiv z \geq 2 x x_i + 2 y y_i - x^2_i -
y^2_i + r^2_i$. The intersection of half-space $S_i$ with the paraboloid $z =
x^2 + y^2$ projects down to $C_i$.

This map is illustrated in Figure \ref{fig:map1}. The intersection of two
disks, and their mapping to two intersecting 3-D half-spaces is shown in Figure
\ref{fig:map2}.
\begin{figure}
\centering
\includegraphics{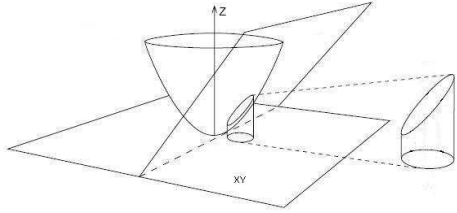}
\caption[Mapping a Disk to a Half-Space]{Mapping a Disk to a Half-Space (courtesy \cite{AurenhammerNotes})
\label{fig:map1}}
\end{figure}
\begin{figure}
\centering
\includegraphics{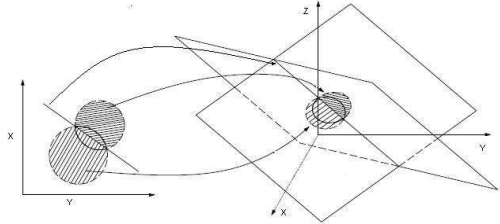}
\caption[Mapping intersections: disks to half-spaces]{Two Intersecting Disks Map to Intersecting Half-Spaces (courtesy
\cite{Molecules})\label{fig:map2}}
\end{figure}
The upper convex polytope formed by the intersection of the half-spaces $\{S_1,
S_2,$ $ \ldots, S_n\}$ is dual to the power diagram of the disks $\{C_1, C_2,
\ldots, C_n\}$. This claim is proved in \cite{AurenhammerPowerDia}.

Figure \ref{figure:UnequalBadCase} shows a power diagram containing disks with
empty power regions. The mapping described here leads to a characterization of
such disks - in terms of half-spaces we call {\em redundant}. A half-space $S$
is redundant if the convex polytope formed by $\{S_1, S_2, \ldots, S_n\}$ is
not changed by the addition of $S$. The mapping described here implies a
duality between redundant half-spaces and disks with empty power regions.

Dynamic maintenance of the 2-D power diagram during addition or deletion of a
disk corresponds to dynamic maintenance of the dual 3-D convex polytope during
addition or deletion of the mapped half-space. The 2-D partition into power
regions is obtained by projecting down from the 3-D convex polytope. This
projection maps the 2-D faces of the 3-D convex polytope to power regions in
2-D.

In the remainder of this section we do not explicitly state our arguments in
the context of 2-D power regions. Instead, we formulate our arguments in terms
of faces of the corresponding convex polytope.
\subsection{Dynamic Maintenance of a 3-D Convex Polytope}
\label{subsec:dynamicpolytope}
The best known dynamic 3-D convex polytope maintenance algorithm is by Chan
\cite{Cha06}. However, this algorithm does not construct the new faces
explicitly. We choose an approach described in Mulmuley 
\cite[Chapter 4]{Ketan}). This algorithm constructs the new faces (i.e.~2-D 
partitions) which we require to compute disk intersections. Our approach, 
however, requires special treatment of redundant half-spaces; Chan's algorithm 
does not. We discuss redundant half-spaces in more detail in Section
\ref{section:hiddenDisks}.

Our approach for dynamic maintenance of a 3-D convex polytope follows that of
Mulmuley \cite[Chapters 3 and 4]{Ketan}, which presents an algorithm
for online construction of a convex polytope. We will refer to the
analysis of this algorithm in the analysis of the dynamic algorithm. We do not,
however, reproduce the online algorithm or its analysis in this report.
\subsubsection{Facial Lattice of a Convex Polytope}
The `facial lattice' of a convex polytope is the adjacency relation between its
vertices, edges, and faces.  Each addition of a half-space causes a change to
the facial lattice. Vertices, edges, and faces that do not belong to the
half-space are removed. New vertices, edges, and faces are created
corresponding to the intersection of the added half-space to the existing
polytope. This operation may be viewed as splitting the faces of the polytope
intersecting with the new half-space; the portion of the split polytope outside
the half-space is removed, and a new face corresponding to the boundary of the
intersection is created.

The size of the facial lattice of a convex polytope constructed with $i$
half-spaces is $O(i)$. $i$ half-spaces contribute to at most $i$ faces. Using
this in Euler's formula, we can show a bound of $O(i)$ on the number of edges
and vertices as well.
\subsubsection{Shuffle - A Randomized Data Structure}
We adopt the terminology {\em Shuffle} from Mulmuley \cite{Ketan}. Each added
half-space is assigned a random priority from the interval $(0,1)$. We begin
the description of Shuffle by assuming that half-spaces are added in order of
increasing priority. Later, we describe the operations of adding and deleting a
half-space from an arbitrary position in the order. We denote Shuffle by
$\mathbb{S}$.

We further assume that the bounded edges and faces of the polytopes we build
are included in a large 3-D cube $\mathbb{C}$. This assumption serves only to
simplify our description, and we later show how it may be removed.

$\mathbb{S}$ contains nodes for vertices and half-spaces, and represents a
relation between these nodes. We will describe these nodes and relation
shortly. $\mathbb{S}$ has eight {\em root} vertex nodes. These root nodes
correspond to the vertices of $\mathbb{C}$.

When a half-space is added to a polytope, a new face and some new vertices are
created. Some edges are split, and some edges are deleted. $\mathbb{S}$
records these actions as follows:
\begin{description}
\item{$vertex(v)$}: vertex corresponding to the vertex node $v$
\item{$created(v)$}: half-space that creates vertex $v$
\item{$next(u)$}: for edge $\{u, v\},\ u\not\in S$ split at $w$, this the vertex
$w$
\item{$prev(w)$}: for edge $\{u, v\},\ u\not\in S$ split at $w$, this the vertex
$u$
\item{$createdFace(S)$}: list of vertices created by $S$, in order of face
traversal
\item{$deletedEdges(S)$}: list of edges deleted by $S$, in order of traversal
\end{description}
The first half-space added, $S_1$, splits some faces of $\mathbb{C}$. The
resultant polytope is $S_1 \cap \mathbb{C}$. Vertex nodes corresponding to the
new vertices are added to $\mathbb{S}$. The relation described by $created,\
next,\ prev,\ createdFace,$ and $deletedEdges$ is created.

This process is repeated for each addition. Consider the addition of (a
non-redundant) half-space $S$ to polytope $H$. The resultant polytope is $S\cap
H$. New nodes that correspond to the new face are created in $\mathbb{S}$. The
relations $created$ and $createdFace$ are updated accordingly. The relations
$next$ and $prev$ are updated to reflect the split of each edge. Edges deleted
are added to the list $deletedEdges$.

Each vertex node corresponds to a vertex in some polytope $\supseteq H$. $prev$
is not defined for root nodes (corresponding to vertices of $\mathbb{C}$). A
vertex node $v$ for which $next(v)$ is not defined is a {\em leaf} vertex
node. Non-leaf vertex nodes correspond to polytopes constructed from
half-spaces with priority less than $S$. We use $next$ and $prev$ to traverse
from root to leaf, as described below.
\subsubsection{Half-Space Intersection using Shuffle}
In order to compute the intersection of a half-space $S$ with a polytope $H$,
we first locate one intersection point, and then the rest of the intersecting
face by traversal on other faces of $H$. We first show the use of $\mathbb{S}$
to locate an intersection point.
\begin{algorithm}[TraverseShuffle(S)]
\label{algo:search}
\end{algorithm}
\begin{enumerate}
\item Find a root node $u$ such that $vertex(u) \in \mathbb{C}\setminus S$
\item \label{item:while-trav} $v \leftarrow next(u)$
\item if $v = \phi$ return $u$ // reached vertex $u$ of $H$ 
\item Find a vertex $w$ on $createdFace(created(v))$ such that $w \not\in S$
\item if such a $w$ exists, then $u \leftarrow w$; repeat \ref{item:while-trav}
\item if no such $w$ exists, then return $\phi$ // $S$ is a redundant half-space
\end{enumerate}
We use the routine {\em TraverseShuffle} to locate a vertex $v\in H\setminus
S$. Since $H\setminus S$ is a connected convex polytope, we can traverse all
faces either intersecting $S$ or outside $S$ starting from the adjacencies of
$v$. We construct the new polytope on this second traversal. If we find a face
that is outside $S$, we remove it from the facial lattice. If we find a face
that intersects $S$, we traverse that face on the path defined by the vertex
adjacencies. During this traversal, we remove edges that lie outside $S$, and
split edges that intersect $S$. A new edge is added between the two new
vertices on this face. After traversing one face, we move to a face with a
common edge that is either outside $S$ or intersects $S$. {\em Doubly-Connected
Edge List} - a data structure for maintaining the facial lattice efficiently is
described in the book by de Berg et al.~\cite{DCEL}.
\subsubsection{Analysis - TraverseShuffle and Intersection}
We first prove the correctness of {\em TraverseShuffle}; that is it returns
$\phi$ if and only if $S$ is redundant for $H$.
\begin{lemma}
\label{lemma:search}
TraverseShuffle {\em returns $\phi$ if and only if $H\subset S$}
\end{lemma}
\begin{proof}
If $x\not=\phi$ is returned, then the path followed by the
algorithm ensures that $x\in H\setminus S$, thus $H\not\subset S$.

Suppose $\phi$ is returned. Let $u$ be the last vertex node visited; and $H_u$
be the corresponding polytope when $u$ was created. Let $v = next(u)$; and
$H_v$ be its corresponding polytope. We can imagine the traversal from $u$ to
$v$ as a traversal from the outer surface of $H_u$ to the outer surface of
$H_v$. Since no vertex on the face creating $v$ is outside $S$, no point in
$H_v$ is outside $S$ either. Since $H_v \supset H$, $H\subset S$.

Thus, $\phi$ is returned if an only if $H \subset S$.
\end{proof}
We use the online algorithm for polytope construction in \cite{Ketan}. This
algorithm is analyzed probabilistically, assuming that the input sequence of
half-spaces was drawn from a uniform random distribution. So far in this
presentation we have assumed that additions are made in increasing order of
priority. Thus, this model and assumption are the equivalent to the online
algorithm from \cite{Ketan}.

The data structure from \cite{Ketan} - {\em history} - is built on the same
principles as $\mathbb{S}$. Thus, the analysis of {\em TraverseShuffle} is the
same as that of search in {\em history}. {\em TraverseShuffle} thus performs in
expected $O(\log{n})$ time.

The method of half-space intersection is also carried over from \cite{Ketan}. 
We have only used a random sequence as a conceptual tool to aid the analysis so
far; whereas the algorithm input may be an arbitrary sequence. Hence, the cost
of an addition is deterministic. If $k$ faces intersect with $S$, $O(k)$ is the 
size of the facial lattice of $S \cap H$. Thus, the addition of $S$ costs 
$O(k)$.

We now show the addition of a half-space $S$ with a random priority $p$. Unlike
the assumption so far, this priority is not necessarily higher than all existing
priorities.

Our approach is similar to Section \ref{section:1-Dvoronoi}. We first add the
half-space to the end of the sequence, and then move it up the sequence using
{\em rotations} in $\mathbb{S}$. We will see that these rotations are analogous
to the rotations described for the {\em treap} in Section
\ref{section:1-Dvoronoi}.

Note that the output polytope is independent of the sequence of additions; our
purpose of rotations is only to manipulate the data structure. This
manipulation is required to maintain the $O(\log{n})$ expected search guarantee
shown earlier. The purpose of the rotation is to maintain the data structure
such that it appears as if all half-spaces were added in increasing priority
order.
\begin{figure}
\centering
\includegraphics{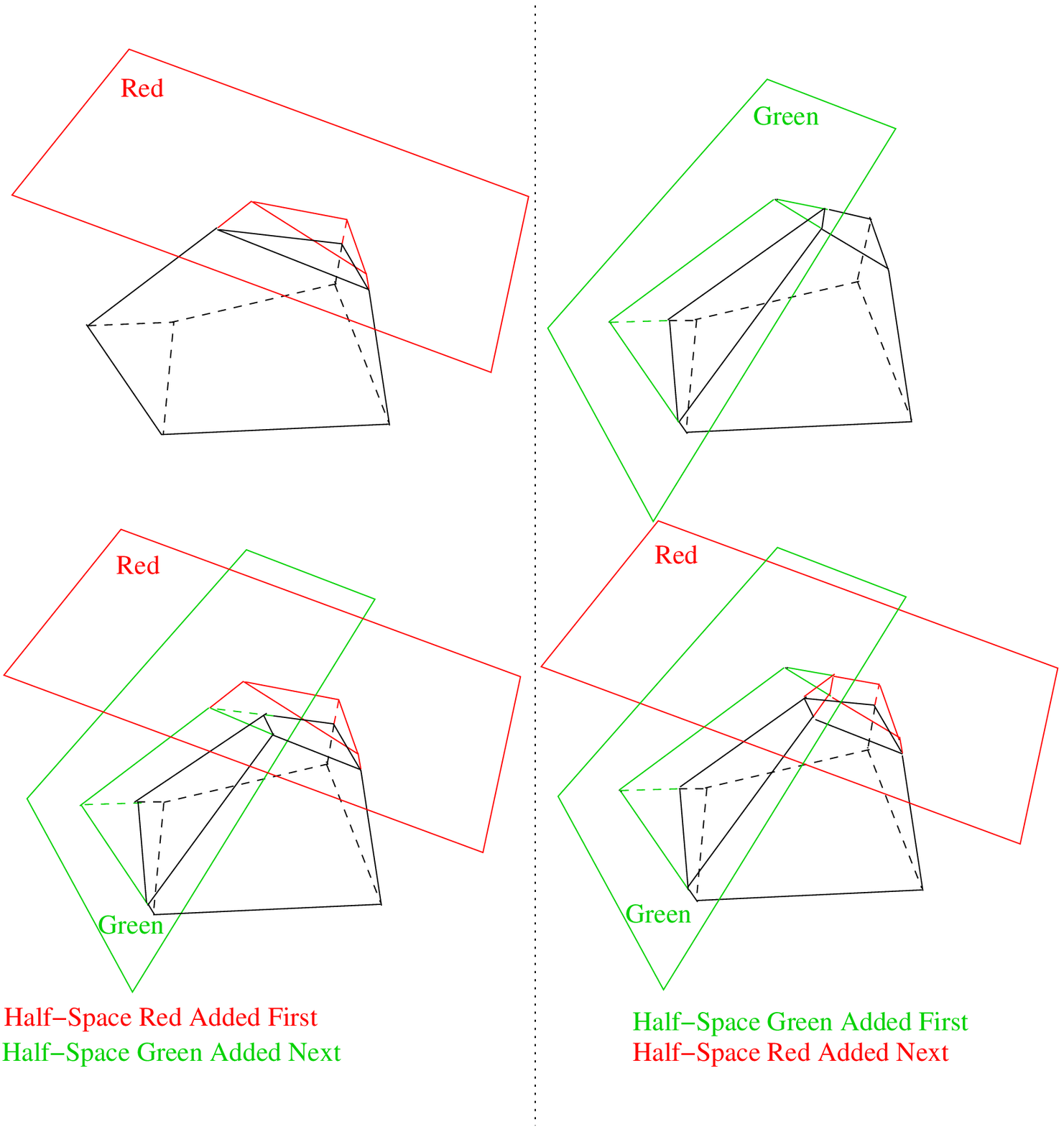}
\caption{Addition of Half-Space to Polytope\label{figure:polytope}}
\end{figure}
\begin{figure}
\centering
\includegraphics{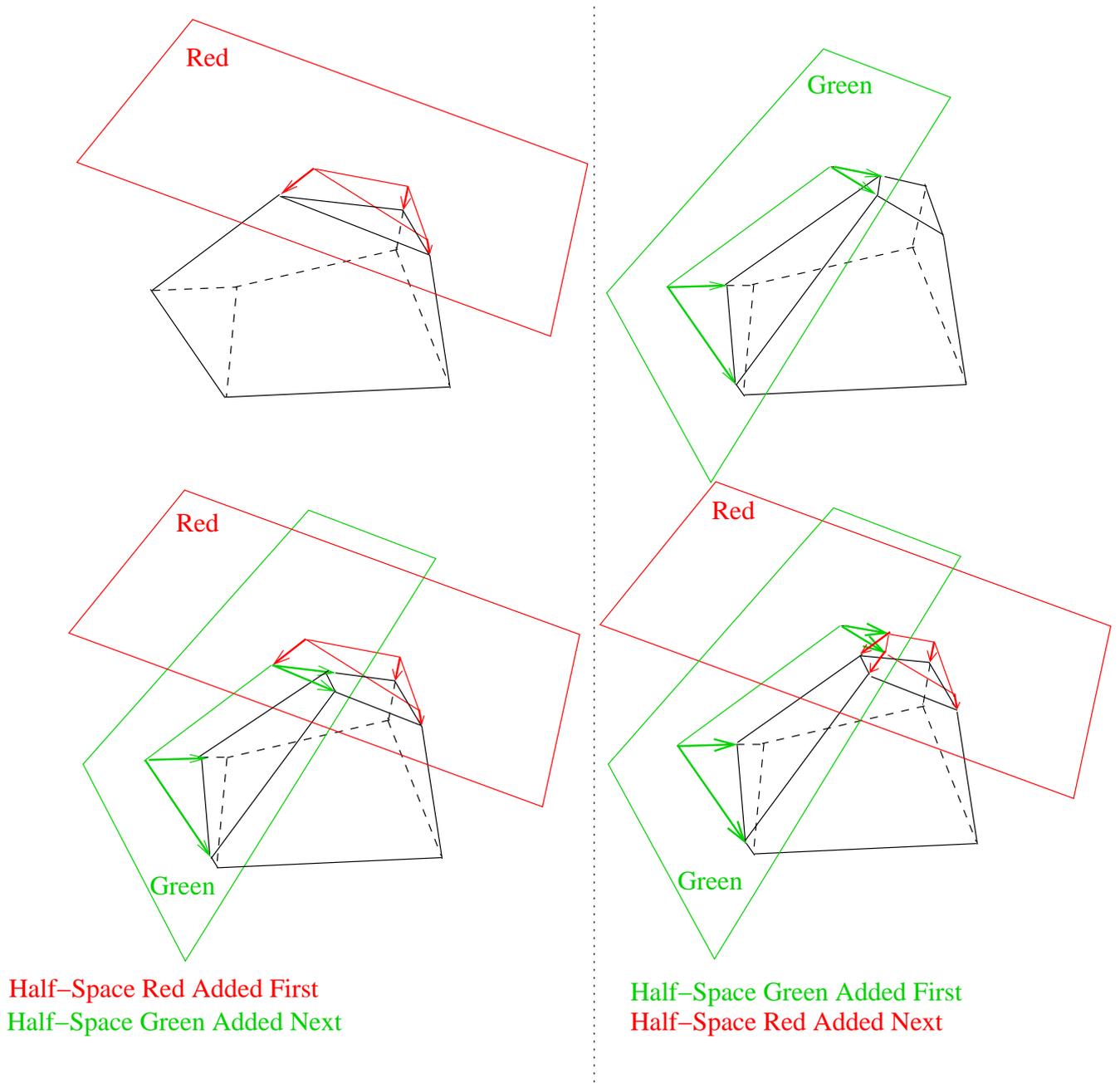}
\caption{Data Structure changes for Addition\label{figure:search-structure}}
\end{figure}
\begin{algorithm}[Rotate-Addition]
\label{algo:rotate-add}
\end{algorithm}
\begin{enumerate}
\item $L \leftarrow \phi$ // $L$, set of half-spaces
\item for each $v \in createdFace(S)$, do
  \begin{enumerate}
  \item $S' \leftarrow created(prev(v))$
  \item add $S'$ to $L$
  \end{enumerate}
\item\label{loop-heap-trav} for the highest priority half-space $S'\in L$, do
  \begin{enumerate}
  \item $S' \leftarrow$ highest priority half-space in $L$
  \item traverse $createdFace(S')$ and $deletedEdges(S')$ to build polytope
  $\mathbb{A}=S'\setminus H$
  \item traverse $\mathbb{A}$ to build polytope $\mathbb{A'} \leftarrow
  \mathbb{A}\setminus S$
  \item use $\mathbb{A'}$ to update relations for $S$ and $S'$ // switch
  order of intersection
  \item $L \leftarrow L\setminus \{S'\}$
  \end{enumerate}
\end{enumerate}
Figure \ref{figure:polytope} shows the addition of two distinct half-spaces to 
a 3-D polytope. Since set intersection is commutative, the order of addition 
does not change the resultant polytope. However, the facial lattice changes are 
distinct. Figure \ref{figure:search-structure} shows the two different search 
structures arising from the two sequences.

This algorithm is adapted from Mulmuley \cite[Chapter 4]{Ketan}. Our algorithm
is meant for use only with non-redundant half-spaces. It is the 2-D extension 
of treap rotation in Section \ref{section:1-Dvoronoi}. The operations Algorithm 
\ref{algo:rotate-add} performs on the facial lattice for addition of half-space 
{\it\textcolor{green}{Green}}, which has a higher-priority than half-space 
{\it\textcolor{red}{Red}}, is depicted in Figure \ref{figure:search-structure}. 
The relevant edges in the facial lattice are marked with arrow-tips. The 
sequence on the left shows the facial lattice before Step \ref{loop-heap-trav}, 
and the sequence on the right shows the changes to the facial lattice after 
Step \ref{loop-heap-trav}. The cost of each iteration of Step 
\ref{loop-heap-trav} is $O(|A'|)$. The sum over all iterations is $O(k)$, where 
$k$ is the structure change due to addition of $S$.
\subsubsection{Deleting a Half-Space}
Deletion of a half-space $S$ from $H$ first requires a search for the
corresponding face on $H$, and then updates for the part of the facial lattice
that changes. The delete operation is the reverse of addition.

We first locate the corresponding face using $\mathbb{S}$. The traversal
technique is the same as {\em TraverseShuffle}. Then we increase the priority
of $S$ beyond that of its neighbors. This makes all vertices of
$createdFace(S)$ leaves.

The update to $H$ is then removal of the face for $S$, and extension of
neighboring faces using $deletedEdges(S)$ and $prev(v)$ for each vertex $v$ on
$createdFace(S)$. Note that during addition, this update is performed in
reverse - the new half-space is added as the highest priority order and moved
up the priority order to fit its final position. Rotations are required here
too to increase the priority for $S$. These too, are the reverse of rotations
for addition.

\subsubsection{Choice of Method}
The choice of this particular method of random priorities with rotation is
primarily motivated by the ready availability of the power frame during
addition and deletion. We show this property in the next section and use
this for disk intersections.

\subsection{Open Problems}
Analyses for addition and deletion are given in Mulmuley \cite{Ketan} that
shows $O(\log{n})$ expected time. However, the analysis requires assumption of 
a random sequence. The cost of rotations in a random sequence is $O(1)$. 
Additionally, this result is derived as a corollary of results on configuration 
spaces. We conjecture that this analysis can be simplified to $O(\log{n})$ 
expected search and rotation time under our assumption of non-redundant 
half-spaces.

We assumed a bounding cube at the beginning of this section. Removing this
assumption is possible by assigning non-numerical symbolic values to vertices 
on $\mathbb{C}$ and its split edges.
\section{Dynamic Disk Intersections}
\label{section:dynamicIntersection}
Updating a coverage map for adding a transmitter requires two sets of
operations. The first set of operations compute the coverage region for the new
transmitter. The second set of operations compute changes in the coverage
regions of the new transmitter's neighbors in the power diagram. We show
dynamic methods for both sets of operations. Similar operations are required to
delete a transmitter.

Algorithm \ref{algo:coverage-map} uses the power frame for a transmitter to
compute its coverage region. The power frame helps us limit the number of arc
intersections for each interfering neighbor. Two arc intersections are
required per neighbor. For the same reason, we use the power frame for dynamic
disk intersection as well.
\subsection{Computing the power frame}
Our method computes the power frame of an added or deleted transmitter in
linear time. The addition of a half-space $S$ begins at the end of the priority
order. Later rotations fix the data structure $\mathbb{S}$ to reflect the
correct priority of $S$. We use this initial addition, before rotation, to get
the power frame. Similarly, the last stage of deletion is used to get the power
frame of the transmitter to delete.

In Algorithm \ref{algo:coverage-map}, we compute the power diagram of
neighboring disks to obtain a transmitter's power frame. This computation takes 
$O(k\log{k})$ time for $k$ neighbors. In $O(k)$ time, however, we can traverse 
$createdFace(S)$, $prev(v)$ for each vertex $v$ on $createdFace(S)$, and 
$killedEdges(S)$. This traversal gives the power frame, just like the traversal 
to obtain $\mathbb{A}$ in Algorithm \ref{algo:rotate-add}.

\subsection{Using the Power Frame}
The coverage region for a transmitter is represented by arc polygons (see
Figures \ref{figure:MediumCircles}, \ref{figure:LargeCircles}, and
\ref{figure:SmallCircles}). The coverage region for the added transmitter is
computed exactly as in Step \ref{loop:parts} of Algorithm
\ref{algo:coverage-map}.  An interference arc corresponding to a neighbor is
removed from each power partition, and the transmission arc is added.

The coverage region for neighboring transmitters must be updated to remove
arcs corresponding to the new transmitter. Consider the addition of a new
transmitter $\tilde{p}$, and a power neighbor $\tilde{q}$. We only have the
power frame for $\tilde{p}$, and not $\tilde{q}$. Thus, the method in Algorithm
\ref{algo:coverage-map} cannot be used.

We use the power frame of $\tilde{p}$ for this update. We also use additional
information in the power diagram. Each edge of the power diagram corresponds to
two neighbors. With each edge we maintain one pointer to the interference arcs
(if any) for the two neighbors, and one pointer to their transmission arcs. 
Thus, given the edges of power frame for the new transmitter, we can lookup the 
interference arcs of old neighbors that need to be removed. Only the arcs that 
must be removed are then traversed in sequence in the coverage region. This 
method takes $O(k)$ time for updating $k$ neighboring coverage regions.

Deleting a transmitter requires only the update of its neighboring coverage 
regions, since removal of the deleted transmitter is trivial. Here too, we use 
the pointers to the transmission and interference disks, as outlined in the 
addition procedure. Each neighbor's interference arc is added, and the deleted 
transmitter's transmission and interference arcs are removed in sequence. This 
operation is $O(k)$ in general, the only special condition being when 
degenerate arcs (i.e. with one point) result. This condition can be avoided by 
the ``standard assumption'' that after any (addition or deletion) operation the 
underlying polytope does not contain a subset of four half-spaces intersecting 
in one vertex.
\section{Hidden Disks and Redundant Half-Spaces}
\label{section:hiddenDisks}
We call disks with empty power regions {\em hidden disks}. As remarked in
Subsection \ref{subsection:mapping}, these disks are dual to redundant
half-spaces in the 3-D upper convex polytope. Claim \itemref{1} implies 
that any coverage point must have a corresponding power region. Thus, disks 
with empty power regions contribute only to interference in the coverage map, 
and do not contribute to coverage.

3-D convex polytope construction is dual to 3-D convex hull construction. In 
this duality, points inside the convex hull map to redundant half-spaces. The 
mapping we use implies that hidden disks (in 2-D) map to points inside the 3-D 
convex hull. Thus, in a {\em random} deployment of disks, it is highly probable 
that a disk will be hidden (i.e. that the mapped 3-D point will lie inside the 
3-D convex hull). However, in an {\em arbitrary} or planned deployment by 
software or a network designer, we don't expect many redundant disks, since 
these do not contribute positively to coverage.

The efficiency of maintaining a dynamic 3-D convex hull depends on structure, as
shown by the mapping from our lower bound in Section \ref{section:2-D}. However,
even the problem of maintaining a location structure for $O(\log{n})$ query to
decide whether a given point lies in the convex hull is open (see Demaine et
al.~\cite{3d-hull}). The best known algorithm, by Chan \cite{Cha06}, is
polylogarithmic.

Our method gives $O(\log{n})$ performance simply because the deletion of any
half-space does not 'expose' a redundant half-space. Thus, we do not require to
look through redundant half-spaces that have become non-redundant following a
deletion. A data structure to maintain this lookup information is apparently
hard to develop, hence the open problem.

We can, however, maintain a relation between hidden disks and power
regions. Thus, during removal of a non-redundant disk, we lookup this relation
to check whether the affected power regions contained a hidden disk that has now
become `visible'. We claim that this method requires $kO(\log{n})$ time for
update to a power region affecting $k$ redundant disks.

The union of all disks remains the same, regardless of the presence of hidden
disks. We observe that dynamic maintenance of coverage without considering
interference (sensor coverage, as in \cite{Ye}) is still possible (by suitably
modifying our methods) in $O(\log{n})$ time per update.

\chapter{Coverage in the SINR Model}
\label{chapter:SINR}
In Chapters \ref{chapter:Fixed} and \ref{chapter:dynamicCoverage}, we have 
shown algorithms in the protocol model that report and maintain the boundary of 
the wireless coverage map. This chapter shows that the partitioning model 
extends, by appropriate extensions to the distance measure, to the SINR model. 
We also show that the boundary of the coverage cannot neither be computed 
efficiently nor represented efficiently in a data structure. In this sense, 
this chapter is exploratory and lays a foundation for Chapter 
\ref{chapter:Optimization}.

The SINR (Signal-to-Interference-plus-Noise-Ratio) model includes the effects 
of path-loss and aggregated interference from all transmitters into the 
coverage decision. Formally, the SINR model is defined as follows:\\
$T$: a set of transmitters\\
$P_{\tilde{t}}>0$: Transmit power of transmitter $\tilde{t}\in T$\\
$N_0>0$: Ambient noise power\\
$d(x, \tilde{t})$: Distance of point $x$ from transmitter $\tilde{t}$\\
$\alpha \ge 2$: Path-loss exponent\\
$SINR(x, \tilde{t}) = \frac{\displaystyle\frac{P_{\tilde{t}}}
{d(x,\tilde{t})^\alpha}}
{\displaystyle\sum_{\tilde{u}\in T\setminus\tilde{t}}
\frac{P_{\tilde{u}}}{d(x,\tilde{u})^\alpha} + N_0}$\\
$\beta > 0$: Receive sensitivity\\
Point $x$ is said to be in coverage if $\exists\tilde{t}\in T$ such that
$SINR(x, \tilde{t}) > \beta$.

In this work, we study coverage in SINR for the following model parameters:
\begin{itemize}
\item Fixed 2-D transmitter locations
\item Fixed values for $\alpha$, $\beta$, $N_0$, and $P_{\tilde{t}}\ \forall
\tilde{t}\in T$
\end{itemize}

As in the case of coverage in the protocol model, we begin our analysis for the
special case of all transmit powers being equal. The method developed for equal
transmit powers is then generalized to unequal transmit powers. We also aim to
generalize the methods to include statistical variations in the received signal
energy - such as fading and shadowing.

A capture transmitter for any point $x$ in the plane is a transmitter
$\tilde{t}$ for which $SINR(x,\tilde{t})\ge SINR(x,\tilde{u})\ \forall
\tilde{u}\in T\setminus\tilde{t}$. The subset of points covered by a
transmitter $\tilde{t}$ is the set of points captured by $\tilde{t}$ at which
$SINR(x, \tilde{t}) > \beta$.

Since we are interested in studying the coverage region for each transmitter, we
first analyze the partition of the plane into the capture regions of the
transmitters. The coverage regions are subsets of the capture regions.

\section{Equal Transmit Powers: Voronoi Partitions and Capture Transmitters}
The following lemma shows that the partition of the plane into capture regions
is the Voronoi diagram of the transmitter locations. First we note that:
\begin{equation}
\label{eq:1}
\frac{A}{M - A} \ge \frac{B}{M - B} \Leftrightarrow A \ge B,\ \forall A,B>0,\ M
> A + B
\end{equation}
\begin{lemma}
\label{lemma:sinr.voronoi}
If all transmit powers are equal, then a point $x$ is in the capture region of
$\tilde{t}$ if and only if $x$ is in the Voronoi partition corresponding to
$\tilde{t}$.
\end{lemma}
\begin{proof}
\[SINR(x, \tilde{t}) = \frac{\frac{P}{d(x,\tilde{t})^\alpha}}{\displaystyle
\sum_{\tilde{u}\in T\setminus\tilde{t}}\frac{P}{d(x,\tilde{u})^\alpha} + N_0}\]
\[\Leftrightarrow SINR(x, \tilde{t}) = \frac{\frac{P}{d(x,\tilde{t})^\alpha}}
    {\displaystyle\sum_{\tilde{u}\in T}\frac{P}{d(x,\tilde{u})^\alpha} + N_0 - 
    \frac{P}{d(x,\tilde{t})^\alpha}}\]
Let \[A = \frac{P}{d(x,\tilde{t})^\alpha},\ \textrm{and}\ M = \sum_{\tilde{u}\in
T} \frac{P}{d(x,\tilde{u})^\alpha} + N_0\]
Let $s$ be another transmitter in $T$, and $\displaystyle B =
     \frac{P}{d(x,\tilde{s})^\alpha}$
\[SINR(x, \tilde{t}) \ge SINR(x, \tilde{s}) \Leftrightarrow
     \frac{A}{M - A} \ge \frac{B}{M - B} \Leftrightarrow A \ge B\]
\[\Leftrightarrow\frac{P}{d(x,\tilde{t})^\alpha}\ge\frac{P}{d(x,\tilde{s})^\alpha}\]
\[\Leftrightarrow d(x,\tilde{t}) \le d(x,\tilde{s})\]

Thus, if $\forall \tilde{s},\ SINR(x,\tilde{t})\ge SINR(x,\tilde{s})$, then
$\forall\tilde{s},\ d(x,\tilde{t}) \le d(x,\tilde{s}) \Leftrightarrow x\in$
Voronoi partition of $\tilde{t}$.
\end{proof}
This lemma shows that the capture region for each transmitter lies in its
Voronoi partition. This is similar to Claim \itemref{1} - for the protocol model
for equal transmitter powers - which states that the coverage region of a
transmitter lies inside its Voronoi partition.

In order to find the coverage region of a transmitter, we need to find the
transmitter's Voronoi partition, and then the set of points in this partition at
which the SINR is at least $\beta$. Some examples of coverage region contours 
for varying $\beta$ are shown in Figures \ref{fig:SINRCoverage1}, 
\ref{fig:SINRCoverage2}, and \ref{fig:SINRCoverage3}.

\begin{figure}
\centering\includegraphics{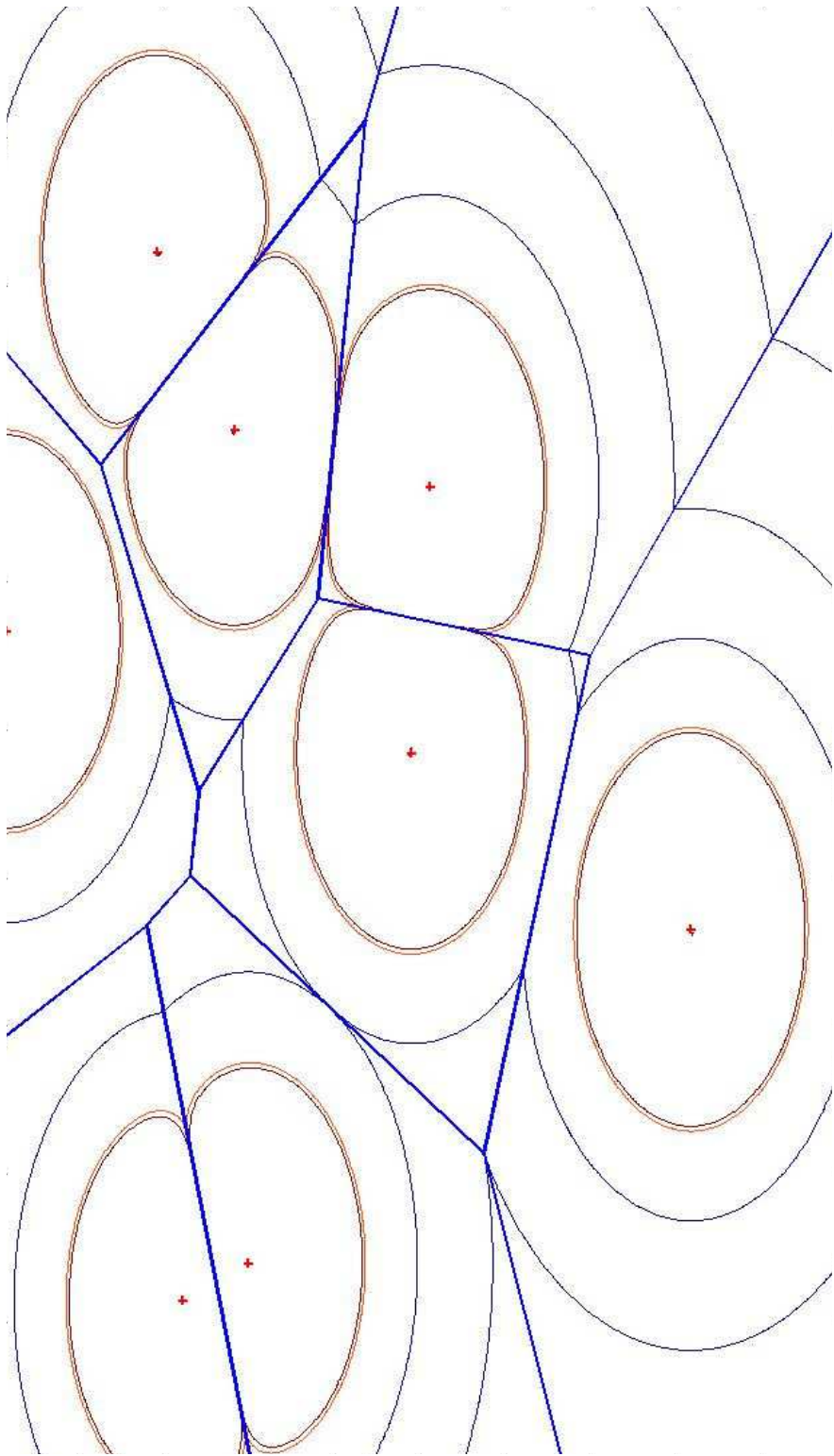}
\caption{\label{fig:SINRCoverage1}SINR Coverage: 8 Transmitters}
\end{figure}
\begin{figure}
\centering\includegraphics{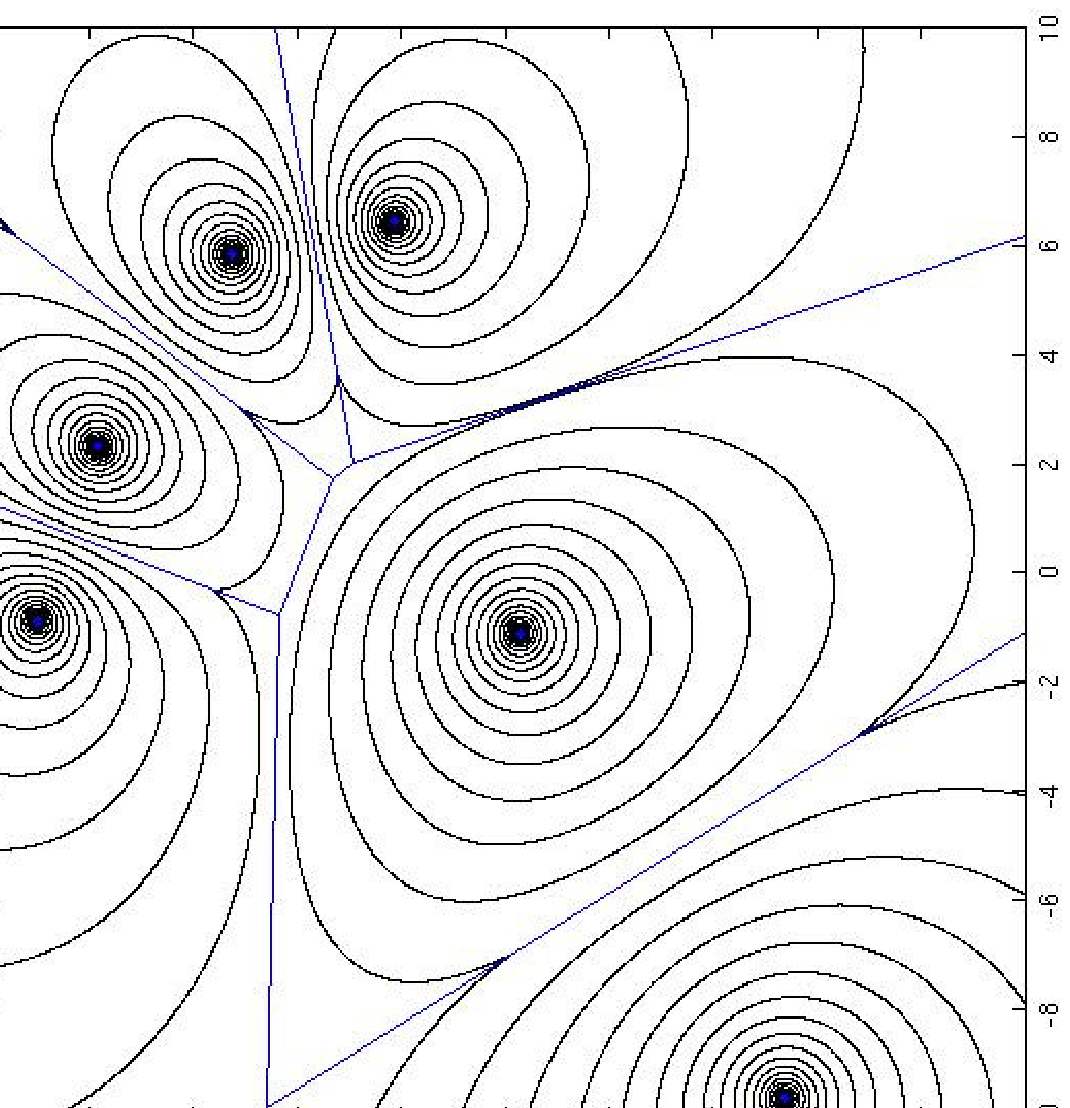}
\caption{\label{fig:SINRCoverage2}SINR Coverage: 6 Transmitters}
\end{figure}
\begin{figure}
\centering\includegraphics{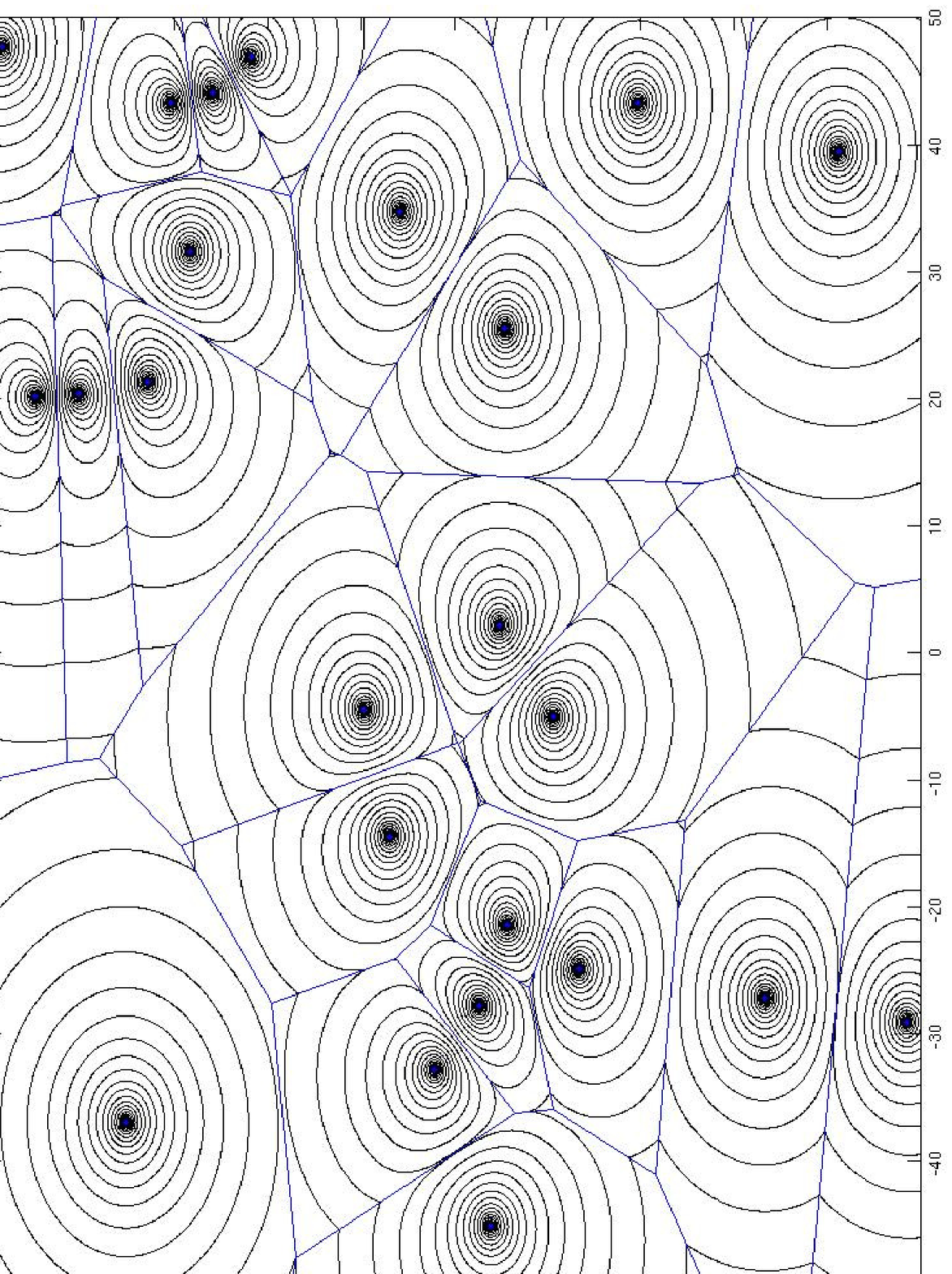}
\caption{\label{fig:SINRCoverage3}SINR Coverage: 25 Transmitters}
\end{figure}
 
We cannot represent the coverage region as a set of circular
arcs, like the conic polygons in Chapter \ref{chapter:Fixed}. Instead, we aim to {\it
approximate} the boundary by a finite sequence of low-degree polynomial
arcs. The choice of approximate boundary is such that the error in approximating
the coverage region is bounded, and the number of evaluations of the SINR
function required to build the representation is minimized.

Alt et al.~\cite{Alt01} demonstrate an efficient strategy for approximate 
representation for convex areas. This may be extended to SINR coverage regions, 
if they can be shown to be convex. As we see later in Section 
\ref{section:SINR_unequal}, though, coverage regions are not convex in general. 
Some recent research has focussed on restrictions on the SINR model parameters 
that yield convex regions.
\section{Concurrent Research in SINR Coverage}
Research on similar lines has recently been reported by Avin et
al.~\cite{Chen09}: an approximation algorithm to decide whether a point $x$ is in a
SINR coverage region. Our work was independently conceived. The algorithm in 
\cite{Chen09} is based on the following ideas:
\begin{enumerate}
\item Lemma \ref{lemma:sinr.voronoi},
\item The SINR coverage region $SINR(x, \tilde{t}) > \beta$ is convex for
$\alpha = 2$ and $\beta > 1$,
\item An error bound $\epsilon$, and approximate representation of the region
$SINR(x, \tilde{t}) > \beta$, and 
\item An algorithm that uses this representation to decide coverage at a point
within the error bound $\epsilon$.
\end{enumerate}
%
Our work has more general aims:
\begin{enumerate}
\item We aim for an approximate representation for $\alpha \ge 2$, since
$2\le\alpha\le 6$ for practical wireless environments
(see \cite{tse}). This representation should also be valid for any $\beta
>0$. The Voronoi partition contains points at which SINR $<1$, for which
decoding is possible in practice.
\item We conjecture that the region $\{x | SINR(x, \tilde{t}) > \beta\}\cap
\triangle(\tilde{t})$ is convex. (Following the terminology from Table
\ref{table:VoronoiPower}, $\triangle(\tilde{t})$ is the Voronoi partition
corresponding to transmitter $\tilde{t}$) Thus, convexity is independent of
$\beta$, but instead applies to the set of all points inside the
capture transmitter's Voronoi partition having SINR $>\beta$.
\end{enumerate}
\section{Convexity of SINR Coverage Regions in 2-D}
\label{section:convex}
We briefly review the definition of convexity, and outline our approach to
proving convexity of the SINR coverage region.

\begin{definition}
A set of points $S$ is a {\bf convex set} if $\forall a,b\in S$, the set of
all points $ka+(1-k)b, 0<k<1$ is in $S$. In other words, if the end-points of
any line segment $l$ lie in a convex set, then so does $l$.
\end{definition}
\begin{figure}
\centering\includegraphics[angle=270]{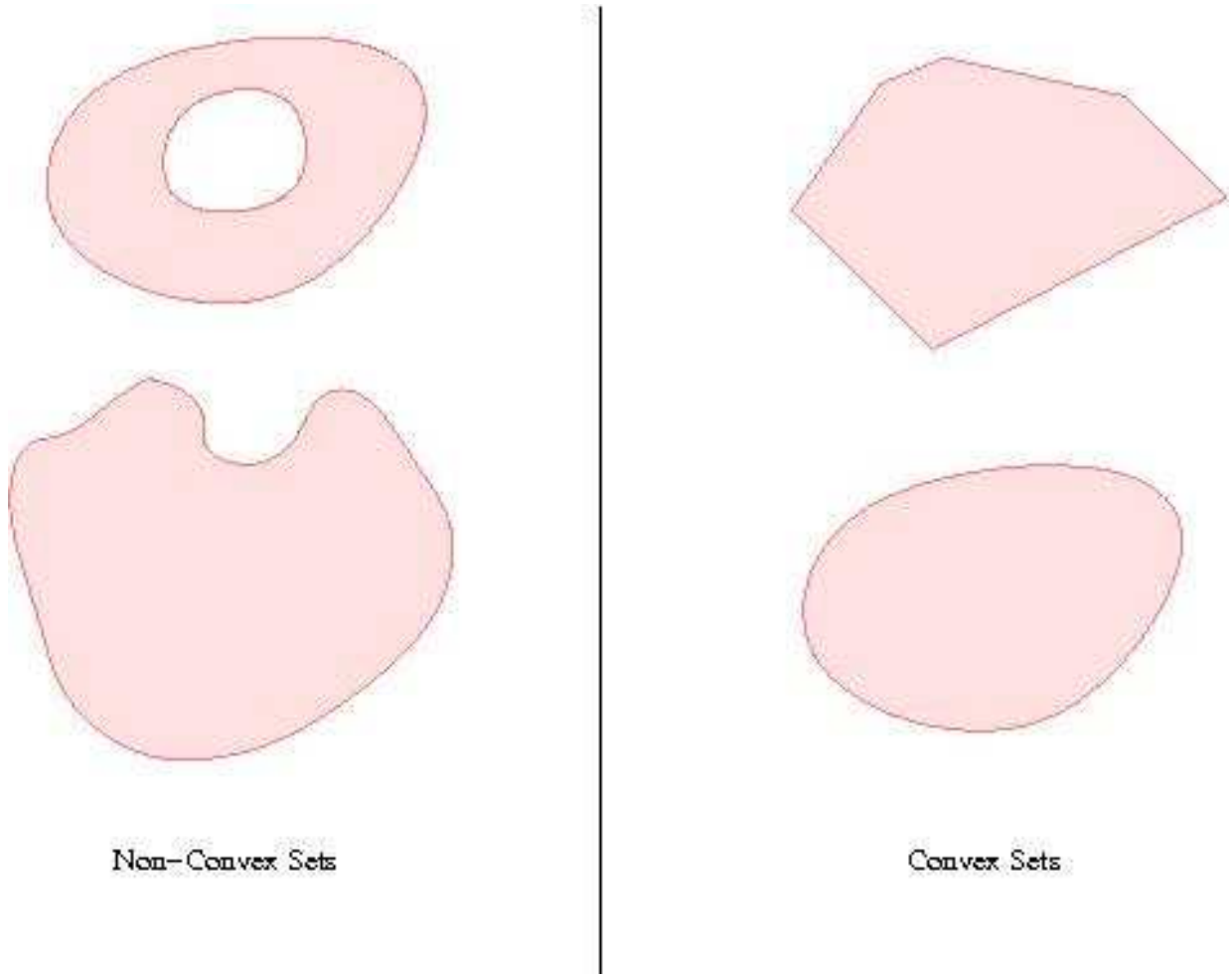}
\caption[Examples of Convex Sets]{Examples of Convex Sets in 2-D\label{figure:convex}}
\end{figure}
Figure \ref{figure:convex} shows 2-D examples of convex and non-convex sets.

Applying this definition to the SINR coverage region:\\
Let $l_{p,q}$ denote the line segment between points $p$ and $q$. The SINR
coverage region is convex if for all points $p$ and $q$ such that
$SINR(p,\tilde{t})>\beta$ and $SINR(q,\tilde{t})>\beta$,
the following is true: $SINR(x,\tilde{t})>\beta, \forall x\in l_{p,q}$.

Figures \ref{fig:SINRCoverage1}, \ref{fig:SINRCoverage2}, and
\ref{fig:SINRCoverage3} show convex SINR coverage regions for varying values of
$\beta$. Many other experiments we conducted also suggest convexity.

A direct analytic proof of convexity would require calculating partial
derivatives of the SINR formula. The intermediate expressions in such a proof
would be hard to analyze due to the presence of the power $\alpha$, $2n$
constants corresponding to the transmitter locations in 2-D, and $N_0$.

Avin et al.~\cite{Chen09} prove convexity for $\alpha=2$. This allows the SINR
formula to be expressed as a polynomial of degree $2$ with $2n+1$
coefficients. Proving convexity is still a significant challenge, even with this
simplification, as noted by the authors in the paper.
\section{Convexity for $\alpha\ge 2$ and $\beta>0$}
We now describe our approach to proving convexity of the SINR coverage
region. We aim for a proof for $\alpha\ge 2$ and $\beta>0$. We have proved a
restricted form of convexity called {\it star-convexity} for these
parameters. In contrast, Avin et al.~prove star-convexity also for
$\alpha=2$. We conjecture that ideas in our proof for star-convexity can be
generalized to prove convexity.

\begin{definition}[Star-Convex Set]
A set $S$ is {\bf star-convex} if $\exists r\in S$ such that for every line
segment $l$ with $r$ as one end-point, and any other point in $S$ as the other
end-point, all points in $l$ belong to $S$.
\end{definition}
\begin{figure}
\centering
\includegraphics[angle=270]{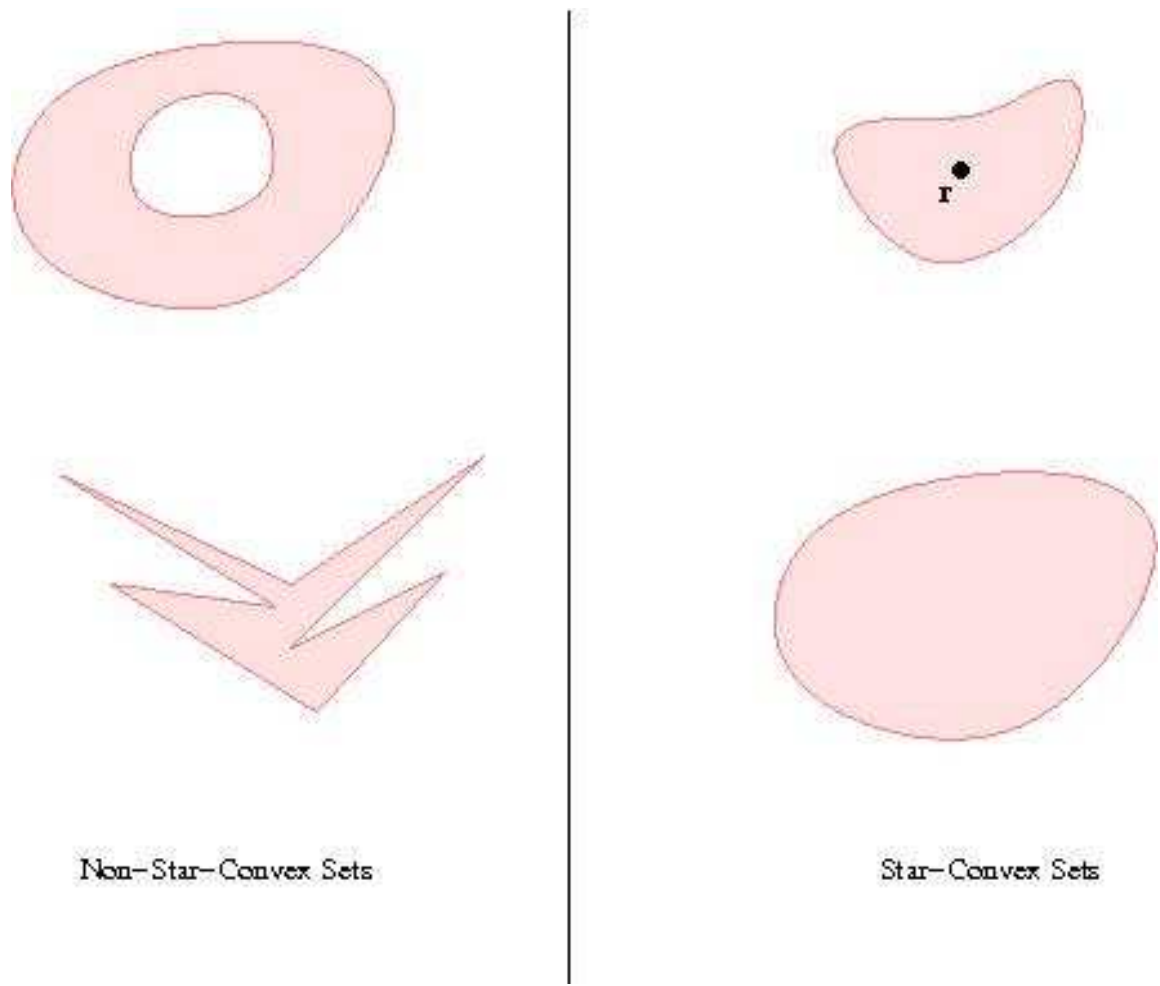}
\caption[Examples of Star-Convexity]{Examples of Star-Convexity in 2-D\label{figure:star.convex}}
\end{figure}
2-D examples star-convex and non-star-convex sets are shown in Figure
\ref{figure:star.convex}. All convex sets, by definition, are star-convex.

Some basic ideas in our approach follow:
\begin{enumerate}
\item Without loss of generality, we assume that $P_{\tilde{t}} = 1,\ \forall
\tilde{t} \in T$. The SINR function can be re-written as:
\[
SINR(p,\tilde{t})=\frac{1}{\displaystyle\sum_{\tilde{u}\in T\setminus\tilde{t}}
\frac{d(p,\tilde{t})^\alpha}{d(p,\tilde{u})^\alpha}+ N_0 d(p,\tilde{t})^\alpha}
\]
\item
Since $SINR(p,\tilde{t})\ge\beta \Leftrightarrow \displaystyle
\frac{1}{SINR(p,\tilde{t})} \le \frac{1}{\beta},\ \forall\beta>0$, testing SINR
coverage region convexity is equivalent to testing convexity for the region
\begin{equation}
\label{eqn:sum}
\sum_{\tilde{u}\in T\setminus\tilde{t}} \frac{d(p,\tilde{t})^\alpha}
{d(p,\tilde{u})^\alpha} + N_0 d(p,\tilde{t})^\alpha \le \frac{1}{\beta}
\end{equation}
\item

We view the sum of terms in Inequality \ref{eqn:sum} as a function
$f:\triangle(\tilde{t})\rightarrow R^+$. $f$ is a function of two variables. In
order to prove convexity of the region $f\le\frac{1}{\beta}$, we must analyze it
on domains that are line segments $l_{p,q}$, for arbitrary points $p,\ q$ in the
region $\triangle(\tilde{t})$.

We denote the restriction of function $f$ to the domain $l_{p,q}$ by
$f_{p,q}$. Since $l_{p,q}$ is a line segment, $f_{p,q}$ is a function of only
one variable.
\item

We use convexity terminology from Boyd et al.~\cite{boyd}. 
$f_{p,q}:l_{p,q}\rightarrow R$ is {\it monotonic} if:
$f_{p,q}$ is differentiable and 
$\forall x\in l_{p,q},\ \frac{d}{dx}f_{p,q}$ has the same sign.

$f_{p,q}$ is {\it unimodal} if $\exists x^*\in l_{p,q}$ such that $f_{p,x^*}$
and $f_{x^*,q}$ are monotonic and $\frac{d}{dx}f_{p,q}(x) = 0 \Rightarrow
\frac{d^2}{dx^2}f_{p,q}(x)\ge 0$.

$f_{p,q}$ is {\it quasi-convex} if it is monotonic or unimodal. If $f_{p,q}$ is
quasi-convex, then $\forall l_{p^*,q^*}\subset l_{p,q},\ f_{p^*,q^*}$ is
quasi-convex.
\begin{lemma}
\label{lemma:monotone}
In order to prove convexity of the region $f\le\frac{1}{\beta}$, it suffices to
prove that: for each pair of points $\{p,q\}$ on the boundary of the region, $f$
is quasi-convex on the line segment $l_{p,q}$.
\end{lemma}
\item
We analyze each term in the sum separately, and then analyze the sum. The
following lemma shows conditions for which monotonicity of individual terms
extends to monotonicty of the sum:
\begin{lemma}
\label{lemma:addMono}
If $f:R\rightarrow R$ and $g:R\rightarrow R$ are two functions monotonically
increasing in an interval $[p,x^*]$ and monotonically decreasing in an interval
$[x^*,q]$, then the function $f+g$ is also monotonically increasing in $[p,x^*]$
and monotonically decreasing in $[x^*,q]$. In other words, the sum of quasi-convex
functions with the same mode(s) and directions of monotonicity is also
quasi-convex.
\end{lemma}
\item
In the quasi-convexity analysis for each term, we work with the power $2$ instead
of $\alpha$. The following lemma shows that this approach is sufficient to prove
quasi-convexity for any $\alpha$.
\begin{lemma}
\label{lemma:alphaPower}
Consider a differentiable function $f:R\rightarrow R^+$. If $f$ is quasi-convex in
interval $[a,b]$, then for all $\alpha>0$, $f^\alpha$ also quasi-convex in
$[a,b]$.
\end{lemma}
\begin{proof}
Monotonicity can be verified by differentiating $f^\alpha$:\\
$\displaystyle\frac{d}{dx}f(x)^\alpha = \alpha f(x)^{\alpha-1}
\frac{d}{dx}f(x)$. Since $\alpha>0$ and $f>0$, $\displaystyle
\frac{d}{dx}f(x)^\alpha$ and $\displaystyle\frac{d}{dx}f(x)$ have the same
sign. Thus, monotonic increase in $f$ is equivalent to monotonic increase in
$f^\alpha$. Similarly, monotonic decrease is also equivalent for both
functions.

Since $\displaystyle\frac{d^2}{dx^2}f(x)^\alpha = \alpha f(x)^{\alpha-1}
\frac{d^2}{dx^2}f(x)$ at ${x|\frac{d}{dx}f(x)=0}$,
$\displaystyle\frac{d}{dx}f(x)^\alpha = 0\Rightarrow
\displaystyle\frac{d^2}{dx^2}f(x)^\alpha \ge 0$.

$f^\alpha$ is quasi-convex in $[a,b]$.
\end{proof}
\item
We want to prove convexity for the SINR coverage region in the Voronoi partition
corresponding to each transmitter. Accordingly, we require to establish
quasi-convexity for functions evaluated on line segments inside the
Voronoi partition.

The sum in Inequality \ref{eqn:sum} has two types of terms: $\displaystyle
\frac{d(p,\tilde{t})^\alpha}{d(p,\tilde{u})^\alpha}$ and $N_0
d(p,\tilde{t})^\alpha$. It is sufficient to show quasi-convexity for
$\displaystyle\frac{d(p,\tilde{t})^2}{d(p,\tilde{u})^2}$ and $N_0
d(p,\tilde{t})^2$, as shown by Lemma \ref{lemma:alphaPower}.

We show quasi-convexity for the term $\displaystyle
\frac{d(p,\tilde{t})^2}{d(p,\tilde{u})^2}$ on intersections of lines with the
half-space $\hbar(\tilde{t}, \tilde{u})$. This restriction to the half-space is
sufficient, since the intersection of the half-spaces
$\displaystyle\bigcap_{\tilde{u}\in T\setminus\tilde{t}} \hbar(\tilde{t},
\tilde{u})$ is the Voronoi partition corresponding to $\tilde{t}$.
\item
One final trick helps us to further simplify the analysis - rotation and
translation of axes. Rotation and translation of axes do not alter the results
of the analysis, since the distance function is invariant of these
operations. Given line $l$ on which we need to evaluate a term, we rotate axes
such that $l$ becomes the new $x$-axis. For evaluating the term $\displaystyle
\frac{d(p,\tilde{t})^2}{d(p,\tilde{u})^2}$, we translate the axes such that the
positive $x$-axis corresponds to $l\cap\hbar(\tilde{t}, \tilde{u})$. (If
$l\cap\hbar(\tilde{t}, \tilde{u})=\phi$, then no translation is done.) This
operation is shown in Figure \ref{figure:rotate}.

\begin{figure}
\centering
\includegraphics[angle=270, width=6in]{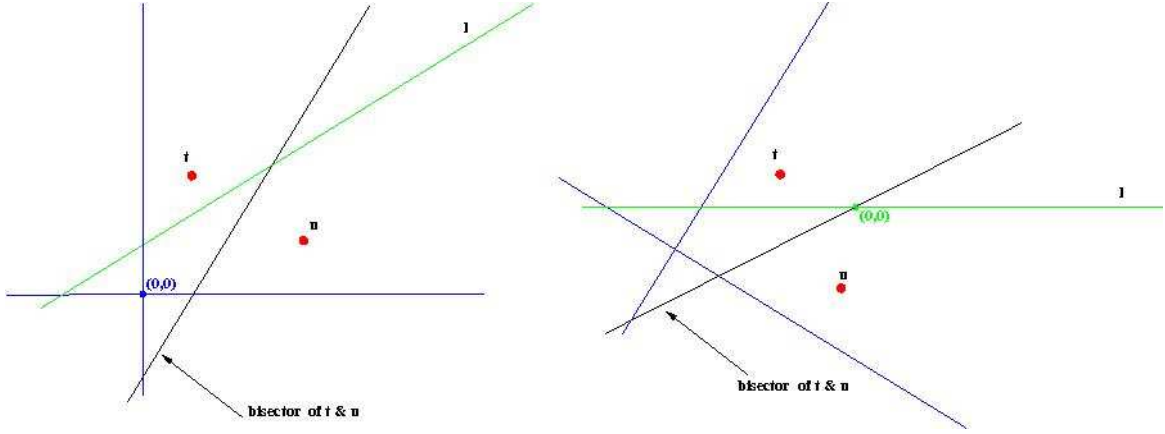}
\caption{\label{figure:rotate}Rotating \& Translating Axes}
\end{figure}

The following lemma shows quasi-convexity for the term $\displaystyle
\frac{d(p,\tilde{t})^2}{d(p,\tilde{u})^2}$ evaluated on the $x$-axis. We choose,
for this term, the brief notation $f_{t,u}$. We recall that $x>0$ corresponds to
points on $l$ closer to $\tilde{t}$ than $\tilde{u}$ (see Figure
\ref{figure:rotate}). Figure \ref{figure:unimodal} helps visualize the lemma.

\begin{figure}
\centering
\includegraphics[angle=270, width=6in]{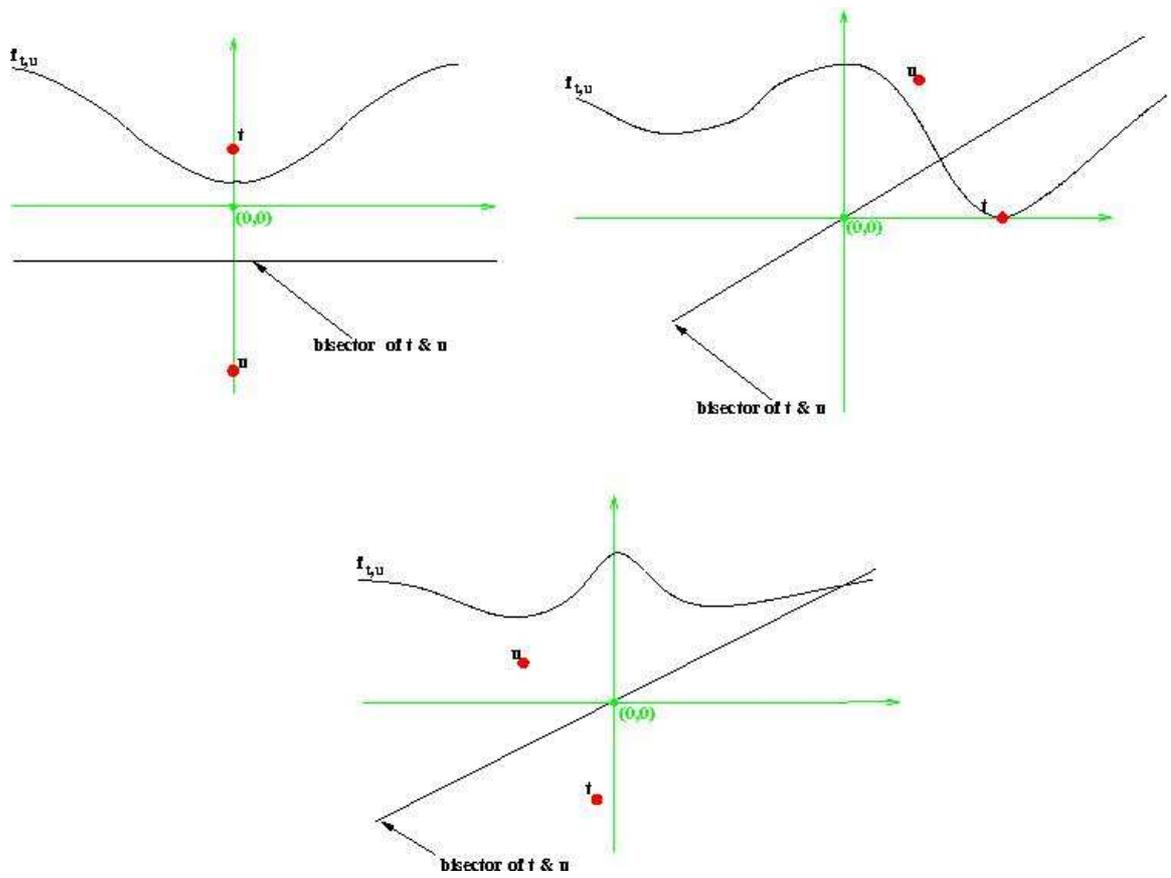}
\caption[Quasi-convexity in the SINR model]{\label{figure:unimodal}$f_{t,u}\equiv\displaystyle
\frac{d(p,\tilde{t})^2}{d(p,\tilde{u})^2}$ is quasi-convex for $x>0$}
\end{figure}

\begin{lemma}
\label{lemma:unimodal}
For an appropriate choice of $x$-axis, $f_{t,u}$ is quasi-convex on $x\ge
0$.
\end{lemma}
\begin{proof}
\begin{description}
\item[Case 1:] The bisector of $t$ and $u$ is parallel to $l$.

Choose the line through $t$ and $u$ as the $y$-axis, and $l$ as the
$x$-axis. Let $t\equiv(0,y_t)$ and $u\equiv(0,y_u)$ in this co-ordinate system.

$f_{t,u} = \frac{x^2+y_t^2}{x^2+y_u^2}$. Since $t$ is closer to any point on
$l$, $y_t^2 \le y_u^2$. Thus, $f_{t,u}$ is unimodal with minimum at $x=0$. This
can be verified by differentiation:\\
$\frac{d}{dx}f_{t,u} = \frac{2x(y_u^2-y_t^2)}{(x^2+y_u^2)^2}$. 
\item[Case 2:] The bisector of $t$ and $u$ intersects $l$.

Choose $l$ as the $x$-axis, and the intersection point of $l$ with the bisector
as the origin. Let $t\equiv(x_t,y_t)$ and $u\equiv(x_u,y_u)$ in this co-ordinate
system.

$f_{t,u} = \frac{(x-x_t)^2+y_t^2}{(x-x_u)^2+y_u^2}$. We verify unimodality of
$f_{t,u}$ in $x\ge 0$ by differentiation:\\ $\frac{d}{dx}f_{t,u} =
\frac{2(x-x_t)}{(x-x_u)^2+y_u^2}- \frac{2(x-x_u)}{((x-x_t)^2+y_t^2)^2}$\\
Thus, the optima are at $x^2 = x_t^2 + y_t^2$. We can verify that
$x=\sqrt{x_t^2+ y_t^2}$ is a minima by substituting for $x$ in
$f_{t,u}$.\qedhere
\end{description}
\end{proof}
\item
\begin{corollary}[Star-Convexity]
The SINR coverage region is star-convex for $\alpha\ge 2$ and $\beta>0$. 
\end{corollary}
\begin{proof}
If one end-point of any line segment is always chosen to be the
transmitter $\tilde{t}$, then minima of $f_{\tilde{t},u}$ for all $u$ are at
$\tilde{t}$. Also, $d(x,\tilde{t})^2$ is convex on $x\ge 0$ with a unique
minimum at $t$. Due to Lemma \ref{lemma:alphaPower}, the same monotonicity
properties hold if we replace the power $2$ with the power $\alpha$. Since the
minima coincide, the sum of quasi-convex terms are also quasi-convex, due to
Lemma \ref{lemma:addMono}.
\end{proof}
\item A proof for convexity requires a generalization of Lemma
\ref{lemma:addMono} that excludes the pre-condition that all minima
coincide. Though we see in Lemma \ref{lemma:unimodal} that all terms are
unimodal, their minima occur at different points on the $x$-axis. We can show
that the minima occur close together (within distance $|y_t|$ of each other),
but a generalization of Lemma \ref{lemma:addMono} eludes us.
\end{enumerate}
\section{SINR Coverage for Unequal Transmit Powers}
\label{section:SINR_unequal}
Lemma \ref{lemma:sinr.voronoi} for equal transmit powers gives us a
characterization of the SINR capture regions by Voronoi partitions - in terms of
distance alone, i.e.~independent of the transmit power(s). The partition of SINR
capture regions for unequal transmit powers corresponds to
``multiplicatively-weighted'' Voronoi partitions. An example of a
multiplicatively-weighted Voronoi diagram is Figure \ref{fig:mwVoronoi}.
The generalization of partitioning of SINR capture regions for unequal transmit
powers is shown by the generalization of Lemma \ref{lemma:sinr.voronoi} below:
\begin{lemma}
\label{lemma:sinr.unequal.tx}
A point $x$ is in the capture region of transmitter $\tilde{t}$ if and only if
$x$ is in the multiplicatively weighted Voronoi partition corresponding to
$\tilde{t}$ with weight $P_{\tilde{t}}^\frac{1}{\alpha}$.
\end{lemma}
\begin{proof}
\[SINR(x, \tilde{t}) = \frac{\frac{P_{\tilde{t}}}{d(x,\tilde{t})^\alpha}}
{\displaystyle \sum_{\tilde{u}\in T\setminus\tilde{t}}
\frac{P_{\tilde{u}}}{d(x,\tilde{u})^\alpha} + N_0}\]
\[\Leftrightarrow
SINR(x,\tilde{t})=\frac{\frac{P_{\tilde{t}}}{d(x,\tilde{t})^\alpha}}
 {\displaystyle\sum_{\tilde{u}\in T}\frac{P_{\tilde{t}}}{d(x,\tilde{u})^\alpha}
 + N_0 - \frac{P_{\tilde{t}}}{d(x,\tilde{t})^\alpha}}\]
Let \[A = \frac{P_{\tilde{t}}}{d(x,\tilde{t})^\alpha},\ \textrm{and}\ M =
\sum_{\tilde{u}\in T} \frac{P_{\tilde{u}}}{d(x,\tilde{u})^\alpha} + N_0\]
Let $s$ be another transmitter in $T$, and $B =
     \frac{P_{\tilde{s}}}{d(x,\tilde{s})^\alpha}$
\[SINR(x, \tilde{t}) \ge SINR(x, \tilde{s}) \Leftrightarrow
     \frac{A}{M - A} \ge \frac{B}{M - B} \Leftrightarrow A \ge B\]
\[\Leftrightarrow\frac{P_{\tilde{t}}}{d(x,\tilde{t})^\alpha}\ge
\frac{P_{\tilde{s}}}{d(x,\tilde{s})^\alpha}\]
\[\Leftrightarrow P_{\tilde{t}}^\frac{1}{\alpha}d(x,\tilde{t}) \le
 P_{\tilde{s}}^\frac{1}{\alpha}d(x,\tilde{s})\]
Thus, if $\forall \tilde{s},\ SINR(x,\tilde{t})\ge SINR(x,\tilde{s})$, then
$\forall\tilde{s},\ P_{\tilde{t}}^\frac{1}{\alpha}d(x,\tilde{t}) \le
 P_{\tilde{s}}^\frac{1}{\alpha}d(x,\tilde{s}) \Leftrightarrow x\in$
multiplicatively-weighted Voronoi partition of $\tilde{t}$.
\end{proof}
The generalization to unequal powers is similar to the generalization of Voronoi
diagrams to Power diagrams introduced in Chapter \ref{chapter:Fixed}. In
fact, another term for the Power Diagram is ``additively-weighted'' Voronoi
diagram (see \cite{AurenhammerNotes}).

As seen in Figure \ref{fig:mwVoronoi}, the multiplicatively-weighted Voronoi
partition is not convex. Hence, we must relax the constraint of representing the
SINR region by approximating a convex region.
\begin{figure}
\centering
\includegraphics[width=6in]{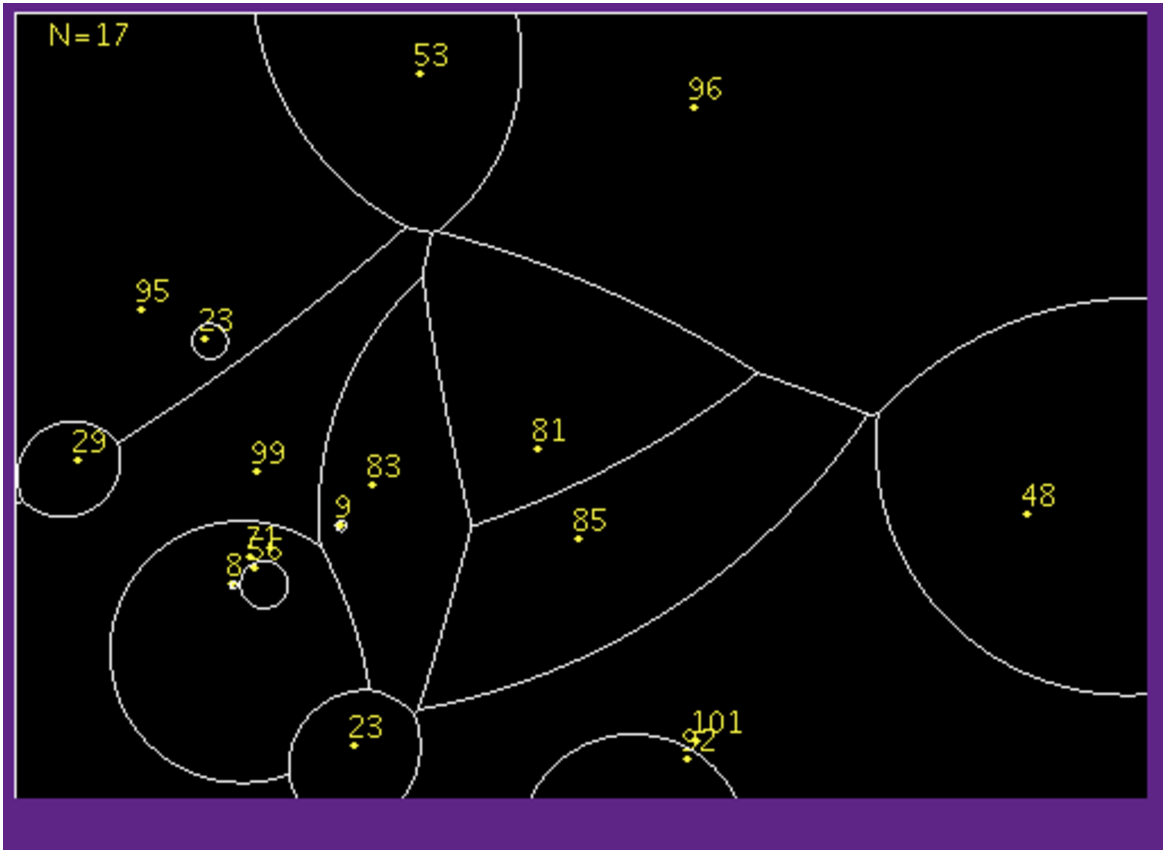}
\caption[Multiplicatively Weighted Voronoi Diagram]{Multiplicatively Weighted 
Voronoi Diagram (courtesy \cite{mwVoronoi}\label{fig:mwVoronoi}}
\end{figure}
\begin{figure}
\centering\includegraphics{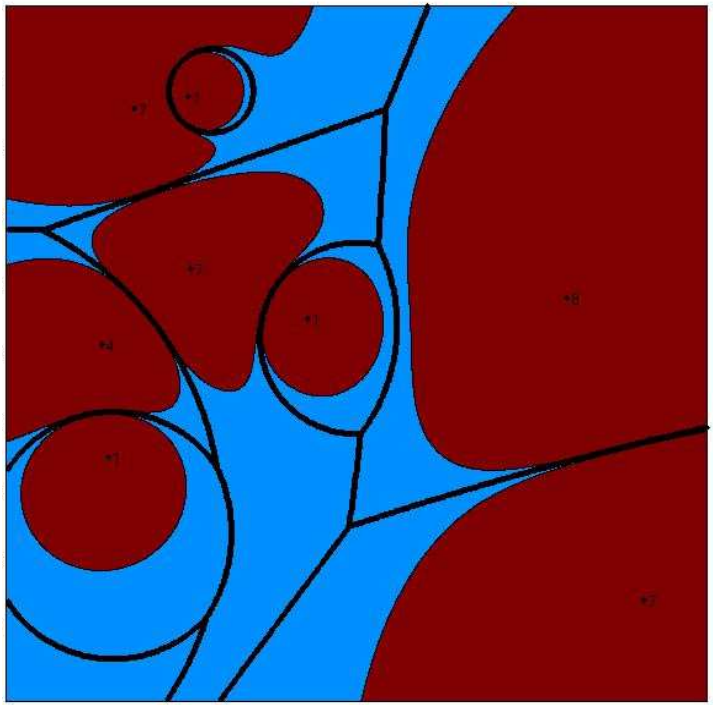}
\caption{\label{fig:non-sc}Non-Star-Convex Coverage Region} 
\end{figure}

Furthermore, as seen in Figure \ref{fig:non-sc}, the coverage region for a
transmitter may not be star-convex if all transmitters do not have the same
power. Thus, it is hard to follow a geometric algorithm approach to coverage
area computation for the SINR model.

\chapter{Coverage Optimization}
\label{chapter:Optimization}
In this chapter we propose an algorithm that optimizes coverage by choosing an
apt transmit power assignment to transmitters operating on the same
channel. Though it is desirable to solve this problem exactly in an analytical
framework, it is observed that this would require solutions to implicit
non-linear equations in many variables (see Section \ref{section:genOpt}). Thus,
to solve this optimization problem, we propose a derivative-free approach with
the Random Hill Climbing algorithm. We demonstrate the efficiency of this
algorithm using the SINR model. An important advantage of this algorithm is that
it uses a coverage estimation procedure that can accommodate any deterministic
coverage model.

In this chapter we presents simulation results to justify the performance
of the proposed algorithm by comparison with optimization results using the
Nelder-Mead method \cite{opti} and exhaustive search.

The rest of the chapter is organized as follows: The remainder of this section
describes related work and notation. Section \ref{section:Approach} initiates
discussion and background for our solution approach. Section \ref{section:Algos}
describes methods we have analyzed for optimizing coverage, Section
\ref{section:Results} presents an experimental comparison of these methods.
\section{Related Work}
Ahmed et al.~\cite{Keshav} study optimum transmit power assignments to access 
points assuming a protocol model. Some optimization problems in link scheduling 
and power control in the SINR
model are closely related to our problem. Goussevskaia et al.~\cite{gous} show
that a discrete problem of `single-shot scheduling' with weighted links is
NP-hard. Lotker et al.~\cite{lotker} and Zander et al.~\cite{zander} give
efficient algorithms for optimizing the maximum achievable SINR in a set of
links. Yates et al.~\cite{yates} report an algorithm to optimize the total
uplink transmit power for users served by a base station, assuming that all
users meet the minimum SINR constraint. A recent study by Altman et 
al.~\cite{Altman} considers {\em SINR games} played co-operatively and 
co-optively between base stations to maximize coverage area for mobile 
receivers or determine optimal placement for base stations themselves.

A more recent study of coverage optimization algorithms for indoor coverage 
appears in Reza et al.~\cite{reza2014comprehensive}. A new optimization model 
based on extrapolation of data collected from measurement tools is given by 
Kazakovtsev in \cite{kazakovtsev2013wireless}.

In contrast to these studies, our work reports a method to obtain the power 
assignment that maximizes the coverage area; that is, the power assignment that 
enables the maximum possible number of receivers to meet the SINR constraint. 
Our Random Hill Climbing procedure is inspired by Mudumbai et 
al.~\cite{madhow}. They report a randomized procedure to synchronize multiple
transmissions to send a common message coherently in a distributed beamforming
system.
\section{Notation}
\label{section:Notation}
We use the following assumptions:
\begin{enumerate}
\item A receiver is in coverage if the SINR due to some transmitter is above a
given threshold.
\item Transmitter locations are fixed; and all transmitters lie in the plane. We measure the coverage area within a bounding rectangle.
\item All receivers have the same receive sensitivity. All transmissions are omni-directional.
\end{enumerate}
The notation we use is as follows:
\begin{enumerate}
\item $T$ denotes the set of transmitters. All transmitters lie in an axis-parallel rectangle $E$ of unit area.
\item The transmit power of transmitter $t\in T$ is denoted $P_t$. The transmit
power is bounded, i.e., $P_{t_{min}} \le P_t \le P_{t_{max}}, \forall t\in T$.
\item $\hat{P}$ is the $|T|$-dimensional vector of power assignments.
\item $d(x, t)$ denotes the Euclidean distance between a point $x$ and
transmitter at $t$.
\item $\alpha$ is the path-loss exponent. In free space $\alpha = 2$, and other
values are experimentally or analytically derived from the propagation
environment \cite[Chapter 2]{tse}.
\item The receive power $R(x, t)$ at point $x$ from transmitter $t$ is
${P_t}/({d(x, t)^\alpha})$.
\item $N_0$ is the the ambient noise power.
\item $SINR(x, t)$ denotes the SINR at point $x$ due to transmitter $t$.
\begin{equation}
\label{eqn:sinr}
SINR(x, t) = \frac{R(x, t)}{\sum_{s\in T \setminus \{t\}}R(x, s) + N_0}
\end{equation}
\item $\beta$ is the minimum SINR required for successful reception.
\item The maximum SINR at a point $x$, denoted $SINR_{max}(x)$, is the maximum
of $SINR(x, t)$ over the set of transmitters $t \in T$.
\item The coverage area $\mathbb{C}$ is the measure of the area of the set $\{x |
SINR_{max}(x) \ge\beta\}\cap E$. Note that $\mathbb{C}$ is a function of the
transmitter locations and their power assignments.
\end{enumerate}
It must be noted that all instances
of $SINR(x, t)$ and $SINR_{max}(x)$, and $\mathbb{C}$ as well, are functions with
domain as the universe of vectors $\{\hat{P}\}$.
\section{Optimization Problem Solution Approach}
\label{section:Approach}
\label{section:genOpt}
We are interested in an optimization method that yields a global maximum for
function $\mathbb{C}$. In this section we describe the characteristics of the objective
function $\mathbb{C}$, and justify our choice of optimization method.
%
Optimization methods for non-linear functions (such as $\mathbb{C}$) are broadly
classified as gradient-based methods and direct-search methods \cite{opti}. In
gradient-based methods, the gradient vector - comprising the first order partial
derivatives of the objective function is required. Thus, the objective function
must be continuous and differentiable over the feasible
set. Further, its partial derivatives must be expressible in an
analytical form (``closed-form expression'').

Optimization of objective functions that are non-differentiable or do not have
closed-form expressions requires a heuristic method to iterate through feasible
solutions efficiently. Such a method is called a direct search method.

Our objective function $\mathbb{C}$ has domain $\hat{P}$ - the vector of power
assignments. $\mathbb{C}$ is the measure of the area of the region 
$\{x | SINR_{max}(x) \ge \beta\}\cap E$. However, the measure of this area is 
not expressible in analytical form; and thereby, the partial derivatives of 
$\mathbb{C}$ are not available in analytical form either.

%
We must choose a direct search method to maximize $\mathbb{C}$. Further, since $\mathbb{C}$
does not have an expression in analytical form, its values required in the
execution of the maximization program must be estimated by numerical
methods. The following section describes our methods for estimating the coverage
area, and direct search methods for coverage area optimization.
\section{Proposed Solution Method Details}
\label{section:Algos}
\subsection{Estimating the Coverage Area}
The coverage area is estimated as follows: We choose an appropriate finite
sample set of points in $E$, and report the estimate of $\mathbb{C}$ as the fraction of
this sample for which $SINR_{max} \ge \beta$. (Note that $E$ is of unit area.)
Computing this fraction involves computations of $SINR_{max}(x)$ for all sample
point $x$. Each computation of $SINR_{max}$ is $O(|T|)$.
We propose two methods for choosing the sample:
\begin{enumerate}
\item The set of points on an axis-parallel grid on $E$.
\item A random finite set of points in $E$.
\end{enumerate}
The accuracy of the estimate depends on the size and distribution of the
sample. However, we can obtain probabilistic guarantees on the number of 
samples required for a desired level of accuracy as follows.

Let the desired estimation error ratio be $\epsilon\in(0,1)$; that is, the 
estimated area lies in $\left[(1-\epsilon) \mathbb{C}, (1+\epsilon) \mathbb{C}\right]$ 
with probability $\delta$.

We present a method to arrive at a `sufficient' number of random samples or 
grid points for a desired accuracy ratio. Let $X\in\{0, 1\}$ be an indicator 
random variable corresponding to the event of a sample point belonging in the 
coverage area. Thus, $Pr(X=0) = 1-\mathbb{C}$ and $Pr(X=1) = E_{X} = \mathbb{C}$. 
Let $S_n$ be the random variable corresponding to the sum of $n$ sampling 
events. In context, $\frac{S_n}{n}$ is the estimated area.

By Chernoff's bound \cite{mitzenmacher}, we know that:
\begin{equation}
\label{eqn:chernoff}
n > \frac{3ln(\frac{2}{\delta})}{\epsilon^{2}\mathbb{C}} \Rightarrow 
Pr(|\frac{S_n}{n} - \mathbb{C}| \le \epsilon \mathbb{C}) \ge 1 - \delta
\end{equation}
Thus, we can choose $n$ by setting the estimation error $\epsilon$ and 
probability guarantee $\delta$ suitably.

The grid comprises of lines parallel to the axes. We could not find an 
analytical method for deriving the number of grid lines required for a desired 
estimation error. However, as a thumb rule, we use four times the number
of points as in the random sample. Using the inequality in \ref{eqn:chernoff}, 
for a desired accuracy of $85\%$ with an probability guarantee of $90\%$, we 
need to sample $400$ points. For these accuracy and guarantee parameters, we 
use $1600$ points in the grid.

In later sections, we refer as {\em EstimateArea} the subroutine that estimates 
the coverage area given the transmitter locations and power assignment. In 
order to give deterministic guarantees in comparisons of optimization 
algorithms, we use the grid method; otherwise, each run of an algorithm will 
use different sample points, and differences of accuracy between two algorithms 
may be due to the random choice of sample points.

We note that this estimation method can be extended to real world propagation 
models, such as espoused by project WINNER \cite{winner}, by coding them as 
subroutines that compute receive powers at a sample point.
\subsection{Direct search methods}
\label{section:Algos-D}
A direct search maximization method iterates through choices of vectors in the
feasible set without recourse to the gradient. We have analyzed two direct
search methods - Random Hill Climbing and Nelder-Mead. We also include an
exhaustive search method as a baseline for comparing the accuracy and power
efficiency of these methods.
\subsubsection{Exhaustive Search} The input to this algorithm is the set of
transmitter locations and their power thresholds. The output of this algorithm
is a power assignment to the transmitters that maximizes the coverage area. This
algorithm discretizes the power assignment into levels, and returns the level
vector that maximizes the coverage area.

%
This algorithm iterates through $k$ levels of power for each transmitter and
executes the subroutine $EstimateArea$ $k^{|T|}$ times. The subroutine $nextVec$
returns the next power vector in some sequence. In our implementation, we have
chosen this sequence in increasing order of total power, such that the added
objective of finding the minimum total power that maximizes the coverage is also
achieved. This algorithm serves as a benchmark to evaluate the accuracy and
total power reported of the other methods we study: Random Hill Climb and
Nelder-Mead.

\subsubsection{Random Hill Climbing} This optimization procedure iteratively
generates random vector increments to the current best known feasible
solution. It then updates the best known solution with the incremented vector if
the objective function value on the increment is better than the best known
value.

We have implemented a variant of Random Hill Climbing for maximizing
the coverage area.
%

We ran a large set of experiments with different parameters to compare this
procedure with Exhaustive Search to evaluate convergence and maxima for a small
set of transmitters. We also verified that, for large networks, it converges 
faster than another well-known deterministic procedure - Nelder-Mead. The 
details of these analyses follow later in Section \ref{section:Results}.

This algorithm also uses the subroutine $EstimateArea$.
\begin{algorithm}[RandomHillClimbing(T, $P_{min}$, $P_{max}$)]
\label{algo:rhc}
\end{algorithm}
\begin{algorithmic}
\STATE $\hat{bestP} \gets P_{min}$
\STATE $bestEstArea \gets 0$

\REPEAT
	\STATE $attempts \gets 0$
	\STATE $scaleUp \gets 1 + scaleFactor$
	\STATE $scaleDown \gets 1 + scaleFactor$
	\STATE $\hat{shrinkIncr} \gets stepSize\cdot (P_{max} - P_{min})$
	\STATE $\hat{stretchIncr} \gets stepSize\cdot (P_{max} - P_{min})$
	\STATE $localMaxima \gets FALSE$

	\WHILE {$(attempts < maxIterations)$ {\bf and} $(localMaxima =
									FALSE)$}
		\STATE $attempts \gets attempts + 1$
		\IF {$attempts > scaleUpIncr$}
			\STATE $scaleUp \gets 1 + 2\cdot scaleFactor$
		\ENDIF
		\IF {$isEvenNumber(attempts)$}
			\STATE $\hat{shrinkIncr} \gets \hat{shrinkIncr} / scaleDown$
			\STATE $\hat{t} \gets \hat{bestP} + getRandom(0, \hat{shrinkIncr})$
		\ELSE
			\STATE $\hat{stretchIncr} \gets \hat{stretchIncr} \cdot scaleUp$
			\STATE $\hat{t} \gets \hat{bestP} + getRandom(0, \hat{stretchIncr})$
		\ENDIF
		\STATE Scale $\hat{t}$ to fit values in range specified by
			$P_{min}$ and $P_{max}$

		\IF {$EstimateArea(T, \hat{t}) >  bestEstArea$}
			\STATE $bestEstArea \gets EstimateArea(T, \hat{t})$
			\STATE $\hat{bestP} \gets \hat{t}$
			\STATE $localMaxima \gets TRUE$
		\ENDIF
	\ENDWHILE
\UNTIL {$attempts \ge maxIterations$}

\RETURN $bestEstArea$, $\hat{bestP}$
\end{algorithmic}
Table \ref{table:rhcParms} describes the parameters used in the Random Hill
Climbing algorithm described in Algorithm \ref{algo:rhc}. Each iteration of the
inner loop checks whether a random increment to the best known vector
($\hat{bestP}$) gives a higher maximum. Random increments increase on odd
numbered attempts and reduce on even numbered attempts. This bi-directional
approach covers both steep and shallow saddle points: a large
increment exits a shallow saddle point, whereas a small increment exits a steep
saddle point. Further, to exit a wider shallow region, the rate of increment
increases after $scaleUpIncr$ attempts.

When an improved best estimate of a local maxima is found, the outer ({\bf
repeat}) loop resets the increments. If no better estimate is found after
$maxIterations$ of the inner loop, the procedure exits with the current
estimate. Since the outer loop iterates through increasing estimates of local
maxima, the returned result is an estimate of the global maxima.
%
%
\begin{table}
\centering
\caption{Random Hill Climbing Parameters}
\begin{tabular}{|l|p{2.8in}|l|}
\hline
{\bf Parameter} & {\bf Description} & {\bf Simulation values}\\
\hline
scaleFactor & Controls magnitude of the next random increment & $0.05\alpha$
\\
\hline
stepSize & Controls initial magnitude of random increment & $0.01$\\
\hline
maxIterations & Maximum number of attempts to exit local maximum & 2000\\
\hline
scaleUpIncr & Number of attempts after which random increment magnitude
increases & 100\\
\hline
\end{tabular}
\label{table:rhcParms}
\end{table}
\subsubsection{Nelder-Mead}
The Nelder-Mead is a well-known direct method for local
optimization. Implementations are available in many public forums, and numerical
packages like MATLAB have standard implementations. Even though it is an
unconstrained optimization method, constraints can be handled by penalty
functions \cite{opti}.

We have adapted an implementation by Flanagan (available online at
\cite{nm}). A detailed description of the algorithm is found in \cite{opti}. We
use a random vector to initialize the method, and maintain the default parameter
values of the algorithm.
%

We describe the Nelder-Mead method informally:
It maintains a simplex of $n+1$
points (vertices). At each iteration, either a vertex is replaced or the simplex
is shrunk. First, the vertex with the least image is chosen as origin, and three
new vertices are generated in the direction of the median of the face formed by
excluding this vertex. These new vertices are called vertex ``contraction
inside'', ``contraction outside'' or ``expansion''. Of these three, the vertex
the highest image is chosen as a potential new vertex. If this vertex has image
higher than the origin (i.e. current least), then this vertex is replaces the
origin. Otherwise, the simplex is ``shrunk'', i.e. the face opposite the vertex
with the highest image is scaled down. This replacement and shrink procedure
continues until the difference between the least and highest image is below a
convergence threshold.

The parameters of the Nelder-Mead algorithm are the distance of the three
potential vertices from the origin, the shrink factor, and the convergence
threshold. We have used the default values for these factors: $0.5$ for
contraction and shrink, and $2$ for expansion.

The Nelder-Mead only returns a local maximum. In order to get the global
maximum, we run the algorithm a number of times with random initial vectors.
\subsubsection{Post-processing}
We use a post-processing routine to try improve the coverage area by forcing
smaller values to $P_{t_{min}}$. This routine iteratively sets the lowest
$i\in\{1\ldots |T|\}$ powers to their corresponding values in $P_{min}$.

Another advantage of the post-processing routine is when the designer wants to
find a subset of the input transmitters that is sufficient to attain maximum
coverage. We can run an optimization procedure by setting $P_{min}$ to
$\hat{0}$. Transmitters that have been assigned power $> 0$ are sufficient, and
those assigned power $0$ may be removed from the network.

The experiments we present in Section \ref{section:Results} use this
post-processing on the output of Random Hill Climbing and Nelder-Mead.

\begin{algorithm}[$postProcess(T, \hat{v}, P_{min})$]
\label{algo:ppr}
\end{algorithm}
\begin{algorithmic}
\STATE $\hat{u} \gets \hat{v}$ in increasing order

\STATE $P_0 \gets P_{min}$ re-ordered to maintain the mapping with $\hat{v}$
in the sorted order

\STATE $bestEstArea \gets 0$
\STATE $\hat{bestP} \gets \hat{u}$

\FOR {$i = 1 \to |T|$}

	\STATE Set first $i$ values in $\hat{u}$ to first $i$ values in $P_0$

	\IF {$EstimateArea(T, \hat{u}) >  bestEstArea$}
		\STATE $bestEstArea \gets EstimateArea(T, \hat{u})$
		\STATE $\hat{bestP} \gets \hat{u}$
	\ENDIF
\ENDFOR
\end{algorithmic}
\begin{figure}
\centering
\includegraphics{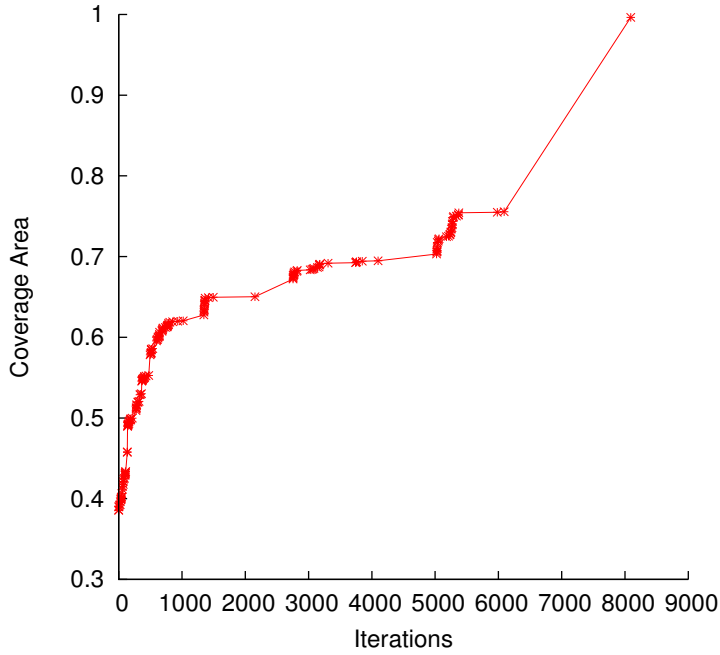}
\caption[Algorithm realization for Random Hill Climbing]{One simulation-based realization of Random Hill Climbing algorithm is
illustrated in this figure.}
\label{figure:Realization}
\end{figure}
Figure \ref{figure:Realization} shows the realization of the algorithm Random
Hill Climbing followed by post-processing. This figure shows the estimates for
the local maximum for the coverage area increasing by iterations of the inner
loop. The final increment is due to the post-processing.
\section{Algorithm Comparisons: Experiments and Results}
\label{section:Results}
\begin{figure}
\centering
\includegraphics[width=4.5in]{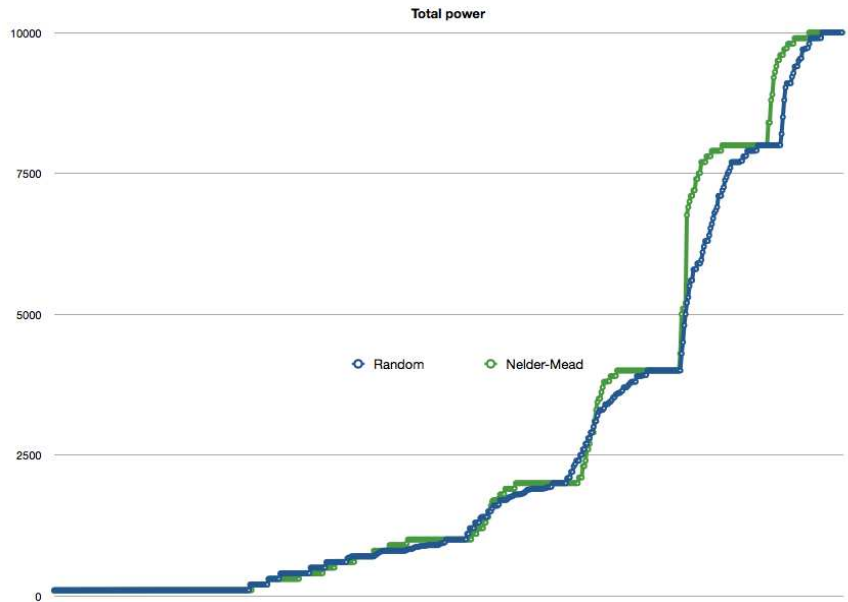}
\caption{Comparisons of the total power used by each algorithm. The X-axis shows an internal variable for the {\it experiment number}. For better readability, the ordering of experiment numbers corresponds with increasing power.}
\label{figure:TotalPower}
\end{figure}
\begin{figure}
\centering
\includegraphics{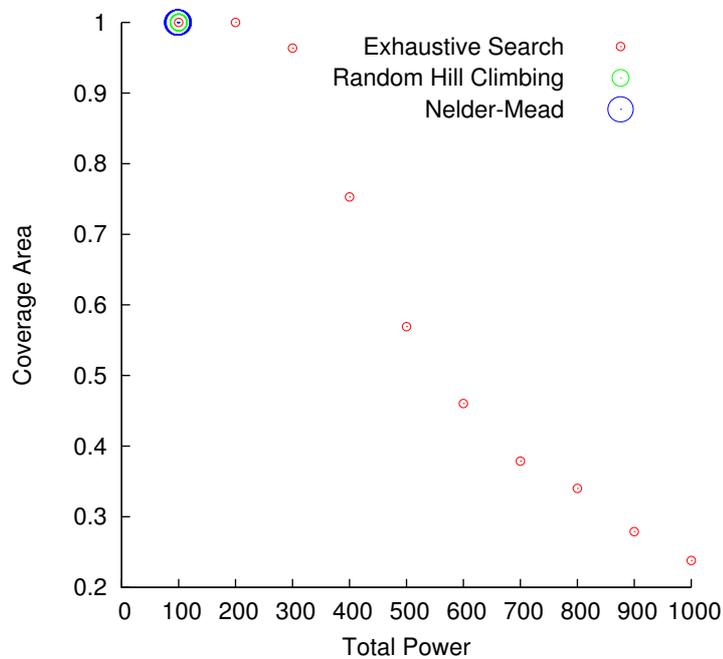}
\caption[Coverage area vs.~ total power]{The tradeoff between coverage-area and total power is illustrated in
this figure with $\alpha = 2$, $N_0 = 10^{-5}$. The Random Hill Climbing
and the Nelder Mead methods achieve the optimum-coverage transmit-power assignment.}
\label{figure:a2n5}
\end{figure}
\begin{figure}
\centering
\includegraphics{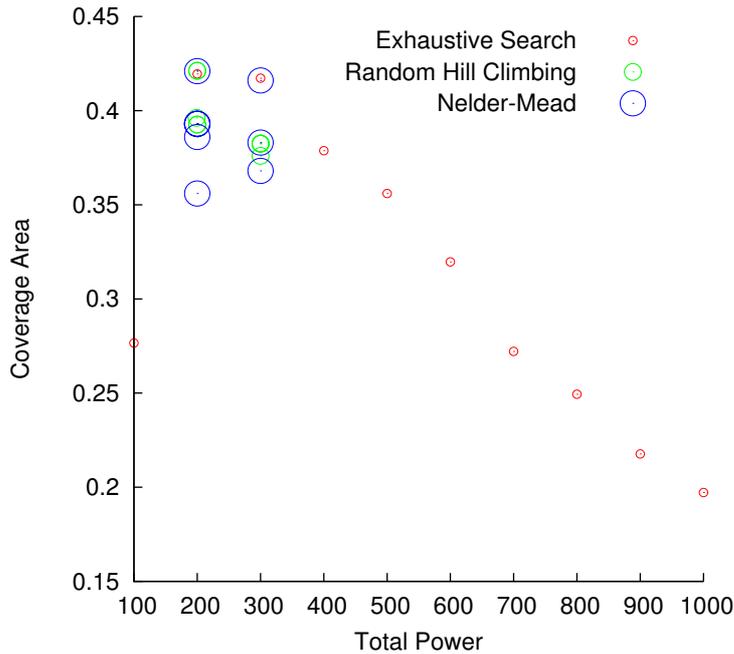}
\caption[Coverage area vs.~ total power]{The tradeoff between coverage-area and total power is illustrated in
this figure with $\alpha = 2$, $N_0 = 10^{-3}$. The coverage-area increases
initially from $0.28$ (noise-limited regime) towards $0.43$, and then it decreases
(interference-limited regime). The Random Hill Climbing and
the Nelder Mead methods achieve the optimum-coverage transmit-power
assignment.}
\label{figure:a2n3}
\end{figure}
We have analyzed the coverage optimization algorithms with the following set of
experiments. All experiments were run on a server with 8 quad-core processors
and 8 GB of RAM. All algorithms were coded in Java, and run on a 64-bit server
VM for J2SE 1.6.
\subsubsection{Comparison with Exhaustive Search}
We have compared the optimum obtained and total transmit power reported by the
three algorithms for the following scenarios:
\begin{itemize}
\item Number of transmitters: 10
\item $p_{min} = 0$ and $p_{max} = 100$ for all transmitters
\item Values of $\alpha$: 2, 3
\item Values of $N_0$: $10^{-3}$, $10^{-5}$
\item Size of Grid: 40 X 40
\end{itemize}

All experiments are run with the same set of transmitter-locations. We show two
of these results in Figures \ref{figure:a2n5} and \ref{figure:a2n3}. We see that
in both cases, both the total power consumed and coverage area reported by the
Random Hill Climbing and the Nelder-Mead methods is comparable to the exhaustive
search. The Nelder-Mead method may give a sub-optimal result, but since its
execution time is small, it can be run multiple times and the best value can be
chosen as the answer. Even with only $10$ transmitters, each run of Exhaustive
Search takes about $11.5$ hours! Hence, we have restricted our comparison
scenarios with Exhaustive Search to 10 transmitters.

\subsubsection{Random Hill Climbing vs Nelder-Mead for 10 to 100 transmitters}

\begin{figure}
\centering
\includegraphics{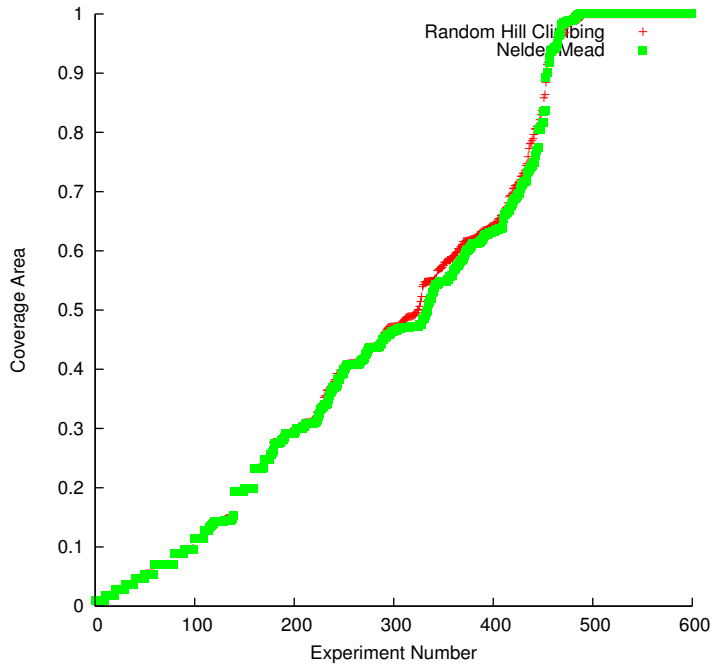}
\caption[Coverage area comparison with Nelder-Mead]{The optimum coverage-area ratio, as found by the Random Hill Climbing and the Nelder-Mead
methods, is illustrated in this figure for up to 100 transmitters. The two
methods are comparable in performance.}
\label{figure:CoverageRatio}
\end{figure}
\begin{figure}
\centering
\includegraphics{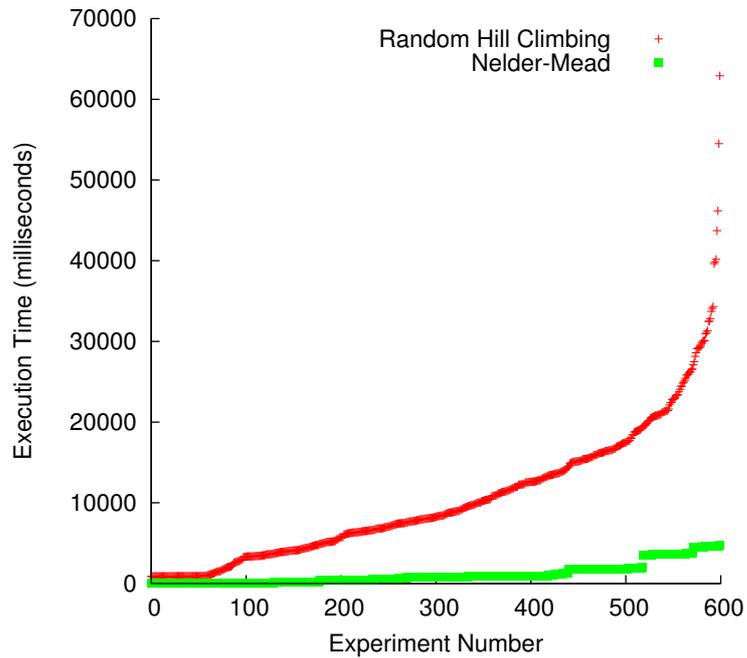}
\caption[Execution time comparison with Nelder-Mead]{Execution time for Random Hill Climbing and
Nelder-Mead for up to 100 transmitters is illustrated in this figure. Nelder-Mead
runs faster.}
\label{figure:ExecutionTime}
\end{figure}
\begin{figure}
\centering
\includegraphics{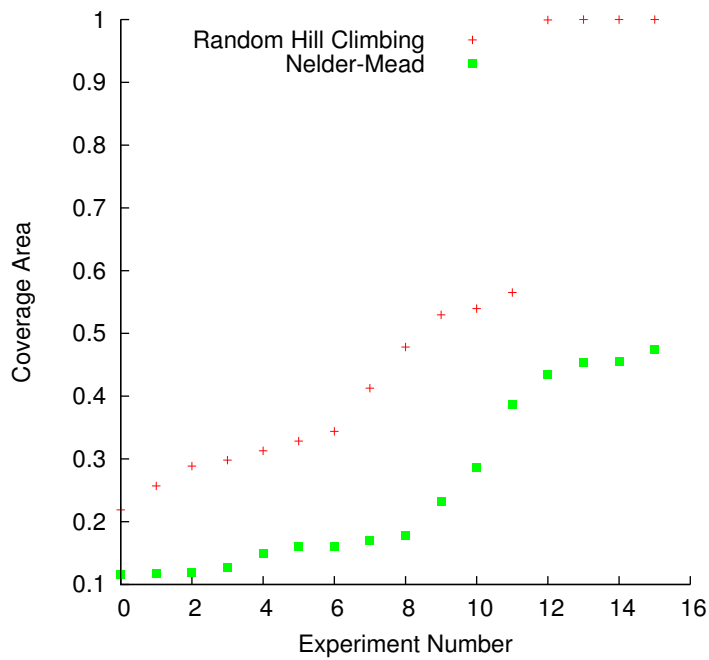}
\caption[Comparing optima for large networks]{Optimum coverage-area found by Random Hill Climbing and Nelder-Mead
for $\ge$ 500 transmitters is illustrated in this figure. Random Hill Climbing
performs better.}
\label{figure:big}
\end{figure}
Since these procedures run faster, we can run them with more transmitters. The
following combinations were run, yielding a total of $600$ experiments:
\begin{itemize}
\item Number of transmitters: 10, 20, 40, 80, 100
\item $p_{min} = 0$ and $p_{max} = 100$ for all transmitters
\item Values of $\alpha$: 2, 3, 4
\item Values of $N_0$: $10^{-3}$, $10^{-4}$, $10^{-5}$, $10^{-6}$
\item Size of Grid: 40 X 40
\end{itemize}
The experiments are run with $5$ fixed sets of transmitter, one set each for
10, 20, 40, 80 and 100 transmitters. Though both methods report nearly the same
maximum area in most cases, the Nelder-Mead
converges faster (Figure \ref{figure:ExecutionTime}).
\subsubsection{Random Hill Climbing vs Nelder-Mead: $\ge$ 500 transmitters}
Our final experiment has the following parameters:
\begin{itemize}
\item Number of transmitters: 500, 1000, 2000, 4000 
\item $p_{min} = 0$ and $p_{max} = 100$ for all transmitters
\item $\alpha$: 2, 3
\item $N_0$: $10^{-3}$, $10^{-5}$
\item Size of Grid: 40 X 40
\end{itemize}
We observed a better optimization result, as shown in Figure \ref{figure:big},
for Random Hill Climbing in each scenario above, though the time taken by both
methods was comparable. The summary of our observations is as follows:
\begin{enumerate}
\item Both Nelder-Mead and Random Hill Climbing report optimum power - this is
verified by comparison with Exhaustive Search. The total power reported at the
optimum point by all three methods is also comparable.
\item Nelder-Mead converges faster than Random Hill Climbing for up to 100
transmitters. However, Random Hill Climbing estimates a better maximum coverage
area for 500 or more transmitters.
\end{enumerate}
\subsection{A note on the asymptotic analysis of presented optimization 
methods}
A key analytic property of an optimization algorithm is its convergence rate. 
In combinatorial optimization (with discrete power values), the asymptotic 
convergence rate is measured by the algorithm's time efficiency. Exhaustive 
search converges in time $O(n k^n)$, where $k$ is the number of power 
levels and $n$ the number of transmitters. We conjecture that optimizing 
$\mathbb{C}$ with discrete power values is NP-hard.

In a general sense, Nelder-Mead and Random Hill Climbing do not necessarily 
converge. Practically, though, the algorithms presented here terminate in 
$O(nm)$ time, where $m$ is the maximum number of iterations. Our experiments 
have not uncovered any non-convergent case and show near-optimal results, and 
further analytical study is required for an asymptotic convergence analysis.

{\em Nelder-Mead}: Few objective functions yield convergence guarantees - for 
example, Singer et al.~\cite{singer} show that strictly convex functions in 
lower dimensions lead to convergence. For our case, further study 
of $\mathbb{C}$ is needed to determine whether it yields convergence for 
Nelder-Mead.

{\em Random Hill Climbing}: Johnson et al.~\cite{johnson} show that convergence 
can be guaranteed if the probability density function for transitioning between 
two feasible solutions is chosen such that paths always exist from local optima 
to global optima. We have used a uniform distribution function for our 
experiments, but further analysis is required to confirm whether this 
guarantees convergence for $\mathbb{C}$.
%
%

\chapter{Conclusions}
\label{chapter:Conc}
In this thesis we present new methods for the computation and optimization of 
an interference-constrained wireless coverage map. We present lower bound 
arguments and algorithms for this coverage problem.

We study this problem in both protocol and SINR models. For the protocol model, 
we exploit the underlying geometric structures by employing Voronoi Diagrams 
and their variants, Power Diagrams, to design efficient algorithms for both 
static and dynamic settings. We present lower bound arguments and optimal 
algorithms meeting these lower bounds for both static and dynamic settings. We 
extend traditional power diagrams and exploit them for building and updating 
coverage maps. The time complexity for the algorithm for the static setting is 
$O(n\log{n})$, and the time complexity for the algorithm to update one 
transmitter in the coverage map in the dynamic setting is $O(\log{n}+k)$ 
(expected) time per update, where $n$ is the total number of transmitters and 
$k$ is the number of neighboring transmitters affected by an update.

In the latter half of the thesis we focus on coverage in the SINR model. We 
relate the analytical difficulty in characterizing the coverage area 
geometrically to the non-convexity of the coverage region of a transmitter. We 
present a probabilistic sampling procedure for estimating the measure of the 
coverage area. We also present a Random Hill Climbing method for coverage area 
optimization by optimal assignment of transmit power to a given set of 
transmitters. The proposed method is flexible, in that the coverage-area 
estimation and optimization can accommodate any computationally tractable 
interference model - including the protocol model. By comparison with 
exhaustive search, we demonstrate that for small networks, the Random Hill 
Climbing method does not get stuck in a local maxima and yields optimum 
coverage, while using the optimal power allocation. By comparison with the 
Nelder-Mead method, we also show that for medium and large networks, the Random 
Hill Climbing method converges fast and yields optimum coverage.

\section{Related recent work and future directions}
The {\em Internet-of-Things} is here, with machine-to-machine (M2M) wireless 
communication networks expected to proliferate, with use in unregulated spectra. 
Coverage mapping and optimization are crucial in this context. Recent work in 
the related literature - for example, Zhang et al.~\cite{zhang2014m2m} 
discusses coverage management using M2M connected wireless devices.

Another associated area of related research is coverage optimization of 
self-organizing and co-operative networks. For example, Gueguen et 
al.~\cite{gueguen2013incentive} discuss co-operative transmission scheduling to 
extend coverage for a wireless network. Altman et al.~\cite{Altman} pose the 
mobile association problem in a game-theoretic framework, with transmitters 
competing and cooperating with one another to maximize network revenue.

More recent tools appearing in the research literature, like those by Kim et 
al.~\cite{kim2014cell}, Chen et al.~\cite{chen2013placement}, and Zhang et 
al.~\cite{zhang2014m2m}, discuss coverage management using measurements from 
wireless devices in the network. These, along with the use of modern data 
analytics tools - for example, Kim et al.~\cite{kim2014cell} and Kazakovtsev 
(\cite{kazakovtsev2013wireless}) who demonstrate methods for estimating and 
optimizing wireless coverage by analyzing radio ``fingerprints'' - appear to be 
candidate tools of the future.

Lastly, we conjecture that ``standard assumptions'', such as those described 
for dynamic addition and deletion of transmitters in the protocol model could 
be removed by tweaks to our algorithms and data structures. We also conjecture 
that a much simpler proof of convexity in geometric SINR (for equal powers) 
exists than is currently known in published literature.

%
\addcontentsline{toc}{chapter}{Appendix I}
\chapter*{Appendix I: A Note on Practical Coverage Planning Tools}
The thesis primarily discusses coverage from an algorithmic perspective. 
However, it was deemed interesting to study practical tools, if only to judge 
the utility of a possible implementation of our algorithms in actual software 
for use as a subroutine in an implemented software product.

Toward this end, we visited {\it AirTight Networks} in Pune (India) to briefly 
survey a practical coverage planning and design tool. {\it AirTight Networks} 
has a software tool called {\it SpectraGuard Planner} that automates placement 
of access points. We discussed the internals of the tool with the developers 
and attempted to identify areas where our ideas may be applied.

Our observations are listed below:
\begin{itemize}
\item The algorithm is based on ray-tracing with obstacles. The input to the
algorithm are the furniture and walls of the room in which coverage is
desired. The algorithm progresses by testing points on a grid for suitability of
placement of access point.
\item The directions in which the rays are traced are arbitrarily chosen. If a
ray hits an obstacle, the obstacle attenuates that ray, and the points on the
grid in coverage are reduced. 
\item Coverage estimates are extrapolated to capacity estimates.
\item The user is allowed to make changes to the access point location to test
whether coverage or capacity improves.
\item We feel that a more structured approach (possibly involving partitioning
using obstacle data) would speed up the algorithm. However, the team maintained 
that their current software suited their business model, and hence an 
additional speed up was not required.
\end{itemize}
Figures I.1 and I.2 show sample outputs of the software.
\begin{figure*}
\begin{center}
\includegraphics[width=5in]{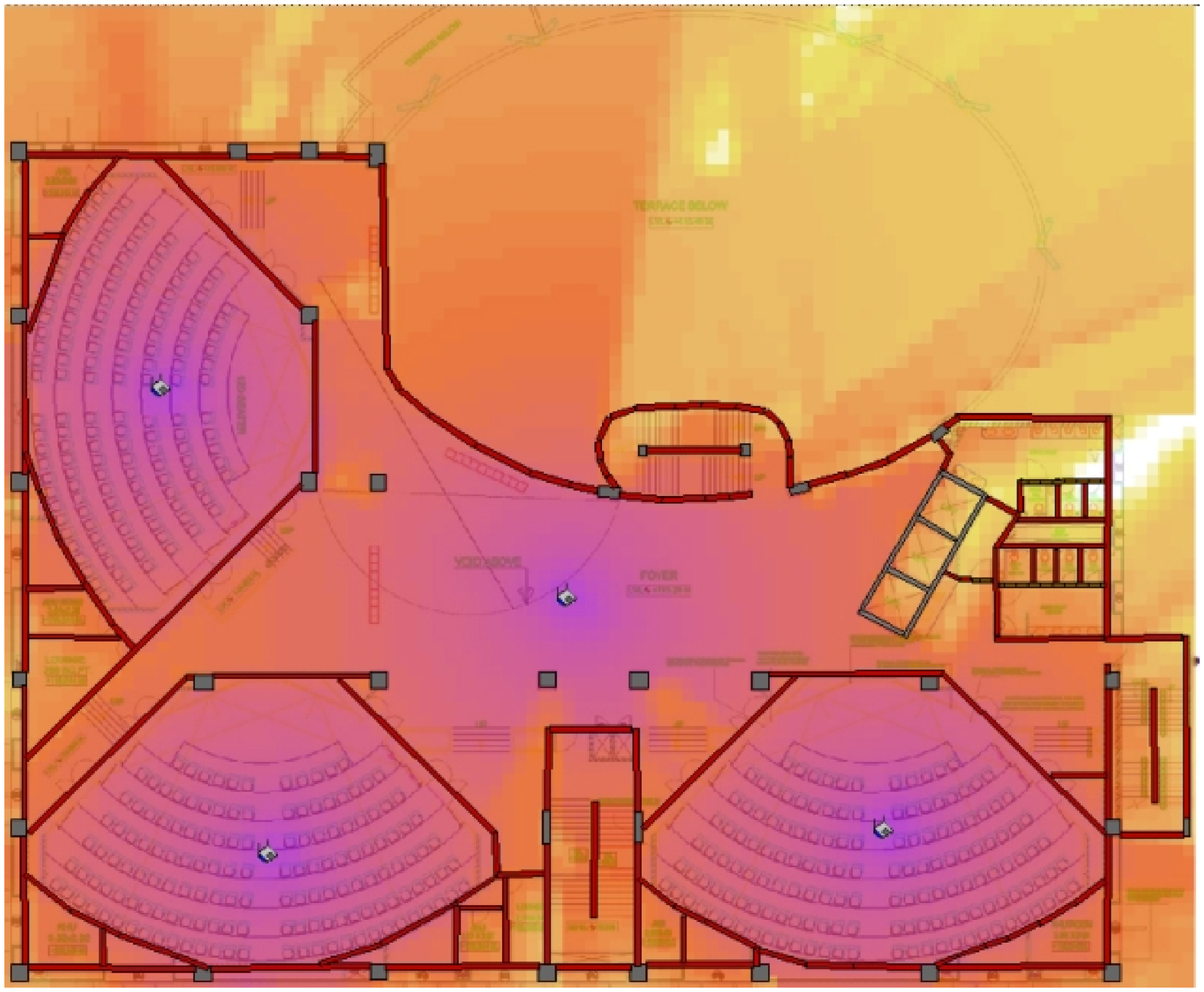}
\end{center}
I.1: WiFi Coverage Map for IITB Convention Center: Created by {\it AirTight 
Networks} using {\em SpectraGuard Planner}
\end{figure*}
\begin{figure*}
\begin{center}
\includegraphics[width=5in]{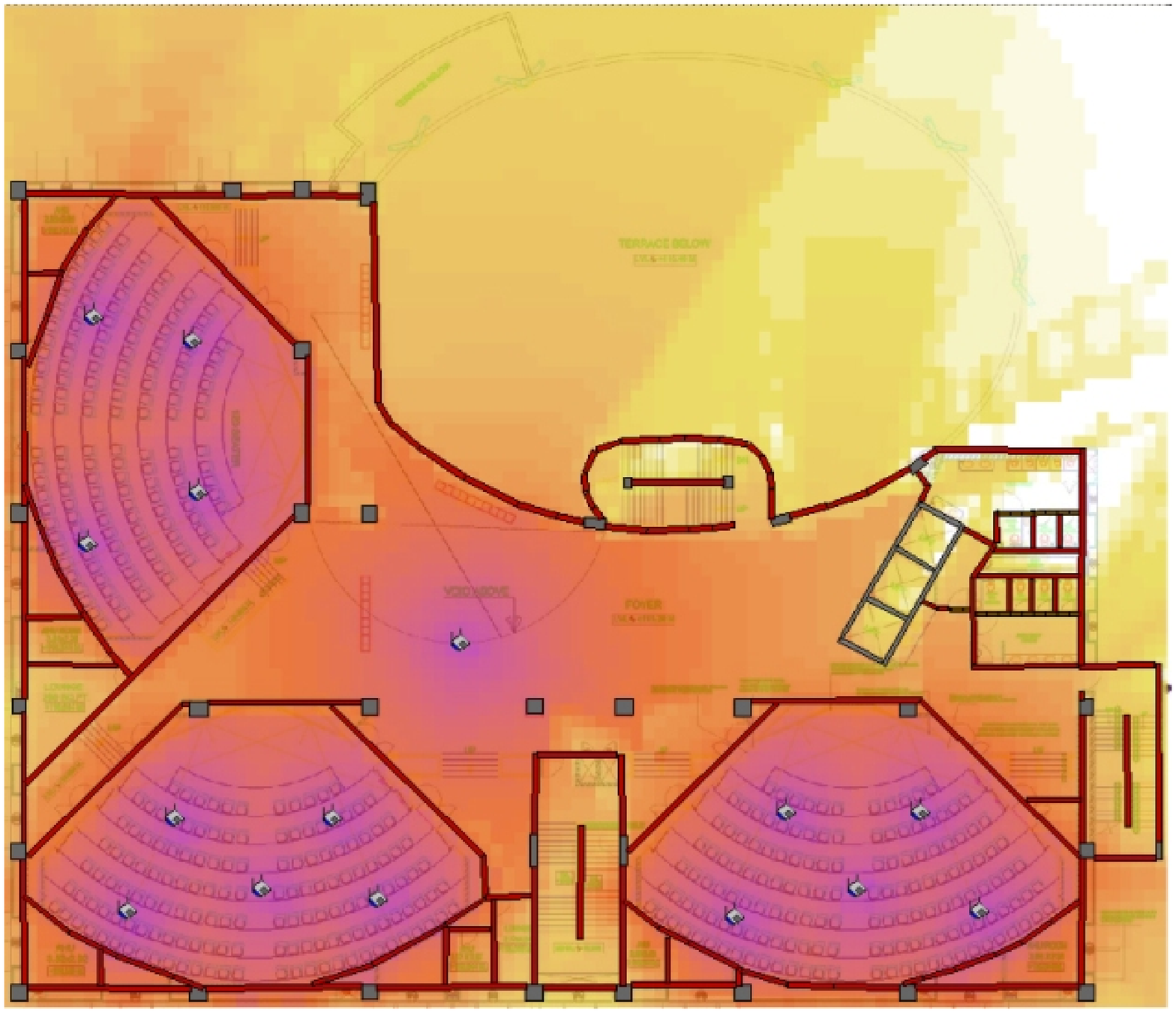}
\end{center}
I.2: WiFi Capacity Map for IITB Convention Center: Created by {\it AirTight 
Networks} using {\em SpectraGuard Planner}
\end{figure*}

%
\bibliographystyle{unsrt}
\bibliography{thesis-bibliography}
\addcontentsline{toc}{chapter}{Acknowledgments and Thanks}
\chapter*{Acknowledgments and Thanks}
This work was done during my affiliations with Reliance Communications
(2005 - 2009), TICET -- the Tata Teleservices IIT-Bombay Center for Excellence
in Telecommunication (2009 - 2010), and Flytxt (2010 to present). Dr.~Vinod
Vasudevan was co-advisor from 2005 to 2009.

To my advisors - Om, Animesh, Vinod: You have always insisted that I challenge 
my comfort zone and reach beyond. Thank you.

A note of thanks to my review committee, Prof.~Ajit Diwan and Prof.~Bhaskaran 
Raman, with special thanks to Prof.~Abhay Karandikar: Your constructive 
feedback, constant support and encouragement made this possible.

Sundar, Amic, Kover: Thank you for being there, on demand, for every demand.

A note of thanks also to Dr.~Pravin Bhagwat and his team at {\em AirTight 
Networks}, for their facilitation to discuss the {\em SpectraGuard Planner} 
tool.

Thanks to Jatin Sharma and Divyanshu Pandey for your programming support with 
MATLAB experiments on hole-filling.

Finally, to one and all associated with me for this project: I'll pay forward.
\end{document}